\documentclass[11pt]{article}
\usepackage{fullpage}
\usepackage{setspace}
\usepackage{my_sty}

\makeatletter
\def\blfootnote{\gdef\@thefnmark{}\@footnotetext}
\makeatother

\begin{document}

\title{Quantitative Group Testing and Pooled Data \\ in the Linear Regime with Sublinear Tests}

\author{Nelvin Tan, Pablo {Pascual Cobo}, and Ramji Venkataramanan}

%\allowdisplaybreaks
\date{}
\maketitle

\begin{abstract}
    In the \textit{pooled data} problem, the goal is to identify the categories associated with a large collection of items via a sequence of pooled tests. Each pooled test reveals the number of items in the pool belonging to each category. A prominent special case is quantitative group testing (QGT), which is the case of pooled data with two categories. We consider these problems in the non-adaptive and linear regime, where the fraction of items in each category is of constant order. We propose a scheme with a \textit{spatially coupled} Bernoulli test matrix and an efficient approximate message passing  (AMP) algorithm  for recovery. We rigorously characterize its asymptotic performance in both the noiseless and noisy settings, and prove that in the noiseless case, the AMP algorithm achieves \textit{almost-exact} recovery with a number of tests sublinear in the total number of items $p$. Although there exist other efficient schemes for noiseless QGT and pooled data that achieve recovery with order-optimal sample complexity ($\Theta(\frac{p}{\log p})$ tests),  there are no guarantees on their performance in the presence of noise, even at low noise-levels. In comparison, our scheme achieves recovery in the  noiseless case with a number of tests sublinear in $p$, and its performance degrades gracefully in the presence of noise.    Numerical simulations illustrate the benefits of the spatially coupled scheme at finite dimensions, showing that it outperforms i.i.d.~test designs as well as other recovery algorithms based on convex programming. 
\end{abstract}

\blfootnote{N.~Tan was supported by a Cambridge Trust Scholarship and the Harding Distinguished Postgraduate Scholars Programme Leverage Scheme. P.~Pascual~Cobo was supported by a Engineering and Physical Sciences Research Council Doctoral Training Award. This work was presented in part at the 2024 International Zurich Seminar on Information and Communication (IZS). 
%The authors are with the Department of Engineering, University of Cambridge, UK. 
Author emails: {\tt tcnt2@cam.ac.uk, pp423@cam.ac.uk, rv285@cam.ac.uk}. }

\section{Introduction}

Consider a large collection of items, each of which is either defective or non-defective. In \textit{group testing} \cite{Ald19}, the goal is to identify the defective set via pooled tests, where groups of items are tested together, with as few tests as possible. In the original Boolean group testing model, which has been studied  extensively \cite{Ald19,Cha14,Geb21,Oli22,Tan23b},  each test returns a positive outcome if it includes at least one defective item and a negative outcome otherwise.     %Pri23,Tan21,Tan20}
Its variant, the \textit{quantitative group testing} (QGT) model \cite{Bsh09}, is useful when tests are more informative: each test reveals the number of defective items in that pool. QGT is of interest in a range of modern applications, including  genomics \cite{Cao14}, multi-access communication \cite{Mar21}, and network traffic monitoring \cite{Wan15b}. 

A more general version of QGT,  where each item belongs to one of $L >2$ categories, is known as the \textit{pooled data} problem \cite{Wan16,Ning17,Sca17, Ala18,Ala19}. The goal is to identify the categories  via a sequence of pooled tests, where each pooled test reveals the number of items of each category within the pool. The pooled data problem is essentially one of inferring categorical information from histogram queries, which has applications in medical testing and learning with privacy constraints \cite{dwork2014algorithmic}.

In this paper, we consider \textit{non-adaptive} QGT and pooled data, where the tests are all designed in advance, making them amenable to being implemented in parallel. We also consider the \textit{linear} regime, which for QGT, means that the number of defective items is proportional to the total number of items. For pooled data, the linear regime implies that the proportion of items in each category is non-vanishing as the number of items increases, a realistic assumption in practical applications.

\subsection{Problem Setup} \label{sec:setup}

\paragraph{Quantitative group testing.} There are $p$ items, whose labels are denoted using the binary vector $\beta\in\{0,1\}^p$, where $1$ represents a defective item and $0$ a non-defective item. Items are allocated to tests using a binary design (or test) matrix $X\in\{0,1\}^{n\times p}$, where $n$ is the number of tests and $p$ is the number of items. The $i$th row $X_{i,:}$ determines the pooling design of the $i$th test, where $X_{ij}=1$ indicates that the $j$th item will be included in the $i$th test, and $X_{ij}=0$ indicates otherwise. Let $d$ be the number of defective items with $d<p$. We consider the linear regime, where the fraction of defective items $d/p$ converges to $\pi\in(0,1)$. Mathematically, we define the QGT model as
\begin{align}
    y_i&=\big(X_{i,:}\big)^\top\beta+\Psi_i
    \quad
    \text{for } i\in\{1,\dots,n\},
    \label{eq:QGT_model}
\end{align}
where $y_i$ is the $i$th element of $y\in\mb{R}^n$, $X_{i,:}$ is the $i$th row of $X\in \{0,1 \}^{n\times p}$ represented as a column vector, and $\Psi_i$ is the $i$th element of the additive noise $\Psi\in\mb{R}^n$. Under the noiseless setting (i.e., all entries of $\Psi$ are zero), the output $y_i$ is the number of defective items in the $i$th test. The goal of QGT is to recover $\beta$ with as few tests as possible.   

We will use the \textit{almost-exact} recovery criterion, which is achieved by an estimator $\tbeta$ if
\begin{align}
    \frac{1}{p}\sum_{j=1}^p\mathds{1}\big\{\tbeta_j\neq\beta_j\big\}\rightarrow0
    \text{ as $p\rightarrow\infty$.}
    \label{eq:almost_exact_recovery}
\end{align}
This is a weaker notion of recovery as compared to the exact recovery criterion \cite{Sca17} where we want the probability of error $\mb{P}\big[\tbeta\neq\beta\big]\rightarrow0$ as $p\rightarrow \infty$. We note that an almost-exact recovery criterion is meaningful in the linear regime, but not in the sublinear regime where the number of defectives $d=o(p)$, since setting $\tbeta$ to be the zero vector would trivially satisfy \eqref{eq:almost_exact_recovery}.

Almost-exact recovery is related to the \emph{approximate} recovery criterion
\cite{Sca17}  which requires that, with high probability, the number of errors is at most $q_{\max} \equiv q_{\max}(p)$. More precisely,  approximate recovery requires that
\begin{align}
    \mb{P}\left[  \sum_{j=1}^p\mathds{1}\big\{\tbeta_j\neq\beta_j\big\}  > q_{\max}\right] \to 0 \ \text{ as }  \ p \to \infty.
    \label{eq:ITlimit_approx}
\end{align} 
Almost-exact recovery is at least as strong a criterion as approximate recovery when $q_{\max}/p$ is either of constant order or is allowed to decay to zero arbitrarily slowly with $p$. Scarlett and Cevher derived information-theoretic lower bounds on the number of tests required for both exact and approximate recovery  \cite[Theorem 3]{Sca17}. These bounds show that for noiseless QGT, approximate recovery  with \emph{any} $q_{\max}$ such that $q_{\max}/p = o(1)$ requires essentially the same  number of tests as exact recovery ($q_{\max}=0$).  In their words, ``recovering all the labels is essentially as easy as recovering all but a vanishing fraction of the labels''.  This provides further justification for using the almost-exact recovery criterion, beyond the fact that it enables a precise asymptotic analysis of our scheme.

\paragraph{Pooled data.} The signal to be estimated is a matrix $B\in\{0,1\}^{p\times L}$, where each row is a one-hot vector. For example, $B_{j,:}=[0,1,0,\dots,0]$ represents the $j$th item belonging to category 2 (the position of one in $B_{j,:}$). We consider the linear regime, where the fraction of items in each category $l$ converges to $\pi_l \in (0,1)$, where $\sum_{l=1}^L \pi_l =1$. The model is
\begin{align}
    Y_{i,:}=B^\top X_{i,:}+\Psi_{i,:} 
    \quad
    \text{for $i\in\{1,\dots,n\}$},
    \label{eq:pooled_data_model}
\end{align}
where $Y_{i,:}$ is the $i$th row of $Y\in\mb{R}^{n\times L}$ represented as a column vector, and $\Psi_{i,:}$ is the $i$th row of the additive noise $\Psi\in\mb{R}^{n\times L}$ represented as a column vector. Under the noiseless setting (i.e., all entries of $\Psi$ are zero), the output of each test $Y_{i,:}$ tells us the number of items from each category present in the test, which can be viewed as a histogram. Similarly to QGT, an estimator $\tB$ achieves almost-exact recovery if
\begin{align*}
    \frac{1}{pL}\sum_{j=1}^p\sum_{l=1}^L\mathds{1}\big\{\tB_{jl}\neq B_{jl}\big\}
    \rightarrow0
    \text{ as }
    p\rightarrow\infty.
\end{align*}
The number of categories $L$ does not grow with $p$.

%\subsection{Information-theoretic Limits vs.~Efficient Algorithms}

\paragraph{Information-theoretic limits.}
For noiseless pooled data in the linear regime, the information-theoretic limit on the number of tests required  was established by Scarlett and Cevher \cite{Sca17}, closing the gap between previously derived upper and lower bounds \cite{Gre00,Wan16, Ala19}.  It was shown in \cite{Sca17} that the minimum number of tests needed for exact recovery is $n^*= \gamma^* \frac{p}{\log p}(1 + o(1))$, where 
\begin{align}
    \gamma^{*}
    &=\max_{r\in\{1,\dots,L-1\}}
    \frac{2[H(\pi)-H(\pi^{(r)})]}{L-r},
    \label{eq:gamma_star}
\end{align}
while $H(\pi)=-\sum_{l=1}^L\pi_l\log\pi_l$ is the Shannon entropy function, and $\pi^{(r)}=(\pi_1^{(r)},\dots,\pi_r^{(r)})$ is a vector whose first entry sums the largest $(L-r+1)$ entries of $\pi$, and whose remaining entries coincide with the remaining $(r-1)$ entries of $\pi$. Setting $L=2$ above gives the information-theoretic limit for noiseless QGT.   

Lower bounds on the number of tests for pooled data with approximate recovery  were also derived in \cite{Sca17}. As mentioned above, these bounds show that even when we allow for a vanishing fraction of errors, the number of tests required is still $\gamma^* \frac{p}{\log p}(1 + o(1))$ (noting that the lower order terms may differ from the exact recovery case). That is, the minimum number of tests required for almost-exact recovery is essentially the same as that for exact recovery.

For noisy pooled data where the entries of the noise matrix are independent with zero mean,  it was shown in \cite{Sca17} that we require $n=\Omega(p\log p)$ tests for exact recovery, in contrast with the sublinear $\Theta\big(\frac{p}{\log p}\big)$ required in the noiseless case. This lower bound  was extended to the approximate recovery in  \cite[Theorem 4]{Sca17} to show that the number of tests must be of order at least $p$ even if we allow a vanishing fraction of mistakes. Pooled data with adversarial noise was studied in \cite{Ning17}, and in this setting the number of tests required can be substantially higher than with random noise.

\paragraph{Efficient Algorithms.} For noiseless pooled data (including the special case of QGT), Wang et al.~\cite{Wan16} proposed a deterministic test design and a polynomial-time algorithm that achieves exact recovery with  $n=\Omega\big(\frac{p}{\log p}\big)$ tests, matching the optimal sample complexity above. However, both the test design and the recovery algorithm (based on Gaussian elimination) are tailored to the noiseless setting, and the guarantees do not extend to the noisy case. In this paper, we focus on \emph{random} test designs, which are more robust with respect to the items included in each test,  and on recovery algorithms whose performance degrades gracefully with the noise level.

For random i.i.d.~designs (where  $X_{ij}\stackrel{\iid}{\sim}\text{Bernoulli}(\alpha)$ for some $\alpha \in (0, 1)$), efficient recovery using approximate message passing (AMP) algorithms was  studied in \cite{Ala19, Tan23d}. Rigorous guarantees on the recovery performance of AMP, established in \cite{Tan23d}, imply that with the i.i.d.~Bernoulli design, the AMP algorithm needs 
$n=\Theta(p)$ tests for almost-exact recovery, even in the noiseless QGT/pooled data setting. The spatially coupled Bernoulli design we propose here improves on the i.i.d. one, enabling an AMP algorithm that requires only $n=o(p)$ tests in the noiseless case.
%Thus, for random designs there is an order-$\log p$ gap between the information-theoretic limit and the best known algorithms.

For the \emph{adaptive} setting, where each test can depend on the outcome of previous test, Bshouty \cite{Bsh09} proposed an efficient noiseless QGT scheme that identifies $d$ defectives out of $p$ items with  $\frac{2 d \log (p/d)}{\log d}(1 + o(1))$ tests. This is just over twice the information-theoretic lower bound of $\frac{d \log (p/d)}{\log d}$ for  any adaptive scheme  \cite{Bsh09}. We focus only on non-adaptive schemes in this paper.

\subsection{Approximate Message Passing and Spatial Coupling}

Approximate message passing (AMP) is a family of iterative algorithms that can be tailored to take advantage of structural information about the signals and the model, e.g., a known prior on the signal vector or on the proportion of observations that come from each signal. AMP algorithms were first proposed for the standard linear model \cite{Kab03,Bay11,Don09,Krz12}, but have since been applied to a range of statistical  problems, including estimation in generalized linear models and their variants \cite{Ran11,Ma19,Mai20,Mon21,Tan23c}, as well as low-rank matrix and tensor estimation \cite{Des14,Fle18,Kab16,Les17,Mont21,LiFanWei_Z2, rossetti2023approximate}. In all these settings, under suitable model assumptions, the performance of AMP in the high-dimensional limit is characterized by a succinct deterministic recursion called \emph{state evolution}.  The literature on AMP is vast, and we refer the interested reader to \cite{Fen21} for a survey.

In this paper, we will use a \emph{spatially coupled} design matrix $X$ and a suitable AMP algorithm for recovery. The spatially coupled matrix has a block-wise structure, with blocks along a band-diagonal having i.i.d.~Bernoulli entries and the remaining blocks being all zeros (see Figure \ref{fig:SC_test_matrix}).  Our scheme is inspired by a line of work on compressed sensing with spatially coupled designs \cite{Kud10,Krz12,Tak15,Don13}. For a noiseless linear model defined via a spatially coupled Gaussian sensing matrix, Donoho et al.~\cite{Don13} showed that AMP recovers the signal  with high probability when the sampling ratio $\delta= n/p$ exceeds the R{\'e}nyi information dimension of the signal prior. The R{\'e}nyi information  dimension is zero for priors supported on a finite set, which implies that AMP can recover the signal with $n=o(p)$ measurements. Recently, it was also shown that using a spatially coupled sensing matrix in a  generalized linear model  allows AMP to achieve the Bayes-optimal error (corresponding to an i.i.d.~Gaussian  matrix) \cite{Cob23}. 

\subsection{Main Contributions}

In Section \ref{sec:SC_QGT}, we describe the  spatially coupled random test design and an  AMP algorithm (SC-AMP) for signal recovery.  In Theorem \ref{thm:SC_AMP_QGT} we give a precise characterization of the performance of SC-AMP in the asymptotic regime where the number of tests $n$ grows proportionally with the number of items $p$ (with $n/p \to \delta$, a constant). Using this characterization, we bound the MSE of the algorithm in the low noise regime (Theorem \ref{thm:achievability}) and show that for noiseless QGT, the SC-AMP algorithm achieves almost-exact recovery with probability one, for \emph{any} constant sampling ratio $\delta >0$ (Corollary \ref{cor:SC_AMP_almost_exact}). This implies that it achieves almost-exact recovery with a sublinear number of tests, i.e., $n =o(p)$. In Section \ref{sec:SCAMP_pooled_data}, we generalize the SC-AMP algorithm to the pooled data setting and again establish almost-exact recovery for any constant $\delta >0$ (Theorem \ref{thm:achievability_pooled_data}).

 To our knowledge, ours is the first efficient scheme for QGT and pooled data in the linear regime that both requires a sublinear number of tests in the noiseless case \emph{and} is provably robust to noise.   Indeed, Theorem \ref{thm:achievability} and Corollary \ref{cor:SC_AMP_almost_exact} give bounds which quantify the noise sensitivity of the asymptotic MSE (and the fraction of errors).  More generally, for any fixed values of noise variance $\sigma^2$ and sampling ratio $\delta$,  our results (Theorem \ref{thm:SC_AMP_QGT} for QGT and  \eqref{eq:matrix_SC_AMP_convergence} for pooled data) give a precise asymptotic characterization of the fraction of errors made by the SC-AMP algorithm. 
 
 Numerical simulations show that the spatially coupled scheme outperforms the i.i.d.~Bernoulli test design with AMP, as well as recovery algorithms based on convex programming.

At the heart of our theoretical guarantees is a rigorous analysis of an AMP algorithm for a  generalized linear model (GLM) with a generic spatially coupled design matrix. (The matrix consists of blocks of independent entries  drawn from an arbitrary zero-mean distribution satisfying certain moment conditions.) GLMs include many important nonlinear estimation problems such as phase retrieval and logistic regression. Theorem \ref{thm:SC_GAMP} shows that the AMP algorithm and its performance characterization originally developed for GLMs with spatially coupled  Gaussian designs \cite{Cob23} remain valid for a much broader class of designs.

\paragraph{Key technical ideas.} Although the QGT model \eqref{eq:QGT_model} is an instance of a linear model, an important constraint is that the test design matrix $X$ can only contain binary entries. Therefore, we cannot apply the analysis from \cite{Don13}, which assumes a spatially coupled Gaussian design matrix. To prove theoretical guarantees for our scheme, we reduce the SC-AMP algorithm to an abstract AMP iteration defined for any generalized white noise matrix. A state evolution result for this abstract AMP iteration was established by Wang et al.~in \cite{Wan22}, using which we obtain a rigorous asymptotic characterization of  the SC-AMP algorithm (Theorem \ref{thm:SC_AMP_QGT}). To establish conditions for almost-exact recovery, we then need to analyze the fixed points of the SC-AMP state evolution. We do this in Theorem \ref{thm:asym_MSE} via the potential function method \cite{Yed14},  a powerful tool for characterizing the fixed points of coupled recursions. This characterization then yields the noise robustness and exact-recovery results  (Theorem \ref{thm:achievability} and Corollary \ref{cor:SC_AMP_almost_exact}).

\subsection{Other Related Work}

\paragraph{Spatial coupling.} Spatial coupling was introduced in coding theory as a means to construct LDPC codes that achieve capacity with an efficient belief propagation decoder \cite{felstrom1999time,kudekar2013spatially}. Spatial coupling has since been applied in many estimation problems to  improve on the performance of `regular' (or i.i.d.) designs. For Boolean group testing in the sublinear regime (the number of defectives is $p^\theta$ for $\theta \in (0,1)$), spatially coupled test designs enable efficient recovery with the asymptotically optimal number of tests, in both the noiseless \cite{Coj20}  and noisy settings \cite{Coj24}. 
%(In noiseless Boolean group testing, each test is positive if there is at least one defective item in the test, and zero otherwise.) 
For QGT,  Mashauri et al.~\cite{Mas23,Mas24} investigated efficient schemes based on spatially coupled LDPC codes, and showed that they outperform previous constructions based on generalized LDPC codes \cite{Kar19a, Kar19b}.

\paragraph{Sublinear category regime.} A few recent works have studied pooled data in  the sublinear category regime, where one category is  dominant with $p-o(p)$ items, and the remaining $(L-1)$ categories have $d=o(p)$ items. (In contrast, we consider the linear category regime, where the proportion of items in each category is $\Theta(1)$, i.e., $\pi_l=\Theta(1)$ for $l\in[L]$.) 
For the sublinear category regime, the information-theoretic lower bound for exact recovery is $n=\Omega(d)$ tests \cite{Hah22a,Geb23}. An efficient algorithm proposed in \cite{Hah22a} achieves the lower bound when $d=\Theta(p^\kappa)$,  for a constant $\kappa\in(0,1)$. A lower complexity algorithm for the special case of QGT with $d=\Theta(p^\kappa)$ was recently proposed in \cite{Sol24a}. For QGT in the sublinear regime, a number of recent works have proposed algorithms based on ideas from coding theory \cite{Mas23} and thresholding \cite{Hah22b}, which require $\Omega(d\log p)$ tests for exact recovery.  Noisy versions of QGT were recently studied in \cite{Li21} and \cite{Hah23}, and QGT  in the adaptive setting has been studied in \cite{Bsh09,Sol24b}.

\section{Preliminaries} \label{sec:prelim}

\paragraph{Notation.} We let $[n]:=\{1,\dots,n\}$ and $[n:m]:= \{n,n+1,\dots,m\}$, for $n<m$.  All vectors (including those corresponding to rows of matrices) are assumed to be column vectors unless otherwise stated. For  $a,b\in\mb{R}^n$, $\langle a,b\rangle=a^\top b\in\mb{R}$ is the inner product, $a\odot b\in\mb{R}^n$ is the entry-wise product, and  $\langle a\rangle=\frac{1}{n}\sum_{i=1}^n a_i$ denotes the empirical average of the entries of $a$. 

Matrices are denoted by upper case letters, and given a matrix $A$, we write $A_{i,:}$ for its $i$th row and $A_{:,j}$ for its $j$th column.
The operator norm is denoted by  $\|A\|_{\text{op}}$.
For $r\in[1,\infty)$ and a vector $a=(a_1,\dots,a_n)\in\mb{R}^n$, we write $\|a\|_r$ for the $\ell_r$-norm, so that $\|a\|_r=\big(\sum_{i=1}^n|a_i|^r\big)^{1/r}$.  We use $1_p$ to denote the vector of $p$ ones, $0_p$ for the vector of $p$ zeros, and $I_p$ for the $p\times p$ identity matrix. We use $\mathds{1}\{ \cdot \}$ to denote the indicator function, and $\E[\cdot ]$ for expectation. Given random variables $U,V$, we write $U \stackrel{d}{=} V$ to denote equality in distribution. We write $\partial_if(\cdot)$ to denote the partial derivative of $f$ with respect to (w.r.t.) the $i$th argument. Throughout, the function $\log(\cdot)$ has base $e$, and we use Bachmann-Landau asymptotic notation (i.e., $O$, $o$, $\Omega$, $\omega$, $\Theta$).

\paragraph{Almost-sure and Wasserstein convergence.} Let $\{A^n\}$ be a sequence of random elements taking values in a Euclidean space $E$. We say that $A^n$ converges almost surely to a deterministic limit $a\in E$, and write $A^n\stackrel{a.s.}{\rightarrow}a$, if $\mb{P}[\lim_{n\rightarrow\infty}A^n=a]=1$.
 
For a vector $a\in\mb{R}^n$ and a random variable $A\in\mb{R}$, we write $a\stackrel{W_r}{\rightarrow}A$ as $n\rightarrow\infty$, for the Wasserstein-$r$ convergence of the empirical distribution of the entries of $a$ to the law of $A$. More generally, for vectors $a^1,\dots,a^k\in\mb{R}^n$ and a random vector $(A^1,\dots,A^k)\in\mb{R}^k$, we write
\begin{align*}
    a^1,\dots,a^k\stackrel{W_r}{\rightarrow}(A^1,\dots,A^k)
    \text{ as $n\rightarrow\infty$},
\end{align*}
for the Wasserstein-$r$ convergence of the empirical distribution of rows of $(a^1,\dots,a^k)\in\mb{R}^{n\times k}$ to the joint law of $(A^1,\dots,A^k)$. This means that, for any continuous function $\phi:\mb{R}^k\rightarrow\mb{R}$ and input vector $(a_i^1,\dots,a_i^k)\in\mb{R}^k$ satisfying the \textit{polynomial growth} condition \cite{Wan22}
\begin{align}
    |\phi(a_i^1,\dots,a_i^k)|
    \leq C\big(1+\|(a_i^1,\dots,a_i^k)\|_2^r\big),
    \label{eq:poly_growth_cond}
\end{align}
for a constant $C>0$, we have
\begin{align}
    \frac{1}{n}\sum_{i=1}^n\phi(a^1_i,\dots,a^k_i)
    \rightarrow\E\big[\phi(A^1,\dots,A^k)\big]
    \text{ as $n\rightarrow\infty$}.
    \label{eq:sum_to_exp}
\end{align}
We write
\begin{align*}
    a\stackrel{W}{\rightarrow} A,\quad
    (a_1,\dots,a_k)\stackrel{W}{\rightarrow}(A_1,\dots,A_k)
    \text{ as $n\rightarrow\infty$},
\end{align*}
to mean that the above Wasserstein-$r$ convergences hold for every order $r\geq 1$.

\paragraph{Model assumptions for QGT.} The signal $\beta\in\{0,1\}^p$ is independent of the design matrix. As $n, p\rightarrow\infty$, we have $n/p\rightarrow\delta>0$ (for a constant $\delta$), and the empirical distribution of the entries of the signal  converges in Wasserstein distance to well-defined limits. More precisely,  \begin{align}
    \beta\stackrel{W}{\rightarrow}\bar{\beta} \quad  \text{ where } \bar{\beta}\sim\text{Bernoulli}(\pi).
    \label{eq:bar_beta}
\end{align}  
We note that the entries of $\beta$ are not assumed to be independent or identically distributed.

\section{Spatially Coupled Design for Quantitative Group Testing} \label{sec:SC_QGT}

\begin{figure}[t]
    \centering
    \includegraphics[width=0.45\textwidth]{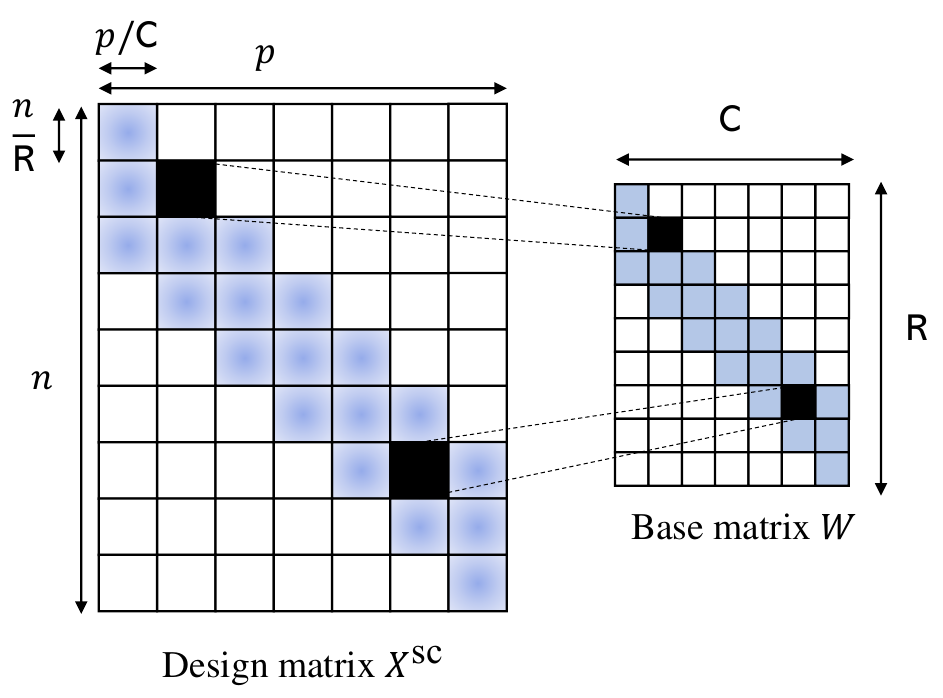}
    \caption{The entries of $\Xsc$ are independent with $X_{ij}^\text{sc}\sim\text{Bernoulli}(\alpha W_{\sfr(i),\sfc(j)})$. Here $W$ is an $(\omega,\Lambda)$ base matrix with $\omega=3$ and $\Lambda=7$ (see Definition \ref{def:omega_Lambda_base_matrix}). The white parts of $X^\text{sc}$ and $W$ correspond to zeros.}
    \label{fig:SC_test_matrix}
\end{figure}

The spatially coupled (SC) design matrix consists of independent Bernoulli entries whose parameters are specified by a \emph{base matrix} $W$ of dimension $\sfR\times\sfC$.
The SC design matrix is obtained by replacing each entry of the base matrix $W_{\sfr\sfc}$ by an $\frac{n}{\sfR}\times\frac{p}{\sfC}$ matrix with entries drawn independently from $\text{Bernoulli}(\alpha W_{\sfr\sfc})$, where $\alpha, \alpha W_{\sfr\sfc} \, \in(0,1)$. 
An example of a SC design matrix is shown in Figure \ref{fig:SC_test_matrix}.  
In this paper, we will use  the following base matrix.

\begin{definition} \label{def:omega_Lambda_base_matrix}
An $(\omega,\Lambda)$ base matrix $W$ is described by two parameters: the coupling width $\omega\geq 1$ and the coupling length $\Lambda\geq2\omega-1$. The matrix has $\sfR=\Lambda+\omega-1$ rows and $\sfC=\Lambda$ columns, with each entry indexed by $(\sfr,\sfc)$, for $\sfr\in[\sfR]$ and $\sfc\in[\sfC]$. For $\alpha \leq 0.5$, the entries are given by
\begin{align}
    W_{\sfr\sfc}
    &=\begin{cases}
        \frac{1}{2\alpha}\left(1-\sqrt{1-\frac{4\alpha(1-\alpha)}{\omega}}\right) & \text{ if  } \sfc\leq \sfr\leq \sfc+\omega-1, \\
        0 &\text{ otherwise}.
    \end{cases}
    \label{eq:W_rc_def}
\end{align}
For $\alpha> 0.5$, the non-zero entries (when $\sfc\leq \sfr\leq \sfc+\omega-1$) are given by $ \frac{1}{2\alpha}\left(1+\sqrt{1-\frac{4\alpha(1-\alpha)}{\omega}}\right)$.
\end{definition}

Figure \ref{fig:SC_test_matrix} shows an $(\omega,\Lambda)$ base matrix with $\omega=3$ and $\Lambda=7$.  The \textit{spatially coupled (SC)} design matrix, denoted by $\Xsc$, has independent entries generated as follows:
    \begin{align}
        \Xsc_{ij}
        \stackrel{\text{indep.}}{\sim}
        \text{Bernoulli}\big(\alpha W_{\sfr(i)\sfc(j)}\big), \quad i\in[n], \, j \in [p].
        \label{eq:SC_matrix}
    \end{align}
    for some fixed constant $\alpha\in(0,1)$.  Here the operators $\sfr(\cdot):[n]\rightarrow [\sfR ]$ and $\sfc(\cdot):[p]\rightarrow [\sfC ]$ map a particular row or column index to its corresponding \emph{row block} or \emph{column block} index in $W$. The band-diagonal structure of the $(\omega,\Lambda)$ base matrix here is similar to the ones used for SC sparse regression codes \cite{Rush21} and for SC generalized linear models \cite{Cob23}, but the values of the non-zero entries are different. Here the base matrix specifies the Bernoulli parameters  for each block of the design, whereas in \cite{Rush21,Cob23} it specifies the  variances for the Gaussian entries in each block.
    
The \textit{i.i.d.}~design matrix, denoted by $X^{\text{iid}}$, has each entry sampled i.i.d.~$\sim\text{Bernoulli}(\alpha)$, for some fixed constant $\alpha\in(0,1)$. 
Note that the i.i.d.~matrix is a special case of the SC matrix, with $\sfR=\sfC=1$ and $W=1$. A key difference between the i.i.d.~design and the SC design (with an $(\omega, \Lambda
)$ base matrix) is that the latter includes many fewer items in each test. Indeed, with the i.i.d.~design each item is included in a test with probability $\alpha$, whereas in the SC design, each test includes items from at most $\omega$ adjacent column blocks, each with $p/\sfC$ items (see Figure \ref{fig:SC_test_matrix}). In the SC design, since the tests corresponding to the first and last row blocks involve the fewest items, the corresponding entries of $\beta$ are easier to recover than the others. A good estimate for these entries helps the algorithm recover the entries in the adjacent blocks, creating a decoding wave that propagates from the ends towards the center.

With the i.i.d.~design, the number of defectives per test has expected value $\alpha \pi p$ and standard deviation $\sqrt{\alpha \pi (1 - \alpha \pi) p} $. Similarly, for the spatially coupled design, it can be verified that the number of defective items per test has mean that is linear in $p$ and standard deviation of order $\sqrt{p}$. Since the fluctuation around the mean contains the useful information in each test, and because AMP requires a design matrix with zero-mean entries, we recenter and rescale the data before applying the AMP algorithm. We now describe this preprocessing of the data, which was also done in \cite{Tan23d} for i.i.d.~designs.

The \textit{rescaled i.i.d.}~matrix, denoted by $\tX^{\text{iid}}$, is defined as
    \begin{align}
        \tX^{\text{iid}}
        &=\frac{X^{\text{iid}}-\alpha1_n1_p^\top}{\sqrt{n\alpha(1-\alpha)}}.
        \label{eq:rescaled_QGT_iid_matrix}
    \end{align}
   We note that $\tX^{\text{iid}}$ has independent entries with   $\E[\tX^{\text{iid}}_{ij}]=0$ and $\Var[\tX^{\text{iid}}_{ij}]=1/n$. 
   
   The \textit{rescaled spatially coupled (SC)} matrix $\tXsc$  is defined as follows.  For $i \in [n], j \in [p]$, using the shorthand $\sfr \equiv \sfr(i), \, \sfc \equiv \sfc(j)$, its entries are given by 
    \begin{align}
        \tX_{ij}^{\text{sc}}
        &=\frac{X_{ij}^{\text{sc}}-\alpha W_{\sfr\sfc}}{\sqrt{n\alpha(1-\alpha)/\sfR}}
        =\begin{cases}
            \frac{1-\alpha W_{\sfr\sfc}}{\sqrt{n\alpha(1-\alpha)/\sfR}}
            &\text{ with probability $\alpha W_{\sfr\sfc}$,} \\
            \frac{-\alpha W_{\sfr\sfc}}{\sqrt{n\alpha(1-\alpha)/\sfR}}
            &\text{ with probability $1-\alpha W_{\sfr\sfc}$.}
        \end{cases} 
        \label{eq:rescaled_QGT_SC_matrix}
    \end{align}
    It is straightforward to verify that $\E[\tXsc_{ij}]=0$, $\Var[\tXsc_{ij}]=\frac{\sfR W_{\sfr\sfc}(1-\alpha W_{\sfr\sfc})}{n(1-\alpha)}$. In particular, for the $(\omega, \Lambda)$ base matrix in Definition \ref{def:omega_Lambda_base_matrix}, we have \begin{align}
    \Var[\tXsc_{ij}]
    =\begin{cases}
        \frac{\sfR}{n\omega}
        &\text{ if $\sfc\leq\sfr\leq\sfc+\omega-1$,} \\
        0 &\text{ otherwise.}
    \end{cases}
    \label{eq:omega_Lambda_var}
\end{align}

\paragraph{Rewriting QGT.} 
The AMP algorithm and its analysis require the design matrix to have independent zero-mean entries, so we recenter and rescale the QGT model in \eqref{eq:QGT_model} to express it in terms of the rescaled design  ($\tXsc$ or $\tX^{\text{iid}}$). For the SC design, using \eqref{eq:rescaled_QGT_SC_matrix}, we have for $i\in[n]$:
\begin{align*}
    &y_i
    =\sum_{j=1}^pX_{ij}^{\text{sc}}\beta_j+\Psi_i
    =\sum_{j=1}^p\left(\alpha W_{\sfr(i)\sfc(j)}+\sqrt{\frac{n\alpha(1-\alpha)}{\sfR}}\tXsc_{ij}\right)\beta_j+\Psi_i \\
    \implies\,
    &y_i-\sum_{j=1}^p\alpha W_{\sfr(i)\sfc(j)}\beta_j
    =\sqrt{\frac{n\alpha(1-\alpha)}{\sfR}}\sum_{j=1}^p\tXsc_{ij}\beta_j+\Psi_i \\
    \implies\,
    &\frac{1}{\sqrt{n\alpha(1-\alpha)/\sfR}}\left(y_i-\alpha\left(W_{\sfr(i)1}\sum_{j\in\mathcal{J}_1}\beta_j+\dots+W_{\sfr(i)\sfC}\sum_{j\in\mathcal{J}_\sfC}\beta_j\right)\right)\\
    &=(\tXsc_{i,:})^\top\beta+\frac{\Psi_i}{\sqrt{n\alpha(1-\alpha)/\sfR}},
\end{align*}
where $\mathcal{J}_\sfc=\big[(\sfc-1)p/\sfC+1:\sfc p/\sfC\big]$ for $\sfc\in[\sfC]$. Denoting the left-hand side  above by
\begin{align}
    \ty_i
    &:=\frac{1}{\sqrt{n\alpha(1-\alpha)/\sfR}}\left(y_i-\alpha\left(W_{\sfr(i)1}\sum_{j\in\mathcal{J}_1}\beta_j+\dots+W_{\sfr(i)\sfC}\sum_{j\in\mathcal{J}_\sfC}\beta_j\right)\right),
    \label{eq:y_tilde}
\end{align}
gives us the rescaled QGT model:
\begin{align}
    \ty_i
    =(\tXsc_{i,:})^\top\beta+\tPsi_i,    \quad \text{ with } \quad  \tPsi_i := \frac{\Psi_i}{\sqrt{n\alpha(1-\alpha)/\sfR}} \, , \quad i \in [n].
    \label{eq:rescaled_SC_QGT}
\end{align}

The term $\sum_{j\in\mathcal{J}_\sfc}\beta_j$ in \eqref{eq:y_tilde} is the number of defective items in the sub-vector of $\beta$ indexed by $\mathcal{J}_{\sfc}$, for $\sfc  \in [\sfC ]$.
In the noiseless setting, the terms $\sum_{j\in\mathcal{J}_c}\beta_j$ can be obtained with an extra $\sfC$ tests, where we only include items from $\mathcal{J}_c$ in the $\sfc$th test, for $\sfc\in[\sfC]$. The extra $\sfC=O(1)$ tests do not affect  our results since the limiting sampling ratio $\lim_{n\rightarrow\infty} n/p = \delta$ remains the same.
%change our main result regarding the SC-AMP algorithm requiring $o(p)$ tests for almost-exact recovery under the noiseless setting. 
Since $\frac{1}{p}\sum_{j\in\mathcal{J}_c}\beta_j\stackrel{a.s.}{\rightarrow}\frac{\pi}{\sfC}$ for $p\rightarrow\infty$ via the strong law of large numbers, we can also estimate $\sum_{j\in\mathcal{J}_\sfc}\beta_j$ using $\frac{p}{\sfC}\pi$. (However the error in this estimate would be of order $\sqrt{p/\sfC}$.)

For an i.i.d.~design $X^{\text{iid}}$, we  similarly recenter and rescale the QGT model to express it in terms of the rescaled i.i.d.~matrix $\tX^{\text{iid}}$ in \eqref{eq:rescaled_QGT_iid_matrix}. For $i \in [n]$, we have
\begin{align}
    y_i=\big(X_{i,:}^{\text{iid}}\big)^\top\beta+\Psi_i 
    \implies 
    & \ty_i
    =\big(\tX_{i,:}^{\text{iid}}\big)^\top\beta + \tPsi_i, \nonumber  \\ 
    & \text{ where $\ty_i:=\frac{y_i-\alpha d}{\sqrt{n\alpha(1-\alpha)}}, \quad \tPsi_i= \frac{\Psi_i}{\sqrt{n\alpha(1-\alpha)}}$,}
    \label{eq:rescaled_iid_QGT}
\end{align}
where $d$ is the number of defective items.

\paragraph{Choice of $\alpha$.}\label{sec:items-test}  For  a spatially coupled design $\Xsc$ constructed from an $(\omega, \Lambda)$ base matrix, recall from  \eqref{eq:rescaled_QGT_SC_matrix} that the rescaled matrix $\tXsc$ has independent zero-mean  Bernoulli entries with variances given by \eqref{eq:omega_Lambda_var}. 
Notice that the distribution of $\tXsc$ does not depend on $\alpha$. Hence, in the rescaled model, $\alpha$ only affects the variance of the noise $\tPsi$, which is minimized when $\alpha=0.5$. We therefore use $\alpha =0.5$ for all our experiments.

Taking $\alpha=0.5$, it is useful to compare the SC design with the i.i.d.~Bernoulli($0.5$) design. (For the i.i.d.~design, taking the Bernoulli parameter to be $0.5$ is optimal with respect to both the information-theoretic limits \cite{Sca17} and efficient  recovery via AMP \cite{Tan23d}.)  Taking $\alpha=0.5$ in Definition \ref{def:omega_Lambda_base_matrix}, we have that the non-zero blocks of the SC design are drawn independently from a Bernoulli distribution with parameter $\frac{1}{2}\big( 1- \sqrt{1 -1/\omega} \big)$. From the structure of $\Xsc$ (see Figure \ref{fig:SC_test_matrix}), it follows that the expected number of items included in  each test is at most $ \frac{p}{\sfC} \frac{\omega}{2}\big( 1- \sqrt{1 -1/\omega} \big)$. For large $\omega$, this is approximately $\frac{p}{4\sfC}$ items per test. In contrast, for the Bernoulli($0.5$) design, the expected number of items per test is $\frac{p}{2}$. Thus, for large $\sfC$, the expected number of items per test is much smaller for the SC design  than the i.i.d.~design.

\paragraph{Noise scaling assumption.}  \label{par:noise_scaling}
In the rescaled QGT models  \eqref{eq:rescaled_SC_QGT} and \eqref{eq:rescaled_iid_QGT}, it can be verified  that the terms  $\big(\tX_{i,:}^{\text{sc}}\big)^\top\beta$  and  $\big(\tX_{i,:}^{\text{iid}}\big)^\top\beta$ have zero mean and variance of constant order, for each $i \in [n]$. Therefore, for the rescaled model to be meaningful,  the noise $\tPsi_i$ should also have a mean and variance of constant order. This is guaranteed by the following assumption.
The empirical distribution of  rescaled noise vector $\tPsi$ in \eqref{eq:rescaled_SC_QGT} converges to a well defined limit. More precisely, there exists, $\bar{\Psi}\sim P_{\bar{\Psi}}$ with $\E\big[\bar{\Psi}^2\big]=:\sigma^2<\infty$, such that 
$\tPsi\stackrel{W}{\rightarrow}\bar{\Psi}$ 
as $n \to \infty$. We emphasize that the base matrix parameter $\sfR$ is fixed as $n \to  \infty$. A similar distributional assumption holds for the rescaled noise vector with the i.i.d.~design in \eqref{eq:rescaled_iid_QGT}.

\subsection{SC-AMP Algorithm} \label{sec:SC-AMP_algo}

Consider the rescaled SC model in \eqref{eq:rescaled_SC_QGT}. 
Given $(\tXsc, \ty)$, the SC-AMP algorithm  iteratively produces signal estimates,  denoted by $\hbeta^k \in \reals^p$, for $k\geq1$. For iteration $k\geq 0$, the algorithm computes:
\begin{align}
\begin{split}
    \tTheta^k&=\ty \, - \, \tXsc \hbeta^k \, + \, b^k\odot Q^{k-1}\odot\tTheta^{k-1}
    \,  \in\mb{R}^n, \\
    \beta^{k+1}&=(\tXsc)^\top(Q^k\odot\tTheta^k) \, - \, c^k\odot\hbeta^k
    \,  \in\mb{R}^p, 
    \qquad
    \hbeta^{k+1}=f_{k+1}(\beta^{k+1},\mathcal{C})
    \in \, \mb{R}^p,
\end{split}
\label{eq:SC_AMP}
\end{align}
where $\odot$ denotes the Hadamard (entry-wise) product, the function $f_k: \reals \times [\sfC] \to \reals$ acts row-wise on the inputs $(\beta^k,\mathcal{C})$, and 
\begin{align}
\label{eq:mC_def}
    \mathcal{C}
:=(\underbrace{1,\dots,1}_{\text{$p/\sfC$ entries}},\underbrace{2,\dots,2}_{{\text{$p/\sfC$ entries}}},\dots,\underbrace{\sfC,\dots,\sfC}_{{\text{$p/\sfC$ entries}}})^\top
\in\mb{R}^p.
\end{align} 
 The algorithm is initialized with $\hbeta^0=\E[\bar{\beta}]1_p$ and $\tTheta^0=\ty-\tXsc\hbeta^0$, where from \eqref{eq:bar_beta} we recall that $\bar{\beta} \sim \text{Bernoulli}(\pi)$.
 To define $b^k \in \reals$, we use the  partitions $[p]=\bigcup_{\sfc=1}^\sfC\mathcal{J}_\sfc$ and $[n]=\bigcup_{\sfr=1}^\sfR\mathcal{I}_\sfr$, where
\begin{align}
\begin{split}
    \mathcal{J}_\sfc
    &=\left\{(\sfc-1)\frac{p}{\sfC}+1,\dots,\sfc\frac{p}{\sfC}\right\}
    \text{ for $\sfc\in[\sfC]$}, 
    \qquad
    \mathcal{I}_\sfr
    =\left\{(\sfr-1)\frac{n}{\sfR}+1,\dots,\sfr\frac{n}{\sfR}\right\}
    \text{ for $\sfr\in[\sfR]$}.
\end{split}
\label{eq:J_c_I_r_def}
\end{align}
Then, letting
\begin{align}
    \tW_{\sfr\sfc}:=\frac{W_{\sfr\sfc}(1-\alpha W_{\sfr\sfc})}{1-\alpha} = \begin{cases}
        1/\omega &\text{ if $\sfc\leq\sfr\leq\sfc+\omega-1$,} \\
        0 &\text{otherwise},
    \end{cases}
    \label{eq:W_tilde}
\end{align}
the entries of $b^k\in\mb{R}^n$ are 
\begin{align*}
    b_i^{k}
    =\sum_{\sfc=1}^\sfC \frac{\tW_{\sfr\sfc}}{n/\sfR}\sum_{j\in\mathcal{J}_c} f'_{k}(\beta_j^{k},\sfc), \quad i \in [n],
\end{align*}
where $f'_k$ is the derivative with respect to the first argument.
The second equality in \eqref{eq:W_tilde} follows from Definition \ref{def:omega_Lambda_base_matrix}.

Next, we define the function $f_k:\reals \times [\sfC]$ and the vectors $Q^k \in \reals^n, c^k \in \reals^p$ in \eqref{eq:SC_AMP}. These are defined in terms of block-wise scalar state evolution parameters, denoted by $\chi^k_\sfc$, for $\sfc \in [\sfC]$, which are recursively computed, as described in \eqref{eq:SC_SE_param1}-\eqref{eq:SEchi_init} below. We let
\begin{equation}
\label{eq:optimal_fk}
     f_{k}(s, \sfc)
    =\E\big[\bar{\beta}\,\big|\,(\chi_{\sfc}^{k})^2\bar{\beta}+\chi_{\sfc}^{k}G=s\big], \quad \text{ for } \, \sfc \in [\sfC].
\end{equation}
 The vectors $Q^k \in \reals^n$ and $c^k \in \reals^p$   have a block-wise structure and are defined as follows. For $i \in [n]$, $j \in [p]$, recalling that 
 $\sfr(i)$ and $\sfc(j)$ denote the respective row-block and column-block indices, we have:
\begin{align}
\begin{split}
    Q_i^k= \left(
    \sigma^2+\frac{1}{\deltain}\sum_{\sfc=1}^\sfC\tW_{\sfr(i) \sfc} \, \E\Big[\big(\bar{\beta}-f_k( \, (\chi_{\sfc}^{k})^2\bar{\beta}+\chi_{\sfc}^{k}G, \, \sfc \, ) \big)^2\Big]
    \right)^{-1},
    \qquad
    c_j^k=-(\chi_{\sfc(j)}^{k+1})^2. 
\end{split}
\label{eq:Qc_def}
\end{align}
 We note that the time complexity of each iteration in \eqref{eq:SC_AMP} is $O(np)$.

%where $\mathcal{J}_\sfc$ is defined as  follows. 

We now give some high-level intuition about the SC-AMP algorithm and its state evolution characterization. For this, we need some additional notation to handle the block-wise structure of the iterates. For $\sfc \in [\sfC]$ and $\sfr \in [\sfR]$, we define $\beta_\sfc:=\beta_{\mathcal{J}_\sfc}\in\mb{R}^{p/\sfC}$ and $\hbeta_\sfc^k:=\hbeta_{\mathcal{J}_\sfc}^k\in\mb{R}^{p/\sfC}$ to be the $\sfc$th blocks of $\beta\in\mb{R}^p$ and $\hbeta^k\in\mb{R}^p$ respectively, and $\Theta_\sfr:=\Theta_{\mathcal{I}_\sfr}\in\mb{R}^{n/\sfR}$ and $\tTheta_\sfr^k:=\tTheta_{\mathcal{I}_\sfr}^k\in\mb{R}^{n/\sfR}$ to be the $\sfr$th blocks of $\Theta:=\tXsc\beta\in\mb{R}^n$ and $\tTheta^k\in\mb{R}^n$ respectively. Similar notation simplifications will be used for other vectors where $\sfc$ and $\sfr$ will replace $\mathcal{J}_\sfc$ and $\mathcal{I}_\sfr$ in the subscripts of the vectors.

\paragraph{State evolution.} The `memory' terms $b^k\odot Q^{k-1}\odot \tTheta^{k-1}$ and $-c^k\odot \hbeta^k$ in \eqref{eq:SC_AMP} debias the iterates $\tTheta^k$ and $\beta^{k+1}$, ensuring that their empirical distributions are accurately captured by state evolution in the high-dimensional limit. These iterates have a block-wise distributional structure.   Recall from the model assumptions that the empirical distribution of the signal $\beta$ converges to the law of $\bar{\beta} \sim \text{Bernoulli}(\pi)$.  Theorem \ref{thm:SC_AMP_QGT} below shows that, for each $k\geq 1$ and $\sfc \in [\sfC]$, the empirical distribution of  $\beta_\sfc^k$ converges to the distribution of $(\chi_\sfc^k)^2\bar{\beta}+\chi_\sfc^kG$, where $G\sim\normal(0,1)$ is a standard Gaussian independent of $\bar{\beta}$, and the deterministic parameter $\chi_\sfc^k\in\mb{R}$ is defined below via the state evolution recursion.  Under this distributional assumption for $\beta^k$, the denoising function $f_k$ in \eqref{eq:optimal_fk} produces a Bayes-optimal (MMSE)  estimate of $\beta$ from $\beta^k$. 
%Indeed, we choose $f_k$ in this manner, and for $k \ge 1$, define:

 State evolution iteratively computes the  parameter  $\chi_\sfc^k \in\mb{R}$ as follows, for $k \ge 1$.  Letting 
\begin{align}
\label{eq:deltain_def}
    \deltain
    =\lim_{n \to \infty} \, \frac{n/\sfR}{p/\sfC}
    =\frac{\sfC}{\sfR}\delta  \, ,
\end{align}
given $\chi_\sfc^k$ for $\sfc \in [\sfC]$, we compute $\chi_\sfc^{k+1}$ as:
\begin{align}
     (\chi^{k+1}_{\sfc})^2
    =\sum_{\sfr=1}^\sfR\tW_{\sfr\sfc}
    \left(
    \sigma^2+\frac{1}{\deltain}\sum_{\sfc'=1}^\sfC\tW_{\sfr\sfc'}\E\bigg[\big(\bar{\beta}-f_k( \, (\chi_{\sfc'}^{k})^2\bar{\beta}+\chi_{\sfc'}^{k}G, \, \sfc' \, ) \big)^2\bigg]
    \right)^{-1}, 
        \label{eq:SC_SE_param1}
\end{align}
where $G \sim \normal(0,1)$ is independent of $\bar{\beta}$, and $\sigma^2$ is the variance defined in the noise scaling assumption on p.\pageref{par:noise_scaling}.
The recursion is initialized with 
\begin{equation}
    (\chi_{\sfc}^{1})^2=\sum_{\sfr=1}^\sfR\tW_{\sfr\sfc}\Big(  \sigma^2+\frac{1}{\deltain}\sum_{\sfc'=1}^\sfC\tW_{\sfr\sfc'} \text{Var}(\bar{\beta}) \Big)^{-1}, \quad \sfc \in [\sfC]. \label{eq:SEchi_init}
\end{equation}
For the reader's convenience, the important notation for the rest of this section is summarized in Table \ref{tab:notation}.   

%% Table summarizing notation
\begin{table}
\begin{center}
 \caption{Summary of key notation for SC-AMP}
 \label{tab:notation}
\begin{tabular}{ | l | l |} 
 \hline
$\beta  \in  \{0,1\}^p$ & Signal vector  \\ 
 $\beta^k  \in \reals^p$ & SC-AMP estimate before denoising \\ 
 $\hbeta^k  \in \reals^p$ & SC-AMP estimate after denoising  \\ 
 $\tbeta^k  \in \{0,1\}^p$ & Quantized version of $\hbeta^k$  (see \eqref{eq:quantized_est_def})  \\ 
 $\bar{\beta}   \in \{0,1 \}$ & Bernoulli($\pi$) random variable \\ 
 $\chi^k_\sfc  \in \reals^+$ & State evolution SNR parameter \\ 
  $ \delta = \lim_{n \to \infty} \frac{n}{p} $   & Overall sampling ratio \\
 $ \deltain = \frac{\sfC}{\sfR} \delta $ & Inner sampling ratio \\
  \hline
\end{tabular}
 \end{center}
\end{table}
The SC-AMP algorithm in \eqref{eq:SC_AMP} and its state evolution are equivalent to those proposed for a spatially coupled \emph{Gaussian} design \cite{Don13}. The key difference is that our algorithm uses a rescaled spatially coupled Bernoulli design. The theorem below shows that the state evolution guarantees remain valid for this setting. For an i.i.d.~design (where $\sfR=\sfC=1$), SC-AMP reduces to the standard AMP algorithm \cite{Bay11} for an  i.i.d.~Gaussian design.

\begin{theorem}[State evolution result for SC-AMP] \label{thm:SC_AMP_QGT}
Consider the QGT model \eqref{eq:QGT_model} with a spatially coupled design defined via the $(\omega, \Lambda)$ base matrix in Definition \ref{def:omega_Lambda_base_matrix}. 
Let the model assumptions in Section \ref{sec:prelim} and the noise scaling assumption (p.~\pageref{par:noise_scaling}) be satisfied. Then, for the SC-AMP algorithm in \eqref{eq:SC_AMP}, run on the recentered and rescaled QGT model \eqref{eq:rescaled_SC_QGT} with the denoising functions $f_k$ in \eqref{eq:optimal_fk}, we have  the following convergence guarantee. For each $k\geq0$ and $\sfc\in[\sfC]$: 
\begin{align}
\begin{split}
    \big(\beta_\sfc, \, \beta_\sfc^{k+1}\big)
    &\stackrel{W_2}{\rightarrow}
    \big(\bar{\beta}, \, (\chi_{\sfc}^{k+1})^2\bar{\beta}+\chi_{\sfc}^{k+1}G\big)
\end{split}
\label{eq:SC_AMP_asym_dist_QGT}
\end{align}
almost surely as $n,p\rightarrow\infty$ with $n/p\rightarrow\delta$. Here $\bar{\beta} \sim \text{Bernoulli}(\pi)$ and $G \sim \normal(0,1)$ are independent, and the parameter $\chi^{k+1}_{\sfc}$ is defined in \eqref{eq:SC_SE_param1}.
\end{theorem}
    The theorem is proved in Section \ref{sec:SC_AMP_QGT_proof}, where the SC-AMP algorithm is shown to be a special case of an AMP algorithm  for  a \emph{generalized} linear model with a generic spatially coupled design. We  prove Theorem \ref{thm:SC_AMP_QGT} by establishing a state evolution result for this general AMP algorithm (Theorem \ref{thm:SC_GAMP}).

\paragraph{Performance measures.}  Theorem \ref{thm:SC_AMP_QGT} allows us to compute the limiting values of performance measures such as the mean-squared error (MSE) and the normalized squared correlation, via the convergence property in \eqref{eq:sum_to_exp}. The MSE of the AMP estimate $\hbeta^k$ satisfies the following almost surely, for $k \ge 1$: 
\begin{align}
    \lim_{p\rightarrow\infty}
    \frac{1}{p}\|\beta-\hbeta^k\|_2^2
    &=
    \frac{1}{\sfC}\sum_{\sfc=1}^\sfC\E\left[\big(\bar{\beta}-f_k((\chi_{\sfc}^k)^2\bar{\beta}+\chi_{\sfc}^kG,\sfc)\big)^2\right],
    \label{eq:mean_sq_error}
\end{align}
while the normalized squared correlation of the AMP estimate $\hbeta^k$ satisfies:
\begin{align}
    \lim_{p\rightarrow\infty}
    \frac{\langle\hbeta^k,\beta\rangle}{\|\hbeta^k\|_2^2\cdot\|\beta\|_2^2}
    &=
    \frac{(\frac{1}{\sfC}\sum_{\sfc=1}^\sfC\E[f_k((\chi_{\sfc}^k)^2\bar{\beta}+\chi_{\sfc}^kG,\sfc)\cdot\bar{\beta}])^2}{(\frac{1}{\sfC}\sum_{\sfc=1}^\sfC\E[f_k((\chi_{\sfc}^k)^2\bar{\beta}+\chi_{\sfc}^kG,\sfc)^2])\cdot(\E[\bar{\beta}^2])}, \quad k\geq1.
    \label{eq:norm_sq_corr}
\end{align}

%\subsection{False Positive Rate and False Negative Rate} 

We can also obtain formulas for the limiting values of the false positive rate (FPR) and false negative rate (FNR).
The choice of $f_k$ in \eqref{eq:optimal_fk} outputs a vector in $\mb{R}^p$, but we can obtain an estimate in $\{0,1\}^p$ by thresholding the AMP iterate $\hbeta^K$ in the final iteration $K$ to output a hard decision. For some chosen constant $\zeta$, let us define the hard decision to be
\begin{align}
    \mathds{1}\left\{\hbeta_j^K>\zeta\right\}
    =\mathds{1}\left\{f_K\big(\beta_j^K,\sfc \big)>\zeta\right\},
    \quad \text{for } j \in \mathcal{J}_\sfc,\, \sfc \in [\sfC],
    \label{eq:s_c_k_def}
\end{align}
where the indicator function is applied component-wise to $\hbeta_\sfc^k$. That is, we declare large entries of $\hbeta^K$ to be one (i.e., defective) and small entries of $\hbeta^K$ to be zero (i.e., non-defective). Based on the above function, let us denote the estimated defective set as $\widehat{\mathcal{S}}=\left\{j:\hbeta_j^K>\zeta\right\}$.

The false positive rate (FPR) and the false negative rate (FNR) are defined as:
\begin{align}
    \text{FPR}&=
    \frac{\sum_{j=1}^p\mathds{1}\{\beta_j=0\cap j\in\widehat{\mathcal{S}}\}}{p-\sum_{j=1}^p\beta_j} \quad 
    \text{and} \quad 
    \text{FNR}=
    \frac{\sum_{j=1}^p\mathds{1}\{\beta_j=1\cap j\notin\widehat{\mathcal{S}}\}}{\sum_{j=1}^p\beta_j}.
    \label{eq:FPR_and_FNR}
\end{align}

\begin{corollary} \label{cor:FPR_FNR}
Under the same assumptions as for Theorem \ref{thm:SC_AMP_QGT}, with a threshold $\zeta\in[0,1]$ for the final iteration $K$, as $p\rightarrow\infty$, we have
\begin{align}
    \textup{FPR}\stackrel{a.s.}{\rightarrow}
    \frac{1}{\sfC}\sum_{\sfc=1}^\sfC\mb{P}\big[f_K\big(\chi_{\sfc}^KG,\sfc\big)>\zeta\big] \ 
    \text{ and } \
    \textup{FNR}\stackrel{a.s.}{\rightarrow}
    \frac{1}{\sfC}\sum_{\sfc=1}^{\sfC}\mb{P}\big[f_K\big(\big(\chi_{\sfc}^K\big)^2+\chi_{\sfc}^KG,\sfc\big)\leq\zeta\big].
\end{align}
\end{corollary}
The result follows from Theorem \ref{thm:SC_AMP_QGT} by applying the convergence property to suitable indicator functions. The proof uses the same steps as the analogous result for i.i.d.~Bernoulli designs \cite[Corollary 5.2]{Tan23d}, and is omitted.
% \begin{proof}
%     The proof is presented in Appendix \ref{sec:FPR_FNR_proof}.
% \end{proof}

\subsection{Almost-Exact Recovery}

Given $\hbeta^k$, the SC-AMP estimate after $k$ iterations, let us define the quantized estimate to be
\begin{align}
    \tbeta^k_j
    =
    \begin{cases}
        1 &\text{ if } \hbeta^k_j>0.5, \\
        0 &\text{ otherwise.}
    \end{cases}
    \label{eq:quantized_est_def}
\end{align}
 Then, recalling the almost-exact recovery criterion in \eqref{eq:almost_exact_recovery}, the SC-AMP algorithm 
 achieves almost-exact recovery if 
%\begin{align}
$\lim_{k\rightarrow\infty}\lim_{p\rightarrow\infty}\frac{1}{p}\sum_{j=1}^p
    \mathds{1}\big\{\tbeta_j^k\neq\beta_j\big\}=0$.
   % \label{eq:almost-exact-recovery-SCAMP}
%\end{align}

In this section, we show  that the SC-AMP algorithm can attain almost-exact recovery with $n=o(p)$ tests, by proving that it 
attains almost-exact recovery 
for any $\delta>0$ (recall that $\delta=\lim_{n\rightarrow\infty}\frac{n}{p}$). To this end, we introduce the \textit{potential function} to analyze the asymptotic MSE of SC-AMP as $k \to \infty$. Potential functions are widely used  to characterize the limiting  MMSE and mutual information in high-dimensional estimation problems (see, e.g., \cite{reeves2019thereplica,barbier2020mutual}).   Here we will use it only to characterize the asymptotic MSE of the AMP algorithm, both with and without spatial coupling (see Theorem \ref{thm:asym_MSE}).

\begin{definition}
\label{def:potential function}
For $b\in[0,\Var(\bar{\beta})]$, $\delta > 0$, the scalar potential function for the rescaled QGT model is defined as
\begin{align}
    U(b;\delta)
    :=-\delta\Big(1-\frac{\sigma^2}{(b/\delta)+\sigma^2}\Big)+\delta\log\Big(1+\frac{b}{\delta\sigma^2}\Big)+2I\bigg(\bar{\beta};\sqrt{\frac{1}{(b/\delta)+\sigma^2}}\bar{\beta}+G\bigg).
    \label{eq:U_function}
\end{align}
Here the mutual information $I(\cdot;\cdot)$ is computed with $\bar{\beta}\sim P_{\bar{\beta}}$ independent of $G\sim\normal(0,1)$, and  $\sigma^2= \E\big[\bar{\Psi}^2\big]$ is the second moment of the rescaled noise (see 
p.~\pageref{par:noise_scaling}).
\end{definition}

Figure \ref{fig:potential_fn} plots the potential function for various values of $\delta$. The next theorem characterizes the limiting MSE of the AMP algorithm via the minimizers and stationary points of the potential function. 
For clarity, we refer to the AMP algorithm under the rescaled i.i.d.~QGT model in \eqref{eq:rescaled_iid_QGT} as the iid-AMP algorithm, and the AMP algorithm under the rescaled spatially coupled QGT model in \eqref{eq:rescaled_SC_QGT} as the SC-AMP algorithm. The SC-AMP algorithm in \eqref{eq:SC_AMP} reduces to  iid-AMP with the trivial base matrix ($\sfR=\sfC=1$). 
\begin{theorem}[MSE of SC-AMP and iid-AMP] \label{thm:asym_MSE}
Consider the QGT model \eqref{eq:QGT_model}, and let the model assumptions in Section \ref{sec:prelim} and the noise scaling assumption (p.~\pageref{par:noise_scaling}) be satisfied.
    \begin{enumerate}
        \item Consider a spatially coupled design defined via an $(\omega,\Lambda)$ base matrix (Definition \ref{def:omega_Lambda_base_matrix}). For any $\gamma>0$, there exist $\omega_0<\infty$ and $k_0<\infty$ such that for all $\omega>\omega_0$ and $k>k_0$, the asymptotic MSE of the SC-AMP algorithm almost surely satisfies:
        \begin{align}
            \lim_{p\rightarrow\infty}\frac{1}{p}\|\beta-\hbeta^k\|_2^2
            <  \left(\max\Bigg\{\argmin_{b\in[0,\Var(\bar{\beta})]}U(b;\deltain)\Bigg\}+\gamma\right)\frac{\Lambda+\omega}{\Lambda}.
            \label{eq:asym_MSE_SC_AMP}
        \end{align}
        \item With an i.i.d.~design (i.e., $1\times1$ base matrix with $W_{11}=1$), the asymptotic MSE of the iid-AMP algorithm almost surely satisfies:
        \begin{align}
            \lim_{k\rightarrow\infty}\lim_{p\rightarrow\infty}
            \frac{1}{p}\|\beta-\hbeta^k\|_2^2
            =\max\Big\{b\in[0,\Var(\bar{\beta})]:\partial_1U(b;\delta)=0{}\Big\},
            \label{eq:asym_MSE_AMP}
        \end{align}
        where $\partial_1$ denotes the partial derivative w.r.t.~the first argument.
    \end{enumerate}
\end{theorem}

\begin{figure}[t]
    \centering
    \input{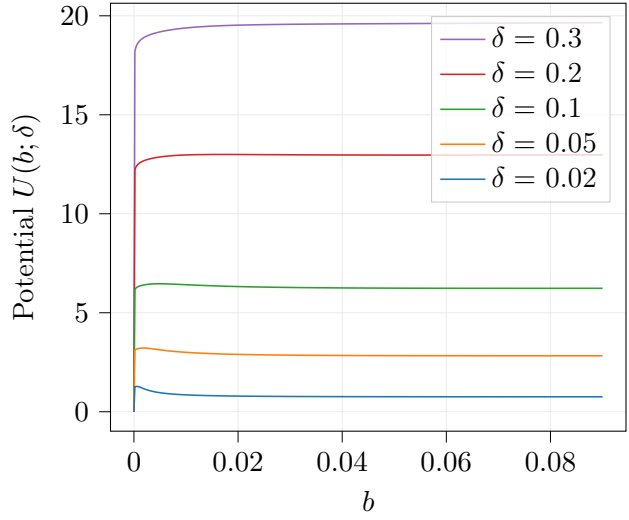}
    \caption{$U(b;\delta)$ vs.~$b$, for different $\delta$.  $\pi=0.1$ and $\sigma = 1\times 10^{-30}$.}
    \label{fig:potential_fn}
\end{figure}

The proof is given in Section \ref{sec:proof:thm:asym_MSE}. Part 1 of the theorem 
says that, for sufficiently large base matrix parameters (with $\omega \ll \Lambda$), the MSE of the SC-AMP algorithm is bounded by the largest minimizer of the potential function. 
(The $\max\{ \cdot \}$ indicates that if there are multiple minimizers, the largest one is chosen.) Part 2 of the theorem says that the MSE of iid-AMP algorithm is given by the largest stationary point of the potential function. In Figure \ref{fig:potential_fn}, we observe that for $\delta=0.05$ and $\delta=0.02$, the  unique  minimizer is $b=0$,  but the largest stationary point is strictly larger than zero. This implies that the limiting MSE of SC-AMP algorithm is 0, but that of iid-AMP algorithm is strictly larger than 0. The next lemma quantifies this observation, showing that for any $\delta >0$, the largest minimizer of the potential function tends to zero as the noise variance $\sigma \to 0$.

\begin{lemma} \label{lem:QGT_pot_fun_vanish}
    Consider the scalar potential function $U(b;\delta)$  in \eqref{eq:U_function} with $\delta>0$. For any $\Delta\in(0,\delta)$, there exists $\sigma_0(\Delta)>0$ such that for all $\sigma<\sigma_0(\Delta)$, we have the rate
    \begin{align}
        \max\bigg\{\argmin_{b\in(0,\Var(\bar{\beta})]}U(b;\delta)\bigg\}
         < \frac{7}{2} \delta  \left(\sigma^{2-2\Delta/\delta}\right).
        \label{eq:QGT_pot_fun_rate}
    \end{align}
\end{lemma}

The proof is given in Appendix \ref{proof:lem:QGT_pot_fun_vanish}. Using Lemma \ref{lem:QGT_pot_fun_vanish} in Part 1 of Theorem \ref{thm:asym_MSE} yields the following bound on the MSE of the SC-AMP algorithm in the low noise regime.
\begin{theorem}[MSE of SC-AMP in the low-noise regime] \label{thm:achievability}
    Consider the setup of part 1 of Theorem \ref{thm:asym_MSE}, for any $\delta >0$. Then for any $\epsilon >0$ and $\Delta \in (0, \delta/(1+\epsilon))$, there exists $\sigma_0(\Delta) >0$ such that the following holds for any noise variance $\sigma < \sigma_0(\Delta)$.  There exist finite $\omega_0$ and $k_0$ such that for all $\omega>\omega_0$, $k>k_0$ and $\Lambda$ sufficiently large, the asymptotic MSE of the SC-AMP algorithm almost surely satisfies:
    \begin{align*}
        \lim_{p\rightarrow\infty}\frac{1}{p}\|\beta-\hbeta^k\|_2^2
        <  4 \delta  \sigma^{2-2\Delta(1+\epsilon)/\delta} + \epsilon.
    \end{align*}
\end{theorem}

\begin{proof}
   Let $\epsilon, \gamma >0$.  Recalling that $\deltain  =\delta\frac{\Lambda}{\Lambda+\omega-1}$, for sufficiently large $\Lambda/\omega$ we  have  $ \deltain < \frac{\delta}{1+\epsilon}$. Then, from  Part 1 of Theorem \ref{thm:asym_MSE}, for $k > k_0$, $\omega > \omega_0$ and sufficiently large $\Lambda/\omega$, we have 
    \begin{align}
        \lim_{p\rightarrow\infty}\frac{1}{p}\|\beta-\hbeta^k\|_2^2
     < \left( 
        \max\left\{\argmin_{b\in[0,\Var(\bar{\beta})]}U(b;\deltain)\right\} + \gamma \right)(1+\epsilon).
        \label{eq:asym_MSE_bound}
    \end{align}
Taking $\gamma$ to be small enough and using Lemma \ref{lem:QGT_pot_fun_vanish} in \eqref{eq:asym_MSE_bound} gives the result.
\end{proof}

\begin{corollary}[SC-AMP achieves almost exact recovery for any $\delta >0$] Consider the setting and assumptions of Theorem \ref{thm:achievability}. Let $\tbeta^k$ be the quantized estimate produced from the SC-AMP iterate $\hbeta^k$, according to \eqref{eq:quantized_est_def}. Then, almost surely we have:
\begin{align}
    \lim_{p\rightarrow\infty} 
    \frac{1}{p}\sum_{j=1}^p\mathds{1}\big\{\tbeta_j^k\neq\beta_j\big\} \le 4 \lim_{p\rightarrow\infty}\frac{1}{p}\|\beta-\hbeta^k\|_2^2 < 4   \left(4\delta \sigma^{2-2\Delta(1+\epsilon)/\delta} + \epsilon\right).
    \label{eq:Hamming_error_vs_MSE}
\end{align}
In particular, for noiseless QGT we have $\lim_{k \to \infty}  \lim_{p\rightarrow\infty} 
    \frac{1}{p}\sum_{j=1}^p\mathds{1}\big\{\tbeta_j^k\neq\beta_j\big\} =0$ almost surely for any $\delta >0$.
    \label{cor:SC_AMP_almost_exact}
\end{corollary}
\begin{proof}
    From the definition of the quantized estimate in \eqref{eq:quantized_est_def}, it follows that for  $j \in [p]$, we have $ | \beta_j-\hbeta_j^k| \ge (0.5)\mathds{1}\big\{\tbeta_j^k\neq\beta_j\big\}$, which implies $\mathds{1}\big\{\tbeta_j^k\neq\beta_j\big\} \le 4 | \beta_j-\hbeta_j^k|^2 $. This gives the first inequality in \eqref{eq:Hamming_error_vs_MSE}. The second inequality follows from Theorem \ref{thm:achievability}. The result for the noiseless case follows by setting $\sigma=0$, and taking a sequence $(\epsilon_k)$ such that $\epsilon_k \to 0$ as $k \to \infty$.  We note that $\lim_{p\rightarrow\infty} 
    \frac{1}{p}\sum_{j=1}^p\mathds{1}\big\{\tbeta_j^k\neq\beta_j\big\}$ exists for each $k$ by the state evolution result in Theorem \ref{thm:SC_AMP_QGT}.
\end{proof}

The guarantees in Corollary \ref{cor:SC_AMP_almost_exact} are analogous to those in \cite[Theorem 1.7 and Corollary 1.8]{Don13} for a linear model with a spatially coupled \emph{Gaussian} design matrix. Specifically, \cite[Theorem 1.7]{Don13}  shows that  when the  sampling ratio $\delta$ is larger than the R{\'e}nyi information dimension of the signal prior, the MSE of the SC-AMP algorithm satisfies
$\lim_{k \to \infty}\lim_{p\rightarrow\infty}\frac{1}{p}\|\beta-\hbeta^k\|_2^2  < C \sigma^2$, for sufficiently small noise variance $\sigma^2$. Here the constant $C$ depends on  $\delta$ and on the prior. The R{\'e}nyi information dimension for a Bernoulli prior is 0, so Corollary \ref{cor:SC_AMP_almost_exact} is consistent with the result in \cite{Don13}. The key difference is that we use a binary-valued SC design for the QGT model rather than the Gaussian one in \cite{Don13}. Our analysis of the fixed point of the SC state evolution to establish Theorem \ref{thm:asym_MSE}  is also simpler than that in \cite{Don13}, where the authors use a continuum version of the state evolution along with a perturbation argument. In contrast, we use a straightforward  potential function analysis based on the recipe provided in \cite{Yed14} for analyzing coupled recursions.

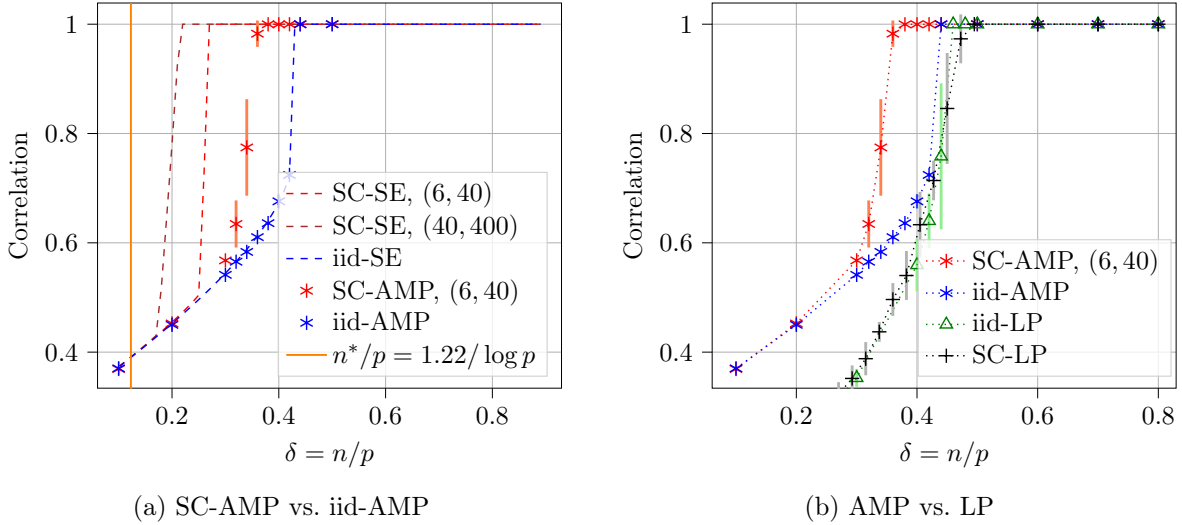
\begin{figure}[t]
\begin{subfigure}[b]{0.49\textwidth}
  \centering
  % This file was created with tikzplotlib v0.10.1.
\begin{tikzpicture}[scale=0.9]

\definecolor{coral}{RGB}{255,127,80}
\definecolor{darkgray176}{RGB}{176,176,176}
\definecolor{lightblue}{RGB}{173,216,230}
\definecolor{lightgray204}{RGB}{204,204,204}
\definecolor{purple}{RGB}{128,0,128}
\definecolor{brown}{RGB}{165,42,42}
\definecolor{lightgreen}{RGB}{144,238,144}
\definecolor{green}{RGB}{0,128,0}
\definecolor{mistyrose}{RGB}{255,228,225}
\definecolor{pink}{RGB}{255,192,203}

\begin{axis}[
legend cell align={left},
legend style={
  fill opacity=0.8,
  draw opacity=1,
  text opacity=1,
  at={(0.97,0.03)},
  anchor=south east,
  draw=lightgray204
},
tick align=outside,
tick pos=left,
x grid style={darkgray176},
xlabel={\(\displaystyle \delta=n/p\)},
xmajorgrids,
xmin=0.0605, xmax=0.9295,
xtick style={color=black},
y grid style={darkgray176},
ylabel={Correlation},
ymajorgrids,
ymin=0.334052228213881, ymax=1.03875983911359,
ytick style={color=black}
]
\addplot[ultra thick, red, dashed]
table {%
0.1 0.37304487583027
0.11 0.380732001906715
0.12 0.388504949934224
0.13 0.39637007550337
0.14 0.404334864459401
0.15 0.412408065134802
0.16 0.420600216309898
0.17 0.428923996015686
0.18 0.43739506317044
0.19 0.446033447098843
0.2 0.454865329569874
0.21 0.463926493758212
0.22 0.473269658178966
0.23 0.482979351193511
0.24 0.493213882101558
0.25 0.504375555649583
0.26 0.728624915866601
0.27 0.999999426696844
0.28 0.999999426696844
0.29 0.999999426696844
0.3 0.999999426696844
0.31 0.999999426696844
0.32 0.999999426696844
0.33 0.999999426696844
0.34 0.999999426696844
0.35 0.999999426696844
0.36 0.999999426696844
0.37 0.999999426696844
0.38 0.999999426696844
0.39 0.999999426696844
0.4 0.999999426696844
0.41 0.999999426696844
0.42 0.999999426696844
0.43 0.999999426696844
0.44 0.999999426696844
0.45 0.999999426696844
0.46 0.999999426696844
0.47 0.999999426696844
0.48 0.999999426696844
0.49 0.999999426696844
0.5 0.999999426696845
0.51 0.999999426696843
0.52 0.999999426696843
0.53 0.999999426696843
0.54 0.999999426696843
0.55 0.999999426696843
0.56 0.999999426696843
0.57 0.999999426696843
0.58 0.999999426696843
0.59 0.999999426696843
0.6 0.999999426696843
0.61 0.999999426696843
0.62 0.999999426696843
0.63 0.999999426696846
0.64 0.999999426696843
0.65 0.999999426696843
0.66 0.999999426696843
0.67 0.999999426696843
0.68 0.999999426696843
0.69 0.999999426696843
0.7 0.999999426696843
0.71 0.999999426696843
0.72 0.999999426696843
0.73 0.999999426696843
0.74 0.999999426696843
0.75 0.999999426696843
0.76 0.999999426696843
0.77 0.999999426696843
0.78 0.999999426696843
0.79 0.999999426696842
0.8 0.999999426696842
0.81 0.999999426696848
0.82 0.999999426696836
0.83 0.999999426696836
0.84 0.999999426696836
0.85 0.999999426696836
0.86 0.999999426696836
0.87 0.999999426696836
0.88 0.999999426696836
0.89 0.999999426696836
};
\addlegendentry{SC-SE, ($6, 40$)}

\addplot [ultra thick, brown, dashdotted]
table {%
0.1 0.373165147353436
0.11 0.380897016828789
0.12 0.388729886858495
0.13 0.396676333014112
0.14 0.404753844185859
0.15 0.412989109933236
0.16 0.421429732888253
0.17 0.430198580493941
0.18 0.517529931825281
0.19 0.647368863481989
0.2 0.780317153095167
0.21 0.924691455386218
0.22 0.999999426696844
0.23 0.999999426696844
0.24 0.999999426696844
0.25 0.999999426696844
0.26 0.999999426696844
0.27 0.999999426696844
0.28 0.999999426696844
0.29 0.999999426696844
0.3 0.999999426696844
0.31 0.999999426696844
0.32 0.999999426696844
0.33 0.999999426696844
0.34 0.999999426696844
0.35 0.999999426696844
0.36 0.999999426696844
0.37 0.999999426696844
0.38 0.999999426696844
0.39 0.999999426696844
0.4 0.999999426696844
0.41 0.999999426696844
0.42 0.999999426696844
0.43 0.999999426696844
0.44 0.999999426696844
0.45 0.999999426696844
0.46 0.999999426696844
0.47 0.999999426696844
0.48 0.999999426696844
0.49 0.999999426696844
0.5 0.999999426696844
0.51 0.999999426696844
0.52 0.999999426696844
0.53 0.999999426696844
0.54 0.999999426696844
0.55 0.999999426696844
0.56 0.999999426696844
0.57 0.999999426696844
0.58 0.999999426696844
0.59 0.999999426696844
0.6 0.999999426696844
0.61 0.999999426696844
0.62 0.999999426696844
0.63 0.999999426696844
0.64 0.999999426696844
0.65 0.999999426696843
0.66 0.999999426696843
0.67 0.999999426696844
0.68 0.999999426696843
0.69 0.999999426696843
0.7 0.999999426696843
0.71 0.999999426696843
0.72 0.999999426696843
0.73 0.999999426696843
0.74 0.999999426696843
0.75 0.999999426696843
0.76 0.999999426696843
0.77 0.999999426696843
0.78 0.999999426696843
0.79 0.99999942669683
0.8 0.999999426696836
0.81 0.999999426696836
0.82 0.999999426696836
0.83 0.999999426696836
0.84 0.999999426696836
0.85 0.999999426696836
0.86 0.999999426696836
0.87 0.999999426696836
0.88 0.999999426696836
0.89 0.999999426696836
};
\addlegendentry{SC-SE, ($40,400$)}

\addplot [ultra thick, cyan!65, dotted]
table {%
0.1 0.372741206678121
0.11 0.380344188045443
0.12 0.388016026875203
0.13 0.395760246077107
0.14 0.403580159056106
0.15 0.41147994258627
0.16 0.419464288090143
0.17 0.427537845896277
0.18 0.435706237135899
0.19 0.4439755571778
0.2 0.452352983951659
0.21 0.460845533762316
0.22 0.469462231621668
0.23 0.478212250847547
0.24 0.487106592395222
0.25 0.496157999057228
0.26 0.505379989633053
0.27 0.514789521217478
0.28 0.524404946590207
0.29 0.534249254954896
0.3 0.544348542827835
0.31 0.554734360284741
0.32 0.565446200432558
0.33 0.5765318681258
0.34 0.588052363509111
0.35 0.600086848670856
0.36 0.612741733421028
0.37 0.626163666266409
0.38 0.640570160134388
0.39 0.65630042694948
0.4 0.673949522902346
0.41 0.694756745770847
0.42 0.722454078914969
0.43 0.999999426553511
0.44 0.999999426550178
0.45 0.999999426546844
0.46 0.999999426543511
0.47 0.999999426540178
0.48 0.999999426536844
0.49 0.999999426533511
0.5 0.999999426530178
0.51 0.999999426526844
0.52 0.999999426523511
0.53 0.999999426520178
0.54 0.999999426516844
0.55 0.999999426513511
0.56 0.999999426510178
0.57 0.999999426506844
0.58 0.999999426503511
0.59 0.999999426500178
0.6 0.999999426496844
0.61 0.999999426493511
0.62 0.999999426490178
0.63 0.999999426486844
0.64 0.999999426483511
0.65 0.999999426480178
0.66 0.999999426476844
0.67 0.999999426473511
0.68 0.999999426470178
0.69 0.999999426466844
0.7 0.999999426463511
0.71 0.999999426460178
0.72 0.999999426456788
0.73 0.999999426453503
0.74 0.999999426450169
0.75 0.999999426446836
0.76 0.999999426443503
0.77 0.999999426440169
0.78 0.999999426436836
0.79 0.999999426433503
0.8 0.999999426430169
0.81 0.999999426426836
0.82 0.999999426423503
0.83 0.999999426420169
0.84 0.999999426416836
0.85 0.999999426413503
0.86 0.999999426410169
0.87 0.999999426406836
0.88 0.999999426403503
0.89 0.999999426400169
};
\addlegendentry{iid-SE}

\path [draw=lightblue, very thick]
(axis cs:0.1,0.366153239339615)
--(axis cs:0.1,0.37447280812562);

\path [draw=lightblue, very thick]
(axis cs:0.2,0.445370157824841)
--(axis cs:0.2,0.454388254493946);

\path [draw=lightblue, very thick]
(axis cs:0.3,0.538298931776124)
--(axis cs:0.3,0.544488617350074);

\path [draw=lightblue, very thick]
(axis cs:0.32,0.563407965412807)
--(axis cs:0.32,0.567401415321631);

\path [draw=lightblue, very thick]
(axis cs:0.34,0.580159084751205)
--(axis cs:0.34,0.586621042020505);

\path [draw=lightblue, very thick]
(axis cs:0.36,0.605596040661169)
--(axis cs:0.36,0.615210783155118);

\path [draw=lightblue, very thick]
(axis cs:0.38,0.633480453526956)
--(axis cs:0.38,0.638498064221972);

\path [draw=lightblue, very thick]
(axis cs:0.4,0.66356722151816)
--(axis cs:0.4,0.688260007355563);

\path [draw=lightblue, very thick]
(axis cs:0.42,0.720238961604502)
--(axis cs:0.42,0.72810596950998);

\path [draw=lightblue, very thick]
(axis cs:0.44,1)
--(axis cs:0.44,1);

\path [draw=lightblue, very thick]
(axis cs:0.5,1)
--(axis cs:0.5,1);

\path [draw=coral, very thick]
(axis cs:0.1,0.366084392345686)
--(axis cs:0.1,0.372201072292165);

\path [draw=coral, very thick]
(axis cs:0.2,0.448035839030667)
--(axis cs:0.2,0.456962000340682);

\path [draw=coral, very thick]
(axis cs:0.3,0.565883140722399)
--(axis cs:0.3,0.569791362621719);

\path [draw=coral, very thick]
(axis cs:0.32,0.591392417103189)
--(axis cs:0.32,0.677592114398055);

\path [draw=coral, very thick]
(axis cs:0.34,0.686218558076211)
--(axis cs:0.34,0.862813213738142);

\path [draw=coral, very thick]
(axis cs:0.36,0.958491812237902)
--(axis cs:0.36,1.00672767498179);

\path [draw=coral, very thick]
(axis cs:0.38,0.999572922609272)
--(axis cs:0.38,0.999985630863944);

\path [draw=coral, very thick]
(axis cs:0.4,0.999500037432323)
--(axis cs:0.4,0.999928943413113);

\path [draw=coral, very thick]
(axis cs:0.42,0.999810479266249)
--(axis cs:0.42,0.999967483444002);

\path [draw=coral, very thick]
(axis cs:0.44,0.999562209806696)
--(axis cs:0.44,0.999879881339557);

\path [draw=coral, very thick]
(axis cs:0.5,1)
--(axis cs:0.5,1);

\addplot [ultra thick, red, mark=star, mark size=3, mark options={solid}, only marks]
table {%
0.1 0.369142732318926
0.2 0.452498919685674
0.3 0.567837251672059
0.32 0.634492265750622
0.34 0.774515885907177
0.36 0.982609743609844
0.38 0.999779276736608
0.4 0.999714490422718
0.42 0.999888981355126
0.44 0.999721045573127
0.5 1
};
\addlegendentry{SC-AMP, ($6, 40$)}
\addplot [ultra thick, cyan!65, mark=square, mark size=2, mark options={solid}, only marks]
table {%
0.1 0.370313023732617
0.2 0.449879206159394
0.3 0.541393774563099
0.32 0.565404690367219
0.34 0.583390063385855
0.36 0.610403411908144
0.38 0.635989258874464
0.4 0.675913614436861
0.42 0.724172465557241
0.44 1
0.5 1
};
\addlegendentry{iid-AMP}

\addplot [ultra thick, orange]
table {%
0.123 0
0.123 0.5
0.123 0.8
0.123 1
0.123 2
};
\addlegendentry{$n^*/p = 1.22/\log p$}

\end{axis}

\end{tikzpicture}
  \vspace{-2\baselineskip}
  \caption{SC-AMP vs.~iid-AMP}
  \label{fig:corr_v_delta_SC_iid}
\end{subfigure}
\begin{subfigure}[b]{0.49\textwidth}
  \centering
  % This file was created with tikzplotlib v0.10.1.
\begin{tikzpicture}[scale=0.9]

\definecolor{coral}{RGB}{255,127,80}
\definecolor{darkgray176}{RGB}{176,176,176}
\definecolor{lightblue}{RGB}{173,216,230}
\definecolor{lightgray204}{RGB}{204,204,204}
\definecolor{purple}{RGB}{128,0,128}
\definecolor{brown}{RGB}{165,42,42}
\definecolor{lightgreen}{RGB}{144,238,144}
\definecolor{green}{RGB}{0,128,0}
\definecolor{black!30!white}{RGB}{255,228,225}
\definecolor{pink}{RGB}{255,192,203}

\begin{axis}[
legend cell align={left},
legend style={
  fill opacity=0.8,
  draw opacity=1,
  text opacity=1,
  at={(0.99,0.03)},
  anchor=south east,
  draw=lightgray204
},
tick align=outside,
tick pos=left,
x grid style={darkgray176},
xlabel={\(\displaystyle \delta=n/p\)},
xmajorgrids,
xmin=0.0605, xmax=0.8295,
xtick style={color=black},
y grid style={darkgray176},
ylabel={Correlation},
ymajorgrids,
ymin=0.334052228213881, ymax=1.03875983911359,
ytick style={color=black}
]

\path [draw=lightblue, very thick]
(axis cs:0.1,0.366153239339615)
--(axis cs:0.1,0.37447280812562);

\path [draw=lightblue, very thick]
(axis cs:0.2,0.445370157824841)
--(axis cs:0.2,0.454388254493946);

\path [draw=lightblue, very thick]
(axis cs:0.3,0.538298931776124)
--(axis cs:0.3,0.544488617350074);

\path [draw=lightblue, very thick]
(axis cs:0.32,0.563407965412807)
--(axis cs:0.32,0.567401415321631);

\path [draw=lightblue, very thick]
(axis cs:0.34,0.580159084751205)
--(axis cs:0.34,0.586621042020505);

\path [draw=lightblue, very thick]
(axis cs:0.36,0.605596040661169)
--(axis cs:0.36,0.615210783155118);

\path [draw=lightblue, very thick]
(axis cs:0.38,0.633480453526956)
--(axis cs:0.38,0.638498064221972);

\path [draw=lightblue, very thick]
(axis cs:0.4,0.66356722151816)
--(axis cs:0.4,0.688260007355563);

\path [draw=lightblue, very thick]
(axis cs:0.42,0.720238961604502)
--(axis cs:0.42,0.72810596950998);

\path [draw=lightblue, very thick]
(axis cs:0.44,1)
--(axis cs:0.44,1);

\path [draw=lightblue, very thick]
(axis cs:0.5,1)
--(axis cs:0.5,1);

\path [draw=coral, very thick]
(axis cs:0.1,0.366084392345686)
--(axis cs:0.1,0.372201072292165);

\path [draw=coral, very thick]
(axis cs:0.2,0.448035839030667)
--(axis cs:0.2,0.456962000340682);

\path [draw=coral, very thick]
(axis cs:0.3,0.565883140722399)
--(axis cs:0.3,0.569791362621719);

\path [draw=coral, very thick]
(axis cs:0.32,0.591392417103189)
--(axis cs:0.32,0.677592114398055);

\path [draw=coral, very thick]
(axis cs:0.34,0.686218558076211)
--(axis cs:0.34,0.862813213738142);

\path [draw=coral, very thick]
(axis cs:0.36,0.958491812237902)
--(axis cs:0.36,1.00672767498179);

\path [draw=coral, very thick]
(axis cs:0.38,0.999572922609272)
--(axis cs:0.38,0.999985630863944);

\path [draw=coral, very thick]
(axis cs:0.4,0.999500037432323)
--(axis cs:0.4,0.999928943413113);

\path [draw=coral, very thick]
(axis cs:0.42,0.999810479266249)
--(axis cs:0.42,0.999967483444002);

\path [draw=coral, very thick]
(axis cs:0.44,0.999562209806696)
--(axis cs:0.44,0.999879881339557);

\path [draw=coral, very thick]
(axis cs:0.5,1)
--(axis cs:0.5,1);

\path [draw=lightgreen, very thick]
(axis cs:0.1,0.12026202339847)
--(axis cs:0.1,0.156008791869709);

\path [draw=lightgreen, very thick]
(axis cs:0.2,0.209977636978268)
--(axis cs:0.2,0.241171983054605);

\path [draw=lightgreen, very thick]
(axis cs:0.3,0.336469588171813)
--(axis cs:0.3,0.369275072804699);

\path [draw=lightgreen, very thick]
(axis cs:0.4,0.511395712541128)
--(axis cs:0.4,0.607469784338639);

\path [draw=lightgreen, very thick]
(axis cs:0.42,0.590203314412434)
--(axis cs:0.42,0.690065359091663);

\path [draw=lightgreen, very thick]
(axis cs:0.44,0.624668539680832)
--(axis cs:0.44,0.891876469691334);

\path [draw=lightgreen, very thick]
(axis cs:0.46,0.999999999999944)
--(axis cs:0.46,1.00000000000003);

\path [draw=lightgreen, very thick]
(axis cs:0.48,0.999999999999999)
--(axis cs:0.48,1);

\path [draw=lightgreen, very thick]
(axis cs:0.5,1)
--(axis cs:0.5,1);

\addplot [very thick, red, loosely dotted, mark=star, mark size=3, mark options={solid}]
table {%
0.1 0.369142732318926
0.2 0.452498919685674
0.3 0.567837251672059
0.32 0.634492265750622
0.34 0.774515885907177
0.36 0.982609743609844
0.38 0.999779276736608
0.4 0.999714490422718
0.42 0.999888981355126
0.44 0.999721045573127
0.5 1
0.6 1
0.7 1
0.8 1
};
\addlegendentry{SC-AMP, ($6, 40$)}
\addplot [very thick, cyan!65, loosely dotted, mark=square, mark size=2, mark options={solid}]
table {%
0.1 0.370313023732617
0.2 0.449879206159394
0.3 0.541393774563099
0.32 0.565404690367219
0.34 0.583390063385855
0.36 0.610403411908144
0.38 0.635989258874464
0.4 0.675913614436861
0.42 0.724172465557241
0.44 1
0.5 1
0.6 1
0.7 1
0.8 1
};
\addlegendentry{iid-AMP}

\addplot [very thick, green, loosely dotted, mark=triangle, mark size=3, mark options={solid}]
table {%
0.1 0.138135407634089
0.2 0.225574810016436
0.3 0.352872330488256
0.4 0.559432748439884
0.42 0.640134336752049
0.44 0.758272504686083
0.46 0.999999999999986
0.48 1
0.5 1
0.6 1
0.7 1
0.8 1
};
\addlegendentry{iid-LP}

\path [draw=black!30!white, very thick]
(axis cs:0.2025,0.229657985166684)
--(axis cs:0.2025,0.270382529061232);

\path [draw=black!30!white, very thick]
(axis cs:0.225,0.252201235347987)
--(axis cs:0.225,0.28544721017185);

\path [draw=black!30!white, very thick]
(axis cs:0.2475,0.26415666144131)
--(axis cs:0.2475,0.306369087958367);

\path [draw=black!30!white, very thick]
(axis cs:0.27,0.301107485810567)
--(axis cs:0.27,0.345115136106449);

\path [draw=black!30!white, very thick]
(axis cs:0.2925,0.328061784894329)
--(axis cs:0.2925,0.375700127781396);

\path [draw=black!30!white, very thick]
(axis cs:0.315,0.357657581421835)
--(axis cs:0.315,0.418349818719617);

\path [draw=black!30!white, very thick]
(axis cs:0.3375,0.419062243167944)
--(axis cs:0.3375,0.455261515760796);

\path [draw=black!30!white, very thick]
(axis cs:0.36,0.466273552436804)
--(axis cs:0.36,0.526342526561903);

\path [draw=black!30!white, very thick]
(axis cs:0.3825,0.495300440868832)
--(axis cs:0.3825,0.585045153942644);

\path [draw=black!30!white, very thick]
(axis cs:0.405,0.573175331989819)
--(axis cs:0.405,0.693168949443228);

\path [draw=black!30!white, very thick]
(axis cs:0.4275,0.677849578054757)
--(axis cs:0.4275,0.75018714015533);

\path [draw=black!30!white, very thick]
(axis cs:0.45,0.744345518860331)
--(axis cs:0.45,0.947264217798048);

\path [draw=black!30!white, very thick]
(axis cs:0.4725,0.928463702540443)
--(axis cs:0.4725,1.01799137831976);

\path [draw=black!30!white, very thick]
(axis cs:0.495,1)
--(axis cs:0.495,1);

\addplot [very thick, black, loosely dotted, mark=+, mark size=3, mark options={solid}]
table {%
0.2025 0.250020257113958
0.225 0.268824222759918
0.2475 0.285262874699838
0.27 0.323111310958508
0.2925 0.351880956337863
0.315 0.388003700070726
0.3375 0.43716187946437
0.36 0.496308039499354
0.3825 0.540172797405738
0.405 0.633172140716523
0.4275 0.714018359105043
0.45 0.84580486832919
0.4725 0.973227540430102
0.495 1
0.6 1
0.7 1
0.8 1
};
\addlegendentry{SC-LP}

\end{axis}

\end{tikzpicture}
  \vspace{-2\baselineskip}
  \caption{AMP vs.~LP}
    \label{fig:corr_v_delta_AMP_LP}
\end{subfigure}
\caption{Normalized squared correlation for Noiseless QGT. $\pi=0.3, p = 20 000$. With  spatial coupling parameters  $\omega=6, \Lambda=40$, inner block size $p/\Lambda=500$. In (b), we use $p=2000$ for both iid-LP and SC-LP due to computational constraints. Error bars indicate one standard deviation.}
\label{fig:corr_v_delta}
\end{figure}

\subsection{Numerical Simulations}
We present simulation results for finite length SC-AMP and compare its performance against alternative algorithms and the information-theoretic limit.
The performance in all the plots is either measured via the normalized squared correlation between the SC-AMP estimate and the signal (see \eqref{eq:norm_sq_corr}) or via the FPR and FNR (see \eqref{eq:FPR_and_FNR}).  In the plots, curves labeled `SC-AMP' show the empirical performance of the SC-AMP algorithm, while the `SC-SE' curves refer to its theoretical performance predicted via state evolution.  The corresponding curves for an i.i.d.~design are labeled `iid-AMP' and `iid-SE'. For empirical performance curves, each point is obtained from 10 independent runs, where in each run, the SC-AMP algorithm is executed for 300 iterations. Other implementation details are described in Appendix \ref{sec:imp_details}. Python code for all the simulations is available at \cite{CobCode24}.

Figure \ref{fig:corr_v_delta_SC_iid} shows how normalized squared correlation varies with the sampling ratio $\delta$ for noiseless QGT, for both spatially coupled and i.i.d.~designs. We observe that SC-AMP outperforms iid-AMP, justifying the use of the SC design. The orange vertical line show the information-theoretic lower bound on the ratio $n/p$ obtained from \eqref{eq:gamma_star}. Specializing \eqref{eq:gamma_star} to the case of $L=2$, we get the information-theoretic lower bound on the number of tests for noiseless QGT: $    n^*=2H(\pi)\frac{p}{\log p}$. We observe that the performance of SC-SE improves and approaches $n^*$ as we increase the size of the spatial coupling parameters $(\omega,\Lambda)$ from (6,40) to (40,400). We did not implement the SC-AMP for $(\omega,\Lambda)=(40,400)$  as it requires a large amount of computational memory.  The difference between the SC-SE plot and the SC-AMP plot for $(6,40)$ is due to finite length effects, since the inner block size $p/\Lambda$ is only 500. 

Figure \ref{fig:corr_v_delta_AMP_LP} shows how the AMP algorithm compares to the linear programming (LP) estimator, defined as  the solution of the following linear program:
\begin{align}
    \text{minimize}\quad & \|\beta\|_1 \\
    \text{subject to}\quad & y=X\beta,
    \text{ and }
    0\leq\beta_j\leq1,\quad j\in[p]. \nonumber
\end{align}
Similar reconstruction algorithms are commonly used for  compressed sensing \cite{Fou13}. LP based estimators have also been used in Boolean in group testing \cite{Ald19}. 
We observe that the AMP algorithm outperforms LP for both i.i.d.~and SC designs, and that  the performance of LP is similar with both designs. This is because the LP algorithm is not equipped to take advantage of the spatially coupled design.     LP is also more computationally intensive than the SC-AMP algorithm and  challenging to implement for large values of $p$. Therefore, we use a smaller $p$ for all our LP experiments. 

Figure \ref{fig:fpr_v_fnr_noiseless} shows the tradeoff between the FPR and the FNR for noiseless QGT with $\delta=0.38$. The tradeoff curve is obtained by thresholding the AMP or LP estimate with different thresholds $\zeta$, as described in \eqref{eq:s_c_k_def}. SC-AMP achieves perfect recovery at this value of $\delta$, so its FPR and FNR are both 0, for all  threshold values. As expected, SC-AMP does significantly better than iid-AMP and LP.

Figure \ref{fig:fpr_v_fnr_noisy} shows the tradeoff between the FPR and the FNR for noisy QGT with $\delta=0.46$ and $\sigma^2=0.0016$.
Following the model in \eqref{eq:QGT_model}, for the i.i.d.~design we consider Gaussian noise with $\Psi_{i} \stackrel{\iid}{\sim} \normal(0,p\sigma^2)$, as previously investigated in \cite{Sca17, Tan23d}. For the SC design, we consider $\Psi_{i} \stackrel{\iid}{\sim} \normal(0,p \sigma^2/(2\sfC) \,  )$. As described in Section~\ref{sec:items-test}, for $\alpha=0.5$, the expected number of items in each test is approximately $\frac{p}{4\sfC}$ for the SC design, compared to $p/2$ for the i.i.d.~design. This choice of noise variance for the SC model ensures that the signal-to-noise ratio $\E\big[\| X \beta \|^2 \big]/\E \big[ \| \Psi \|^2 \big]$ is similar for both designs.

In the noisy setting, the AMP algorithm is compared to the  following convex programming (CVX) estimator:
\begin{align}
    \text{minimize}\quad & \frac{1}{2\text{Var}(\Psi)}\|y-X\beta\|_2^2 \, +  \, \|\beta\|_1\log\frac{1-\pi}{\pi} \\
    \text{subject to}\quad & 0\leq\beta_j\leq1,\quad j\in[p]. \nonumber
\end{align}
This estimator is obtained via a convex relaxation of the MAP estimator  for QGT. Figure \ref{fig:fpr_v_fnr_noisy} shows that in the presence of a small amount of noise, SC-AMP continues to achieve perfect recovery, outperforming both CVX and iid-AMP. Surprisingly, in the presence of noise, the performance of CVX is worse with the SC design than with the i.i.d.~one,  possibly because it does not take advantage of the band-diagonal structure in the SC design matrix.

\begin{figure}[t]
\begin{subfigure}[b]{0.49\textwidth}
  \centering
  % This file was created with tikzplotlib v0.10.1.
\begin{tikzpicture}[scale=0.9]

\definecolor{darkgray176}{RGB}{176,176,176}
\definecolor{green}{RGB}{0,128,0}
\definecolor{lightgray204}{RGB}{204,204,204}
\definecolor{pink}{RGB}{255,192,203}
\definecolor{purple}{RGB}{128,0,128}

\begin{axis}[
legend cell align={left},
legend style={fill opacity=0.8, draw opacity=1, text opacity=1, draw=lightgray204},
tick align=outside,
tick pos=left,
x grid style={darkgray176},
xlabel={FNR},
xmajorgrids,
xmin=-0.037288784691927, xmax=0.783064478530468,
xtick style={color=black},
y grid style={darkgray176},
ylabel={FPR},
ymajorgrids,
ymin=-0.0203987629231372, ymax=0.428374021385881,
ytick style={color=black}
]
\addplot [ultra thick, red, dashed, mark=o, mark size=3, mark options={solid}]
table {%
0 0
0 0
0 0
0 0
0 0
0 0
0 0
0 0
0 0
};
\addlegendentry{SC-SE, (6, 40)}
\addplot [ultra thick, cyan!65, dotted, mark=diamond, mark size=3, mark options={solid}]
table {%
0.048606994665086 0.407975258462743
0.109495614181147 0.253653635628271
0.172239132027467 0.172048865679406
0.237496919603312 0.119164401455452
0.308222827571049 0.0817062399780482
0.387460620874766 0.0538148443506919
0.479649400904471 0.0325675919018342
0.592478503826352 0.0164209948636157
0.745775693838541 0.00520927121652561
};
\addlegendentry{iid-SE}
\addplot [ultra thick, red, mark=star, mark size=3, mark options={solid}, only marks]
table {%
0.000200132628318542 4.29232484798508e-05
0.000200132628318542 4.29232484798508e-05
0.000200132628318542 4.29232484798508e-05
0.000200132628318542 4.29232484798508e-05
0.000200132628318542 4.29232484798508e-05
0.000200132628318542 4.29232484798508e-05
0.000200132628318542 4.29232484798508e-05
0.000200132628318542 4.29232484798508e-05
0.000200132628318542 4.29232484798508e-05
};
\addlegendentry{SC-AMP, (6, 40)}
\addplot [ultra thick, cyan!65, mark=square, mark size=2, mark options={solid}, only marks]
table {%
0.0488379579546189 0.405284324937064
0.110521206847751 0.252575476061984
0.171723713949312 0.170983418004003
0.237389285913635 0.118476736062038
0.306261631169358 0.0816589727565602
0.385414287499207 0.0540100678323032
0.475100233922364 0.0329887843594975
0.588004574402556 0.0167056454036335
0.741179095361817 0.0054912375703692
};
\addlegendentry{iid-AMP}
\addplot [ultra thick, black, mark=+, mark size=3, mark options={solid}, only marks]
table {%
0.1500630623816 0.261287152660988
0.187097566879083 0.214255390227051
0.224665319660118 0.173415842569413
0.268123362268961 0.143977159219892
0.316223526677457 0.117971465496343
0.369059081504289 0.0993931967405601
0.421359052745265 0.0822160691155132
0.47856597490258 0.0703049513427243
0.544342978419698 0.0577060699581273
};
\addlegendentry{SC-LP}
\addplot [ultra thick, green, mark=triangle, mark size=3, mark options={solid}, only marks]
table {%
0.177555315698011 0.310134486206304
0.219007686994193 0.259352396816911
0.258745314178127 0.214867765999933
0.305534322061262 0.172167116573358
0.360210710495085 0.137158706263726
0.417615656184292 0.106765357768374
0.475033946089539 0.081000018452686
0.537914827750051 0.059505327887403
0.59449724071094 0.0443609311372143
};
\addlegendentry{iid-LP}

\end{axis}

\end{tikzpicture}
  \vspace{-2\baselineskip}
  \caption{Noiseless QGT.  $\delta=0.38$}
  \label{fig:fpr_v_fnr_noiseless}
\end{subfigure}
\begin{subfigure}[b]{0.49\textwidth}
  \centering
  % This file was created with tikzplotlib v0.10.1.
\begin{tikzpicture}[scale=0.9]

\definecolor{darkgray176}{RGB}{176,176,176}
\definecolor{lightgray204}{RGB}{204,204,204}
\definecolor{purple}{RGB}{128,0,128}
\definecolor{green}{RGB}{0,128,0}
\definecolor{pink}{RGB}{255,192,203}

\begin{axis}[
legend cell align={left},
legend style={fill opacity=0.8, draw opacity=1, text opacity=1, draw=lightgray204},
tick align=outside,
tick pos=left,
x grid style={darkgray176},
xlabel={FNR},
xmajorgrids,
xmin=-0.027464785852524, xmax=0.627103432542412,
xtick style={color=black},
y grid style={darkgray176},
ylabel={FPR},
ymajorgrids,
ymin=-0.0137847228313352, ymax=0.339622947868769,
ytick style={color=black}
]
\addplot [ultra thick, red, dashed, mark=o, mark size=3, mark options={solid}]
table {%
1.55877108821394e-05 6.84203365523507e-05
2.52570930336269e-05 4.53516844579777e-05
3.44263347290031e-05 3.36745184597425e-05
4.30954359682678e-05 2.71038807115001e-05
5.37650990319781e-05 2.10331827919283e-05
6.45181188383738e-05 1.69265341992767e-05
7.90221920656051e-05 1.34983753741068e-05
0.000102278723274786 9.96308658565023e-06
0.000150792347517595 6.53492776048026e-06
};
\addlegendentry{SC-SE, (6, 40)}
\addplot [ultra thick, cyan!65, dotted, mark=diamond, mark size=3, mark options={solid}]
table {%
0.0405808128844081 0.275831690109674
0.0809613914947155 0.173161743183061
0.121061642121282 0.121022173167455
0.163444563472596 0.0871448744986442
0.209231467483623 0.062490094252739
0.26106544657249 0.0439622417901881
0.325584266925636 0.0292339127094842
0.412719548538457 0.0167915802219173
0.549623058979005 0.00673947777614368
};
\addlegendentry{iid-SE}
\addplot [ultra thick, red, mark=star, mark size=3, mark options={solid}, only marks]
table {%

0 0
0 0
0 0
0 0
0 0
0 0
0 0
0 0
0 0
};
\addlegendentry{SC-AMP, (6, 40)}

\addplot [ultra thick, cyan!65, mark=square, mark size=2, mark options={solid}, only marks]
table {%
0.0444477565768826 0.273256838631946
0.0873094915982304 0.172217873305232
0.126056753825066 0.119793930887086
0.166904401590695 0.0867633059343443
0.212318086560219 0.0628978394908813
0.264021149666342 0.0447981725347235
0.328382564224 0.0302466040827123
0.412954356254446 0.0178833871827914
0.548614013326529 0.00738168912306697
};
\addlegendentry{iid-AMP}
\addplot [ultra thick, black, mark=+, mark size=3, mark options={solid}, only marks]
table {%
0.0853489194380504 0.287625799257458
0.113550082309137 0.215854542949953
0.145977515737426 0.156488475170563
0.179052334584256 0.115077530063792
0.230059761329547 0.0836915675070814
0.292135761862204 0.0572752864672797
0.372887868453887 0.0401910113494349
0.45424301042245 0.0284491348707463
0.53860300580294 0.0198637042475831
};
\addlegendentry{SC-CVX}

\addplot [ultra thick, green, mark=triangle, mark size=3, mark options={solid}, only marks]
table {%
0.0444920926638841 0.281843873353821
0.0675750898160707 0.208865881494992
0.100658985429585 0.143339395529369
0.147727096497222 0.0974192388678748
0.201920880707658 0.0621034449585609
0.281281803153978 0.0381060880885044
0.363508708624836 0.0225793505185139
0.463067678665816 0.0119974681536618
0.554683149012466 0.00719018344841792
};
\addlegendentry{iid-CVX}

\end{axis}

\end{tikzpicture}
  \vspace{-2\baselineskip}
  \caption{Noisy QGT. $\delta=0.46$, $\sigma^2=0.0016$}
    \label{fig:fpr_v_fnr_noisy}
\end{subfigure}
\caption{FPR vs.~FNR tradeoff. In both (a) and (b), $\pi=0.3$,  and thresholds $\zeta\in\{0.1, 0.2, \dots, 0.9\}$. We take $p=20000$ for AMP, and $p=2000$ for LP and CVX.}
\label{fig:fpr_v_fnr}
\end{figure}
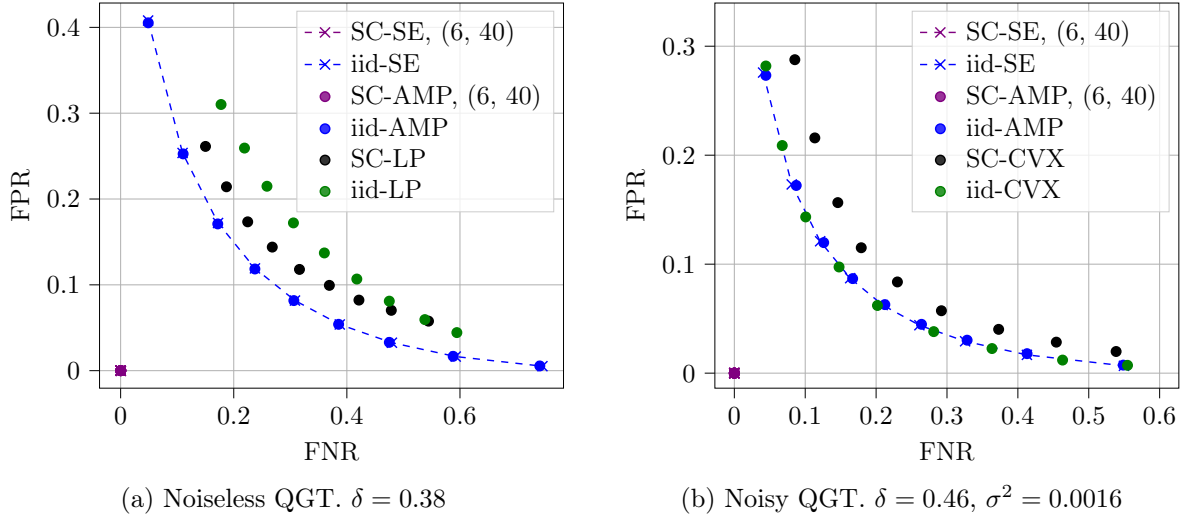

\section{SC-AMP for Pooled Data} \label{sec:SCAMP_pooled_data}

In this section we extend the SC-AMP algorithm to the pooled data model in \eqref{eq:pooled_data_model} with the spatially coupled design $\Xsc$ defined in \eqref{eq:SC_matrix}. We apply SC-AMP to a centered and rescaled version of the pooled model, as we did for QGT in \eqref{eq:rescaled_SC_QGT}. Recalling the decomposition $X_{ij}^{\text{sc}}=\alpha W_{\sfr\sfc}+\sqrt{\frac{n\alpha(1-\alpha)}{\sfR}}\tX_{ij}^{\text{sc}}$, we have 
\begin{align*}
    &Y_{i,:}
    =\sum_{j=1}^pX_{ij}^{\text{sc}}B_{j,:}+\Psi_{i,:}
    =\sum_{j=1}^p\left(\alpha W_{\sfr(i)\sfc(j)}+\sqrt{\frac{n\alpha(1-\alpha)}{\sfR}}\tXsc_{ij}\right)B_{j,:}+\Psi_{i,:} \\ 
    \implies \, 
    &Y_{i,:}-\sum_{j=1}^p\alpha W_{\sfr(i)\sfc(j)}B_{j,:}
    =\sqrt{\frac{n\alpha(1-\alpha)}{\sfR}}\sum_{j=1}^p\tXsc_{ij}B_{j,:}+\Psi_{i,:} \, .
\end{align*}
Defining
\begin{align*}
    \tY_{i,:}
    &:=\frac{1}{\sqrt{n\alpha(1-\alpha)/\sfR}}\left(Y_{i,:}-\alpha\left(W_{\sfr(i)1}\sum_{j\in\mathcal{J}_1}B_{j,:}+\dots+W_{\sfr(i)\sfC}\sum_{j\in\mathcal{J}_\sfC}B_{j,:}\right)\right),
\end{align*}
and $\tPsi_{i,:} :=\frac{1}{\sqrt{n\alpha(1-\alpha)/\sfR}}\Psi_{i,:}$ gives us the rescaled pooled data model:
\begin{align}
    \tY_{i,:}
    &=B^\top\tXsc_{i,:}+\tPsi_{i,:}
    \in\mb{R}^L, \quad 
    \text{ for } i\in[n].
    \label{eq:rescaled_pooled_data}
\end{align}
The sets $ \big(\mathcal{J}_\sfc\big)_{\sfc \in [\sfC]}$ are defined in \eqref{eq:J_c_I_r_def}. In the noiseless setting, the terms $\sum_{j\in\mathcal{J}_c}B_{j,:}$ can be obtained with an extra $\sfC =O(1)$ tests, where the $\sfc$th test 
only includes items from $\mathcal{J}_\sfc$. 

\paragraph{Model and noise scaling assumptions.}  The signal matrix $B \in \mb{R}^{p \times L}$ and the rescaled noise matrix $\tPsi \in \mb{R}^{n \times L_{\Psi}}$ are both independent of the design matrix. As $p\to\infty$, we assume that $n/p \to \delta >0$. As $p,n \to \infty$, the empirical distributions of the  rows of $B$ and  $\tPsi$ each converge  to well-defined limits. More precisely,  $B\stackrel{W}{\rightarrow}\bar{B}$ and $\tPsi\stackrel{W}{\rightarrow}\bar{\Psi}$, for $L$-dimensional random vectors $\bar{B}\sim \text{Categorical}(\pi)$ and $\bar{\Psi}\sim P_{\bar{\Psi}}$. \label{assump:model_noise_pooled_data}

\subsection{Matrix SC-AMP Algorithm}

The goal is to recover $B$ from $\tY$ generated according to the rescaled model \eqref{eq:rescaled_pooled_data}.
The matrix SC-AMP algorithm is initialized with $  \hB_{j,:}^0=\E[\bar{B}]$ for $j\in[p]$, and 
        $\tTheta^0=\tY-\tXsc \hB^0$. For iteration $k \ge 1$, we compute:
    \begin{align}
        \tTheta^k
        &=\tY-\tXsc \hB^k+U^k, \nonumber\\
        B^{k+1}&=V^k+\hB^k, \qquad
        \hB^{k+1}
        =f_{k+1}(B^{k+1},\mathcal{C}),
        \label{eq:matrix_SC_AMP}
    \end{align}
    where $f_k: \reals^{L} \times [\sfC] \to \reals$ acts row-wise on its input, and the vector $\mathcal{C}\in\mb{R}^p$ is defined in \eqref{eq:mC_def}.
Similarly to QGT, the function $f_k$ and the matrices $U^k\in\mb{R}^{n\times L}$, $V^k\in\mb{R}^{p\times L}$ are defined in terms of block-wise state evolution parameters. Here, the key state evolution parameters are two sets of $L \times L$ covariance matrices, denoted by $\phi_\sfr^{k}$ for $\sfr \in [\sfR]$ and $\Tau_{\sfc}^k$ for $\sfc \in [\sfC]$. These matrices are computed recursively as given in   \eqref{eq:matrix_SCAMP_SE} below.

Recalling the partition in \eqref{eq:J_c_I_r_def}, the function 
 $f_k: \reals^{L} \times [\sfC] \to \reals$ is defined as follows, for $j \in \mathcal{J}_\sfc, \sfc \in [\sfC]$:
\begin{align}
\label{eq:optimal_fk_pooled_data}
    f_k(B_{j,:}^k,\sfc)
    &=\E\big[\bar{B}\big|\bar{B}+G_\sfc^k=B_{j,:}^k\big], \quad  
    G_\sfc^k\sim\normal
    \left(
    0,\Tau_\sfc^k
    \right) \text{ independent of } \bar{B}.
\end{align}
The matrices $U^k\in\mb{R}^{n\times L}$, $V^k\in\mb{R}^{p\times L}$ are defined in terms of a matrix $Q^k\in\mb{R}^{L\sfR\times L\sfC}$, whose  sub-matrices $Q_{\sfr,\sfc}^k\in\mb{R}^{L\times L}$ for $\sfr \in [\sfR], \sfc \in [\sfC]$, are given by
\begin{align*}
    Q_{\sfr,\sfc}^k
    &=(\phi_\sfr^k)^{-1}\left(\sum_{\sfr'=1}^\sfR\tW_{\sfr'\sfc}(\phi_{\sfr'}^k)^{-1}\right)^{-1}.
\end{align*} 
The rows of the matrix $U^k \in \reals^{n \times L}$ in \eqref{eq:matrix_SC_AMP} are defined as
\begin{align}
    U_{i,:}^k
    &=\frac{1}{(n/\sfR)}\tTheta_{i,:}^{k-1}\sum_{\sfc=1}^\sfC\tW_{\sfr(i),\sfc} \, Q_{\sfr(i),\sfc}^{k-1} \, 
   \sum_{j\in\mathcal{J}_\sfc}f_k'(B_{j,:}^k,\sfc)^\top,
    \quad \text{for $i\in[n]$}
    \label{eq:U_k}
\end{align}
where $f_k'$ denotes the $L \times L$ Jacobian of $f_k$ with respect to the first argument. The rows of the matrix $V^k \in \reals^{p \times L}$ are given by 
\begin{align*}
    V_{j,:}^k
    =\sum_{i=1}^n\tXsc_{ij}\, \tTheta_{i,:}^k \, Q_{\sfr(i),\sfc(j)}^k,
    \quad \text{for $j\in[p]$}.
\end{align*}

\paragraph{State evolution.}
The memory term $U^k \in \reals^{n \times L}$ in the matrix SC-AMP  \eqref{eq:matrix_SC_AMP} debiases the iterates and ensures that: \emph{i}) for $\sfr \in [\sfR]$,  the row-wise empirical distribution of $\tTheta^k_\sfr$ converges to $\normal(0, \phi^k_\sfr)$, and \emph{ii}) for $\sfc \in [\sfC]$, the row-wise empirical distribution of $B^k_\sfc$ converges to the law of $(\bar{B} +  G_\sfc^k)$, where $G_\sfc^k \sim \normal(0, \Tau^k_\sfc)$.  The $L \times L$ covariance matrices 
$\phi^k_\sfr$ and $\Tau^k_\sfc$ are iteratively computed as follows for $k \ge 0$,  starting from the initialization $\psi_\sfc^0=\Cov(\bar{B})$, for $\sfc \in [\sfC]$:
\begin{align}
\label{eq:matrix_SCAMP_SE}
\begin{split}
    & \phi_\sfr^{k}
    =\Cov[\bar{\Psi}]+\frac{1}{\deltain}\sum_{\sfr=1}^\sfC\tW_{\sfr\sfc}\psi_\sfc^{k} \, , 
    \qquad      \sfr \in [\sfR],
   \\
    & \Tau_{\sfc}^k ={\left[\sum_{\sfr=1}^\sfR \tW_{\sfr\sfc}{[\phi_\sfr^k]}^{ -1}\right]^{-1}}, \qquad \psi_\sfc^{k+1}
    =\E\Big[\Big(f_k\big(\bar{B}+G_\sfc^k,\sfc\big)-\bar{B}\Big)
    \Big(f_k\big(\bar{B}+G_\sfc^k,\sfc\big)-\bar{B}\Big)^\top\Big],  \quad \sfc \in [\sfC].
\end{split}
\end{align}
where $G_\sfc^k \sim \normal(0, \Tau_{\sfc}^k)$ is independent of $\bar{B}$, and we recall that $\tW$ is defined in \eqref{eq:W_tilde}.

The matrix SC-AMP algorithm was proposed and analyzed in \cite{Liu24} for a model with a spatially coupled \emph{Gaussian} design matrix. Similarly to our analysis of SC-AMP for QGT, we could apply a reduction technique along with the universality result of \cite{Wan22} to establish a state evolution characterization for the matrix SC-AMP applied to pooled data.  Such a result would be analogous to Theorem \ref{thm:SC_AMP_QGT}, and show that for each iteration $k \ge 1$, the joint empirical distribution of the rows of $(B, B^k)$  converges as: 
\begin{align}
(B, B^k)_{\mathcal{J}_\sfc,:}
\stackrel{W_2}{\rightarrow}
(\bar{B}, \bar{B}+G_\sfc^k), \quad  \sfc \in [\sfC].
\label{eq:matrix_SC_AMP_convergence}
\end{align}

To analyze the limiting MSE and error rate of the matrix SC-AMP algorithm, we need to characterize the fixed point of the state evolution recursion in \eqref{eq:matrix_SCAMP_SE} (as $k \to \infty$). In QGT, the state evolution fixed point was characterized via the minimizer of a scalar potential function (Theorem \ref{thm:asym_MSE}). Extending this approach to the pooled data setting is challenging as the state evolution parameters are now $L \times L$ matrices rather than scalars.  In the following section, we circumvent this issue by showing that a suboptimal AMP algorithm still achieves almost-exact recovery for any $\delta >0$. The suboptimal algorithm  applies the SC-AMP algorithm column-wise to  $\tY \in \reals^{n \times L}$, ignoring the correlation between the columns of the signal matrix $B$.

\subsection{Almost-Exact Recovery via Column-wise SC-AMP}

Given $\tXsc, \tY$ from the rescaled model \eqref{eq:rescaled_pooled_data}, we run the SC-AMP algorithm \emph{column-wise} on $\tY$. Specifically, for $l \in [L]$, we run the SC-AMP algorithm \eqref{eq:SC_AMP}  with inputs $Y_{:,l}$ and $\tXsc$ to produce the estimate $\hB^k_{:,l}$ after $k \ge 1$ iterations.  For the SC-AMP algorithm applied to column $l \in [L]$, the denoiser $f_k$ in \eqref{eq:optimal_fk} is computed with $\bar{\beta} \sim \text{Bernoulli}(\pi_l)$. (We recall that the rows of the signal follow the prior $\bar{B}\sim \text{Categorical}(\pi)$ where $\pi=(\pi_1, \ldots, \pi_L)$.) 

The column-wise SC-AMP algorithm can be viewed as an instance of the matrix SC-AMP algorithm with a suboptimal denoiser, obtained by replacing the conditional expectation $\E[\bar{B} \mid  \bar{B} +  G^k_\sfc ]$ in \eqref{eq:optimal_fk_pooled_data} with  the \emph{marginal} conditional expectations $\E[\bar{B}_l \mid  \bar{B}_l +  (G^k_\sfc)_l ]$, for $l \in [L]$.
Let us define the quantized estimate after $k$ iterations of the column-wise SC-AMP algorithm to be
\begin{align}
    \tB^k_{jl}
    =
    \begin{cases}
        1 &\text{ if $\hB^k_{jl}>0.5$,} \\
        0 &\text{ otherwise,}
    \end{cases}
    \label{eq:quant_col_SCAMP}
\end{align}
where $\hB^k$ is the estimate obtained from the SC-AMP algorithm. In practice, we can quantize the estimate in a better manner, by setting the largest entry in the row of $\hB$ to one and the remaining entries in the row to zero. We do not use this form of quantization for our almost-exact recovery result since we want to directly apply the SC-AMP results for QGT  to the pooled data setting. 

\begin{theorem} \label{thm:achievability_pooled_data}
Consider the noiseless pooled data problem with the assumptions stated on p.\pageref{assump:model_noise_pooled_data}, for any $\delta >0$. There exist finite $\omega_0$ and $k_0$ such that for all $\omega>\omega_0$, $k>k_0$, and sufficiently large $\Lambda$, the quantized estimate $\tB^k \in \{ 0,1\}^{p \times L}$ produced by the column-wise SC-AMP algorithm almost surely satisfies:
$$ \lim_{k \to \infty}\lim_{p\rightarrow\infty} 
    \frac{1}{p}\sum_{j=1}^p\mathds{1}\big\{\tB^k_{j,:}\neq B_{j,:} \big\} =0. $$
\end{theorem}

\begin{proof}
The model assumptions imply that for each $l \in [L]$, the empirical distribution of column $B_{:,l}$ converges in Wasserstein distance to $\text{Bernoulli}(\pi_l)$.  By Corollary \ref{cor:SC_AMP_almost_exact}, we have that the SC-AMP algorithm applied to $\tY_{:,l}$ satisfies $ \lim_{k \to \infty} \lim_{p\rightarrow\infty} 
    \frac{1}{p}\sum_{j=1}^p\mathds{1}\big\{\tB^k_{j,l}\neq B_{j,l} \big\} =0$ almost surely, for each $l \in [L]$. 
    The result follows by noting that $\frac{1}{p}\sum_{j=1}^p\mathds{1}\big\{\tB^k_{j,:}\neq B_{j,:} \big\} \le \frac{1}{p}\sum_{l=1}^L\sum_{j=1}^p\mathds{1}\big\{\tB^k_{jl}\neq B_{jl} \big\}$.
\end{proof}

 We can also obtain error guarantees in the low-noise regime for the column-wise SC-AMP algorithm, similar to Theorem \ref{thm:achievability} and Corollary \ref{cor:SC_AMP_almost_exact}. We remark that although column-wise SC-AMP is convenient for theoretical analysis, at finite dimensions it is inferior to the matrix SC-AMP algorithm that takes advantage of the correlation in the columns of $B$ via the denoiser in \eqref{eq:optimal_fk_pooled_data}. This is illustrated in the numerical experiments below.

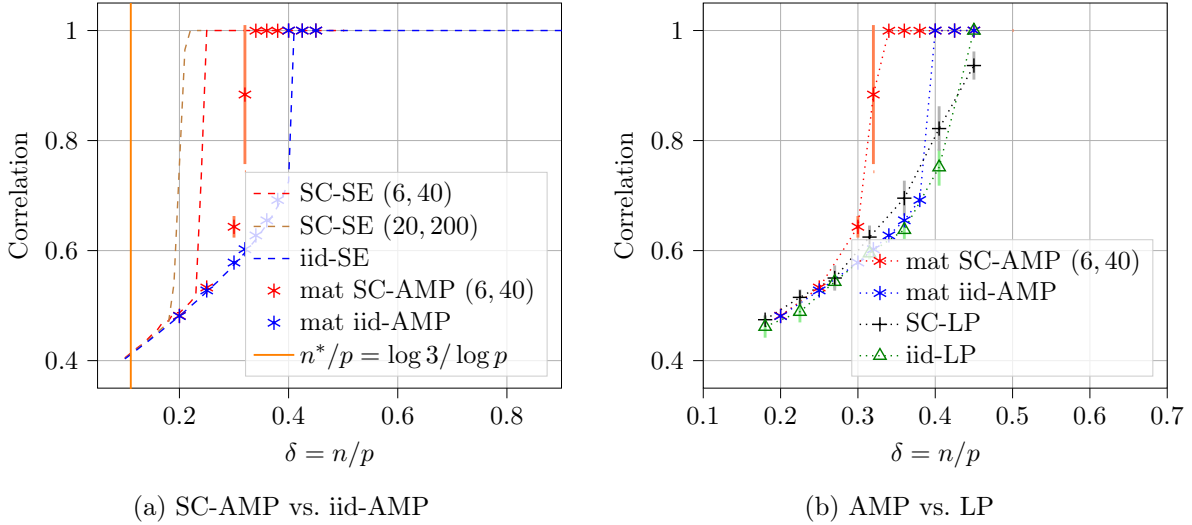
\begin{figure}[t]
\begin{subfigure}[b]{0.49\textwidth}
  \centering
  % This file was created with tikzplotlib v0.10.1.

% This file was created with tikzplotlib v0.10.1.
\begin{tikzpicture}[scale=0.9]

\definecolor{coral}{RGB}{255,127,80}
\definecolor{darkgray176}{RGB}{176,176,176}
\definecolor{lightgray204}{RGB}{204,204,204}
\definecolor{lightblue}{RGB}{173,216,230}
\definecolor{green}{RGB}{0,128,0}
\definecolor{lightgreen}{RGB}{144,238,144}
\definecolor{mistyrose}{RGB}{255,228,225}
\definecolor{pink}{RGB}{255,192,203}

\begin{axis}[
legend cell align={left},
legend style={
  fill opacity=0.8,
  draw opacity=1,
  text opacity=1,
  at={(0.97,0.03)},
  anchor=south east,
  draw=lightgray204
},
tick align=outside,
tick pos=left,
x grid style={darkgray176},
xlabel={\(\displaystyle \delta=n/p\)},
xmajorgrids,
xmin=0.05, xmax=0.9,
xtick style={color=black},
y grid style={darkgray176},
ylabel={Correlation},
ymajorgrids,
ymin=0.35, ymax=1.05,
ytick style={color=black}
]
\path [draw=coral, very thick]
(axis cs:0.2,0.480760992916505)
--(axis cs:0.2,0.480760992916505);

\path [draw=coral, very thick]
(axis cs:0.25,0.530789642716406)
--(axis cs:0.25,0.530789642716406);

\path [draw=coral, very thick]
(axis cs:0.3,0.633440536332498)
--(axis cs:0.3,0.633440536332498);

\path [draw=coral, very thick]
(axis cs:0.32,0.743279752590953)
--(axis cs:0.32,0.743279752590953);

\path [draw=coral, very thick]
(axis cs:0.34,0.999700000000044)
--(axis cs:0.34,0.999700000000044);

\path [draw=coral, very thick]
(axis cs:0.36,0.99965)
--(axis cs:0.36,0.99965);

\path [draw=coral, very thick]
(axis cs:0.38,0.99975)
--(axis cs:0.38,0.99975);

\path [draw=coral, very thick]
(axis cs:0.4,0.9998)
--(axis cs:0.4,0.9998);

\path [draw=coral, very thick]
(axis cs:0.45,0.99995)
--(axis cs:0.45,0.99995);

\path [draw=coral, very thick]
(axis cs:0.5,1)
--(axis cs:0.5,1);

\addplot[ultra thick, red, dashed]
table {%
0.1 0.4049633817718936
0.11 0.4129859349431553
0.12 0.42123632932737776
0.13 0.42842201494731585
0.14 0.437413142051631
0.15 0.44750841888953313
0.16 0.4566121441153994
0.17 0.47049119507144715
0.18 0.4722281476812286
0.19 0.48773125401092976
0.2 0.4868048764502625
0.21 0.4963484718489456
0.22 0.5062344637845697
0.23 0.5182319303314461
0.24 0.7342149414818087
0.25 1
0.26 1
0.27 1
0.28 1
0.29 1
0.3 1
0.31 1
0.32 1
0.33 1
0.34 1
0.35 1
0.36 1
0.37 1
0.38 1
0.39 1
0.4 1
0.41 1
0.42 1
0.43 1
0.44 1
0.45 1
0.46 1
0.47 1
0.48 1
0.49 1
0.5 1
};
\addlegendentry{SC-SE ($6,40$)}

\addplot[ultra thick, brown, dashdotted]
table {%
0.1 0.403856100483508
0.15 0.443153898697603
0.18 0.469367073526288
0.19 0.539720789625208
0.2 0.755826296389276
0.21 0.965055376316642
0.22 1
0.23 1
0.24 1
};
\addlegendentry{SC-SE ($20,200$)}

\addplot[ultra thick, cyan!65, dotted]
table {%
0.1 0.403230617663388
0.11 0.41054217791051
0.12 0.418072646119556
0.13 0.42586601711218
0.14 0.433049684112449
0.15 0.441068112897524
0.16 0.448677486355195
0.17 0.456796115890354
0.18 0.464965565224308
0.19 0.473060411748787
0.2 0.481383952640562
0.21 0.489768081710178
0.22 0.49849278309212
0.23 0.507063825756621
0.24 0.515892837028079
0.25 0.52529028713855
0.26 0.53471322546817
0.27 0.543685117125823
0.28 0.553979541774879
0.29 0.563521208474632
0.3 0.574534881337513
0.31 0.584950335194167
0.32 0.596399280594805
0.33 0.60824276728242
0.34 0.620328597456081
0.35 0.633617102706393
0.36 0.647536846843159
0.37 0.663010132666371
0.38 0.680425387426825
0.39 0.700418484800851
0.4 0.72536240571466
0.41 1
0.42 1
0.43 1
0.44 1
0.45 1
0.46 1
0.47 1
0.48 1
0.49 1
0.5 1
0.51 1
0.52 1
0.53 1
0.54 1
0.55 1
0.56 1
0.57 1
0.58 1
0.59 1
0.6 1
0.61 1
0.62 1
0.63 1
0.64 1
0.65 1
0.66 1
0.67 1
0.68 1
0.69 1
0.7 1
0.71 1
0.72 1
0.73 1
0.74 1
0.75 1
0.76 1
0.77 1
0.78 1
0.79 1
0.8 1
0.81 1
0.82 1
0.83 1
0.84 1
0.85 1
0.86 1
0.87 1
0.88 1
0.89 1
0.9 1
0.91 1
0.92 1
0.93 1
0.94 1
0.95 1
0.96 1
0.97 1
0.98 1
0.99 1
1 1
};
\addlegendentry{iid-SE}

\path [draw=lightblue, very thick]
(axis cs:0.2,0.478691923185454)
--(axis cs:0.2,0.483238842806353);

\path [draw=lightblue, very thick]
(axis cs:0.25,0.5251164181721)
--(axis cs:0.25,0.529335011700271);

\path [draw=lightblue, very thick]
(axis cs:0.3,0.57474270547671)
--(axis cs:0.3,0.581984446964378);

\path [draw=lightblue, very thick]
(axis cs:0.32,0.595904189751598)
--(axis cs:0.32,0.603862431900766);

\path [draw=lightblue, very thick]
(axis cs:0.34,0.625517409861708)
--(axis cs:0.34,0.633116091837664);

\path [draw=lightblue, very thick]
(axis cs:0.36,0.651847269489094)
--(axis cs:0.36,0.661231157242804);

\path [draw=lightblue, very thick]
(axis cs:0.38,0.686192923218183)
--(axis cs:0.38,0.70482575459105);

\path [draw=lightblue, very thick]
(axis cs:0.4,1)
--(axis cs:0.4,1);

\path [draw=lightblue, very thick]
(axis cs:0.425,1)
--(axis cs:0.425,1);

\path [draw=lightblue, very thick]
(axis cs:0.45,1)
--(axis cs:0.45,1);

\path [draw=coral, very thick]
(axis cs:0.2,0.472970970524797)
--(axis cs:0.2,0.478802536971305);

\path [draw=coral, very thick]
(axis cs:0.25,0.515058009130187)
--(axis cs:0.25,0.523030540670444);

\path [draw=coral, very thick]
(axis cs:0.3,0.612944644558127)
--(axis cs:0.3,0.667647921553092);

\path [draw=coral, very thick]
(axis cs:0.32,0.829695927598342)
--(axis cs:0.32,1.04492804393565);

\path [draw=coral, very thick]
(axis cs:0.34,0.99891711607325)
--(axis cs:0.34,0.999423260733696);

\path [draw=coral, very thick]
(axis cs:0.36,0.999231466049773)
--(axis cs:0.36,0.999608720315422);

\path [draw=coral, very thick]
(axis cs:0.38,0.999359055023196)
--(axis cs:0.38,0.999741064539586);

\path [draw=coral, very thick]
(axis cs:0.4,0.999426466143209)
--(axis cs:0.4,0.999733633859042);

\path [draw=coral, very thick]
(axis cs:0.425,0.999614026357751)
--(axis cs:0.425,0.999886026198474);

\path [draw=coral, very thick]
(axis cs:0.45,0.9997519137754)
--(axis cs:0.45,0.9999081037246);

\addplot[ultra thick, red, mark=star, mark size=3, mark options={solid}, only marks]
table {%
0.2 0.475886753748051
0.25 0.519044274900315
0.3 0.640296283055609
0.32 0.937311985766997
0.34 0.999170188403473
0.36 0.999420093182597
0.38 0.999550059781391
0.4 0.999580050001125
0.425 0.999750026278112
0.45 0.99983000875
};
\addlegendentry{mat SC-AMP ($6,40$)}
\addplot[ultra thick, cyan!65, mark=square, mark size=2, mark options={solid}, only marks]
table {%
0.2 0.480965382995904
0.25 0.527225714936186
0.3 0.578363576220544
0.32 0.599883310826182
0.34 0.629316750849686
0.36 0.656539213365949
0.38 0.695509338904617
0.4 1
0.425 1
0.45 1
};
\addlegendentry{mat iid-AMP}

\addplot[ultra thick, orange, solid]
table {%
0.1109 0
0.1109 0.4
0.1109 0.6
0.1109 0.8
0.1109 1
0.1109 1.4
};
\addlegendentry{$n^*/p = \log 3 / \log p$}
\end{axis}

\end{tikzpicture}
  \vspace{-2\baselineskip}
  \caption{SC-AMP vs.~iid-AMP}
  \label{fig:pool_sc_amp}
\end{subfigure}
\begin{subfigure}[b]{0.49\textwidth}
  \centering
  % This file was created with tikzplotlib v0.10.1.

% This file was created with tikzplotlib v0.10.1.
\begin{tikzpicture}[scale=0.9]

\definecolor{coral}{RGB}{255,127,80}
\definecolor{darkgray176}{RGB}{176,176,176}
\definecolor{lightgray204}{RGB}{204,204,204}
\definecolor{lightblue}{RGB}{173,216,230}
\definecolor{green}{RGB}{0,128,0}
\definecolor{lightgreen}{RGB}{144,238,144}
\definecolor{black!30!white}{RGB}{255,228,225}
\definecolor{pink}{RGB}{255,192,203}

\begin{axis}[
legend cell align={left},
legend style={
  fill opacity=0.8,
  draw opacity=1,
  text opacity=1,
  at={(0.97,0.03)},
  anchor=south east,
  draw=lightgray204
},
tick align=outside,
tick pos=left,
x grid style={darkgray176},
xlabel={\(\displaystyle \delta=n/p\)},
xmajorgrids,
xmin=0.1, xmax=0.7,
xtick style={color=black},
y grid style={darkgray176},
ylabel={Correlation},
ymajorgrids,
ymin=0.35, ymax=1.05,
ytick style={color=black}
]
\path [draw=coral, very thick]
(axis cs:0.2,0.480760992916505)
--(axis cs:0.2,0.480760992916505);

\path [draw=coral, very thick]
(axis cs:0.25,0.530789642716406)
--(axis cs:0.25,0.530789642716406);

\path [draw=coral, very thick]
(axis cs:0.3,0.633440536332498)
--(axis cs:0.3,0.633440536332498);

\path [draw=coral, very thick]
(axis cs:0.32,0.743279752590953)
--(axis cs:0.32,0.743279752590953);

\path [draw=coral, very thick]
(axis cs:0.34,0.999700000000044)
--(axis cs:0.34,0.999700000000044);

\path [draw=coral, very thick]
(axis cs:0.36,0.99965)
--(axis cs:0.36,0.99965);

\path [draw=coral, very thick]
(axis cs:0.38,0.99975)
--(axis cs:0.38,0.99975);

\path [draw=coral, very thick]
(axis cs:0.4,0.9998)
--(axis cs:0.4,0.9998);

\path [draw=coral, very thick]
(axis cs:0.45,0.99995)
--(axis cs:0.45,0.99995);

\path [draw=coral, very thick]
(axis cs:0.5,1)
--(axis cs:0.5,1);

\path [draw=lightblue, very thick]
(axis cs:0.2,0.478691923185454)
--(axis cs:0.2,0.483238842806353);

\path [draw=lightblue, very thick]
(axis cs:0.25,0.5251164181721)
--(axis cs:0.25,0.529335011700271);

\path [draw=lightblue, very thick]
(axis cs:0.3,0.57474270547671)
--(axis cs:0.3,0.581984446964378);

\path [draw=lightblue, very thick]
(axis cs:0.32,0.595904189751598)
--(axis cs:0.32,0.603862431900766);

\path [draw=lightblue, very thick]
(axis cs:0.34,0.625517409861708)
--(axis cs:0.34,0.633116091837664);

\path [draw=lightblue, very thick]
(axis cs:0.36,0.651847269489094)
--(axis cs:0.36,0.661231157242804);

\path [draw=lightblue, very thick]
(axis cs:0.38,0.686192923218183)
--(axis cs:0.38,0.70482575459105);

\path [draw=lightblue, very thick]
(axis cs:0.4,1)
--(axis cs:0.4,1);

\path [draw=lightblue, very thick]
(axis cs:0.425,1)
--(axis cs:0.425,1);

\path [draw=lightblue, very thick]
(axis cs:0.45,1)
--(axis cs:0.45,1);

\path [draw=coral, very thick]
(axis cs:0.2,0.472970970524797)
--(axis cs:0.2,0.478802536971305);

\path [draw=coral, very thick]
(axis cs:0.25,0.515058009130187)
--(axis cs:0.25,0.523030540670444);

\path [draw=coral, very thick]
(axis cs:0.3,0.612944644558127)
--(axis cs:0.3,0.667647921553092);

\path [draw=coral, very thick]
(axis cs:0.32,0.829695927598342)
--(axis cs:0.32,1.04492804393565);

\path [draw=coral, very thick]
(axis cs:0.34,0.99891711607325)
--(axis cs:0.34,0.999423260733696);

\path [draw=coral, very thick]
(axis cs:0.36,0.999231466049773)
--(axis cs:0.36,0.999608720315422);

\path [draw=coral, very thick]
(axis cs:0.38,0.999359055023196)
--(axis cs:0.38,0.999741064539586);

\path [draw=coral, very thick]
(axis cs:0.4,0.999426466143209)
--(axis cs:0.4,0.999733633859042);

\path [draw=coral, very thick]
(axis cs:0.425,0.999614026357751)
--(axis cs:0.425,0.999886026198474);

\path [draw=coral, very thick]
(axis cs:0.45,0.9997519137754)
--(axis cs:0.45,0.9999081037246);

\addplot [very thick, red, loosely dotted, mark=star, mark size=3, mark options={solid}]
table {%
0.2 0.475886753748051
0.25 0.519044274900315
0.3 0.640296283055609
0.32 0.937311985766997
0.34 0.999170188403473
0.36 0.999420093182597
0.38 0.999550059781391
0.4 0.999580050001125
0.425 0.999750026278112
0.45 0.99983000875
0.495 1
0.54 1
0.585 1
0.675 1
0.765 1
};
\addlegendentry{mat SC-AMP ($6,40$)}
\addplot [very thick, cyan!65, loosely dotted, mark=square, mark size=2, mark options={solid}]
table {%
0.2 0.480965382995904
0.25 0.527225714936186
0.3 0.578363576220544
0.32 0.599883310826182
0.34 0.629316750849686
0.36 0.656539213365949
0.38 0.695509338904617
0.4 1
0.425 1
0.45 1
0.495 1
0.54 1
0.585 1
0.675 1
0.765 1
};
\addlegendentry{mat iid-AMP}

\path [draw=lightgreen, very thick]
(axis cs:0.18,0.231257711653474)
--(axis cs:0.18,0.261540181734323);

\path [draw=lightgreen, very thick]
(axis cs:0.225,0.274291409375539)
--(axis cs:0.225,0.294382109246665);

\path [draw=lightgreen, very thick]
(axis cs:0.27,0.32822532245415)
--(axis cs:0.27,0.369966733137866);

\path [draw=lightgreen, very thick]
(axis cs:0.315,0.405258204550386)
--(axis cs:0.315,0.486537193316153);

\path [draw=lightgreen, very thick]
(axis cs:0.36,0.527400634315624)
--(axis cs:0.36,0.592625305070334);

\path [draw=lightgreen, very thick]
(axis cs:0.405,0.668524103357221)
--(axis cs:0.405,0.904874862434365);

\path [draw=lightgreen, very thick]
(axis cs:0.45,1)
--(axis cs:0.45,1);

\path [draw=black!30!white, very thick]
(axis cs:0.18,0.2234590073385)
--(axis cs:0.18,0.26037130293777);

\path [draw=black!30!white, very thick]
(axis cs:0.225,0.286985017727758)
--(axis cs:0.225,0.31490682496333);

\path [draw=black!30!white, very thick]
(axis cs:0.27,0.33158180703887)
--(axis cs:0.27,0.38015869258407);

\path [draw=black!30!white, very thick]
(axis cs:0.315,0.404090318504366)
--(axis cs:0.315,0.473093918754971);

\path [draw=black!30!white, very thick]
(axis cs:0.36,0.517596183431265)
--(axis cs:0.36,0.597759418044975);

\path [draw=black!30!white, very thick]
(axis cs:0.405,0.649889453268737)
--(axis cs:0.405,0.750404184964443);

\path [draw=black!30!white, very thick]
(axis cs:0.45,0.86287131812601)
--(axis cs:0.45,0.933111402323307);

\path [draw=black!30!white, very thick]
(axis cs:0.495,0.935442413824821)
--(axis cs:0.495,0.96748259987872);

\path [draw=black!30!white, very thick]
(axis cs:0.54,0.955954143572963)
--(axis cs:0.54,0.980240456427148);

\path [draw=black!30!white, very thick]
(axis cs:0.585,0.969848338218352)
--(axis cs:0.585,0.988012351962036);

\path [draw=black!30!white, very thick]
(axis cs:0.675,0.987659585838038)
--(axis cs:0.675,0.999578814162034);

\path [draw=black!30!white, very thick]
(axis cs:0.765,0.993948688490443)
--(axis cs:0.765,1.0012609115082);

\addplot [very thick, black, loosely dotted, mark=+, mark size=3, mark options={solid}]
table {%
0.18 0.241915155138135
0.225 0.300945921345544
0.27 0.35587024981147
0.315 0.438592118629668
0.36 0.55767780073812
0.405 0.70014681911659
0.45 0.897991360224659
0.495 0.9514625068517703
0.54 0.9680973000000558
0.585 0.978930345090194
0.675 0.993619200000036
0.765 0.997604799999323
};
\addlegendentry{SC-LP}
\addplot [very thick, green, loosely dotted, mark=triangle, mark size=3, mark options={solid}]
table {%
0.18 0.246398946693899
0.225 0.284336759311102
0.27 0.349096027796008
0.315 0.44589769893327
0.36 0.560012969692979
0.405 0.786699482895793
0.45 1
0.495 1
0.54 1
0.585 1
0.675 1
0.765 1
};
\addlegendentry{iid-LP}

% \addplot [semithick, green, dotted, mark=asterisk, mark size=3, mark options={solid}]
% table {%
% 0.1 0.403580348073018
% 0.2 0.48684459934697
% 0.3 0.5838730024653
% 0.4 0.772484619382269
% 0.425 0.9441440887308684
% 0.45 0.9999999971120569
% 0.5 0.999999999249688
% 0.6 0.999999999882812
% 0.7 1.00000000000334
% 0.8 1.00000000003197
% 0.9 0.999999999976076
% 1 0.999999999998866
% };
% \addlegendentry{iid-LP}

% \addplot [semithick, black, dotted, mark=asterisk, mark size=3, mark options={solid}]
% table {%
% 0.1 0.3201
% 0.2 0.34158
% 0.3 0.36752
% 0.4 0.38642
% 0.425 0.3935
% 0.45 0.4
% 0.5 0.41212
% 0.6 0.43654
% 0.7 0.46518
% 0.8 0.48968
% 0.9 0.52242
% 1 0.55392
% };
% \addlegendentry{i.i.d.~IHT}

\end{axis}

\end{tikzpicture}
  \vspace{-2\baselineskip}
  \caption{AMP vs.~LP}
  \label{fig:pool_sc_amp_lp}
\end{subfigure}
\caption{Noiseless Pooled Data with $\pi=[1/3, 1/3, 1/3]$ with $p=20000$ and spatial coupling parameters $\omega=6, \Lambda=40$. In (b), we set $p=1000$ for LP due to the high computational cost. Error bars indicate one standard deviation.}
\label{fig:corr_v_delta_pool}
\end{figure}

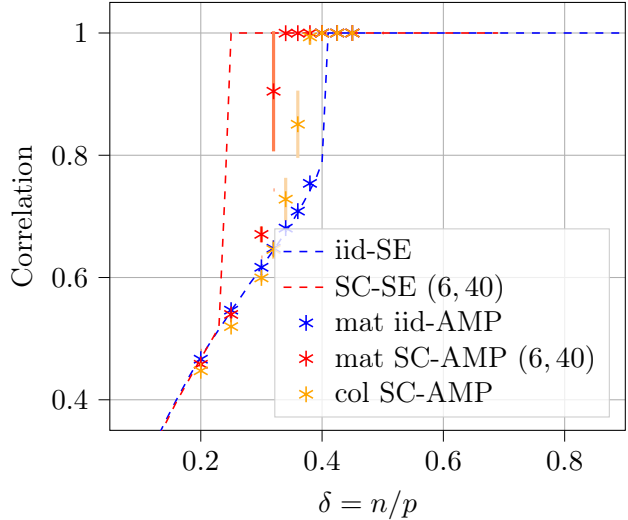
\begin{figure}[t]
    \centering
    % This file was created with tikzplotlib v0.10.1.

% This file was created with tikzplotlib v0.10.1.
\begin{tikzpicture}

\definecolor{coral}{RGB}{255,127,80}
\definecolor{darkgray176}{RGB}{176,176,176}
\definecolor{lightgray204}{RGB}{204,204,204}
\definecolor{lightblue}{RGB}{173,216,230}
\definecolor{green}{RGB}{0,128,0}
\definecolor{lightgreen}{RGB}{144,238,144}
\definecolor{orange}{RGB}{255, 165, 0}
\definecolor{desert}{RGB}{250, 213, 165}

\begin{axis}[
legend cell align={left},
legend style={
  fill opacity=0.8,
  draw opacity=1,
  text opacity=1,
  at={(1.09,0.03)},
  anchor=south east,
  draw=lightgray204
},
tick align=outside,
tick pos=left,
x grid style={darkgray176},
xlabel={\(\displaystyle \delta=n/p\)},
xmajorgrids,
xmin=0.05, xmax=0.9,
xtick style={color=black},
y grid style={darkgray176},
ylabel={Correlation},
ymajorgrids,
ymin=0.1, ymax=1.05,
ytick style={color=black}
]
\addplot[ultra thick, red, dashed]
table {%
% 0.1 0.403796954450652
% 0.11 0.411521598728357
% 0.12 0.41903314890413
% 0.13 0.42677416022468
% 0.14 0.434860835065784
% 0.15 0.442973864489119
% 0.16 0.451301762136113
% 0.17 0.45987962633433
% 0.18 0.468208702408787
% 0.19 0.477222717895803
% 0.2 0.486255051962064
% 0.21 0.496110024472947
% 0.22 0.506549821527592
% 0.23 0.517512278858722
% 0.24 0.734279102868909
% 0.25 1
% 0.26 1
% 0.27 0.999999999999957
% 0.28 1
% 0.29 0.999999999999303
% 0.3 1
% 0.31 1
% 0.32 1
% 0.33 1
% 0.34 1
% 0.35 1
% 0.36 1
% 0.37 1
% 0.38 1
% 0.39 1
% 0.4 1
% 0.41 1
% 0.1 0.40395394139854
% 0.11 0.411362969717796
% 0.12 0.419305727504372
% 0.13 0.427094548580149
% 0.14 0.435096552971405
% 0.15 0.442953508745975
% 0.16 0.451098522095666
% 0.17 0.459643894191328
% 0.18 0.468332561524542
% 0.19 0.477395591553174
% 0.2 0.486345694847277
% 0.21 0.496116682936664
% 0.22 0.506334646179962
% 0.23 0.517650486952499
% 0.24 0.735306461868504
% 0.25 1
% 0.26 1
0.05 0.0661092553963631
0.06 0.0859434983228297
0.07 0.105799310347291
0.08 0.124494376871714
0.09 0.142825659156404
0.1 0.159299714498344
0.11 0.176567244612415
0.12 0.192937418228568
0.13 0.208033842325784
0.14 0.223972085242265
0.15 0.239075780188181
0.16 0.254144144732523
0.17 0.26899704926906
0.18 0.283954701135605
0.19 0.298721522460622
0.2 0.3139861173659
0.21 0.328453723685228
0.22 0.343642633436012
0.23 0.360264381631681
0.24 0.518248434577549
0.25 1
0.26 1
0.27 1
0.28 1
0.29 1
0.3 1
0.31 1
0.32 1
0.33 1
0.34 1
0.35 1
0.36 1
0.37 1
0.38 1
0.39 1
0.4 1
0.41 1
0.42 1
0.43 1
0.44 1
0.45 1
0.46 1
0.47 1
0.48 1
0.49 1
0.6 1
0.8 1
};
\addlegendentry{SC-SE ($6,40$)}

\addplot[ultra thick, cyan!65, dotted]
table {%
0.05 0.0673497881069961
0.06 0.0875443607891912
0.07 0.10786412859505
0.08 0.126671101766083
0.09 0.144211916712801
0.1 0.162802319775954
0.11 0.179071532918781
0.12 0.19430941498553
0.13 0.210246392103372
0.14 0.229055682751317
0.15 0.239864172917209
0.16 0.255817683030707
0.17 0.269566111223976
0.18 0.282833741181952
0.19 0.298308815969174
0.2 0.317978868922099
0.21 0.328031139719319
0.22 0.349178050437286
0.23 0.355954573587355
0.24 0.368044313068871
0.25 0.383282110175724
0.26 0.40276327892846
0.27 0.411563491702293
0.28 0.424012389206652
0.29 0.438028687089313
0.3 0.450483205755408
0.31 0.470428802099874
0.32 0.492626255292435
0.33 0.500347284582075
0.34 0.519811197170258
0.35 0.535699587385178
0.36 0.565497257831324
0.37 0.583030005067494
0.38 0.606403787784619
0.39 0.617885522637798
0.4 0.656413384581079
0.41 1
0.42 1
0.43 1
0.44 1
0.45 1
0.46 1
0.47 1
0.48 1
0.49 1
0.6 1
0.8 1
};
\addlegendentry{iid-SE}

\addplot [very thick, red, mark=star, mark size=3, mark options={solid}, loosely dotted]
table {%
0.2 0.296612668917556
0.25 0.368227709058095
0.3 0.509228474486878
0.32 0.825283012300509
0.34 0.99933012225
0.36 0.9989607635
0.38 0.99945008325
0.4 0.9995600535
0.425 0.99971002675
0.45 0.999880006
};
\addlegendentry{mat SC-AMP ($6,40$)}

\addplot [very thick, cyan!65, mark=square, mark size=2, mark options={solid}, only marks]
table {%
0.2 0.307864323935568
0.25 0.380858079571817
0.3 0.453469862907793
0.32 0.482573797313844
0.34 0.520184179833109
0.36 0.550142063416918
0.38 0.60896341765287
0.4 1
0.425 1
0.45 1
};
\addlegendentry{mat iid-AMP}

\path [draw=lightblue, very thick]
(axis cs:0.2,0.304762901204977)
--(axis cs:0.2,0.310965746666158);

\path [draw=lightblue, very thick]
(axis cs:0.25,0.377199203183082)
--(axis cs:0.25,0.384516955960553);

\path [draw=lightblue, very thick]
(axis cs:0.3,0.448466671967128)
--(axis cs:0.3,0.458473053848458);

\path [draw=lightblue, very thick]
(axis cs:0.32,0.479575685116011)
--(axis cs:0.32,0.485571909511677);

\path [draw=lightblue, very thick]
(axis cs:0.34,0.514426816522676)
--(axis cs:0.34,0.525941543143542);

\path [draw=lightblue, very thick]
(axis cs:0.36,0.545188120817762)
--(axis cs:0.36,0.555096006016074);

\path [draw=lightblue, very thick]
(axis cs:0.38,0.598844880131249)
--(axis cs:0.38,0.619081955174492);

\path [draw=lightblue, very thick]
(axis cs:0.4,1)
--(axis cs:0.4,1);

\path [draw=lightblue, very thick]
(axis cs:0.425,1)
--(axis cs:0.425,1);

\path [draw=lightblue, very thick]
(axis cs:0.45,1)
--(axis cs:0.45,1);

\path [draw=coral, very thick]
(axis cs:0.2,0.294001945119416)
--(axis cs:0.2,0.299223392715696);

\path [draw=coral, very thick]
(axis cs:0.25,0.362050130494585)
--(axis cs:0.25,0.374405287621606);

\path [draw=coral, very thick]
(axis cs:0.3,0.47451812313655)
--(axis cs:0.3,0.543938825837206);

\path [draw=coral, very thick]
(axis cs:0.32,0.63863314844464)
--(axis cs:0.32,1.01193287615638);

\path [draw=coral, very thick]
(axis cs:0.34,0.999129947399586)
--(axis cs:0.34,0.999530297100413);

\path [draw=coral, very thick]
(axis cs:0.36,0.997558324275471)
--(axis cs:0.36,1.00036320272453);

\path [draw=coral, very thick]
(axis cs:0.38,0.999275488784588)
--(axis cs:0.38,0.999624677715412);

\path [draw=coral, very thick]
(axis cs:0.4,0.999417256857702)
--(axis cs:0.4,0.999702850142298);

\path [draw=coral, very thick]
(axis cs:0.425,0.999558719230561)
--(axis cs:0.425,0.999861334269439);

\path [draw=coral, very thick]
(axis cs:0.45,0.999782032840173)
--(axis cs:0.45,0.999977979159827);

\path [draw=desert, very thick]
(axis cs:0.2,0.288505097272814)
--(axis cs:0.2,0.292767525675312);

\path [draw=desert, very thick]
(axis cs:0.25,0.350340953113248)
--(axis cs:0.25,0.36297514797703);

\path [draw=desert, very thick]
(axis cs:0.3,0.426348341521294)
--(axis cs:0.3,0.44641925675236);

\path [draw=desert, very thick]
(axis cs:0.32,0.469645858867019)
--(axis cs:0.32,0.505668111419176);

\path [draw=desert, very thick]
(axis cs:0.34,0.567656672954291)
--(axis cs:0.34,0.639902644207185);

\path [draw=desert, very thick]
(axis cs:0.36,0.742448809721177)
--(axis cs:0.36,0.890416546683657);

\path [draw=desert, very thick]
(axis cs:0.38,0.99958164374339)
--(axis cs:0.38,0.999828399756561);

\path [draw=desert, very thick]
(axis cs:0.4,0.999766522519404)
--(axis cs:0.4,0.999913487480721);

\path [draw=desert, very thick]
(axis cs:0.425,0.999675366823345)
--(axis cs:0.425,0.99989466267558);

\path [draw=desert, very thick]
(axis cs:0.45,0.999855634363275)
--(axis cs:0.45,0.9999743716366);

\addplot [very thick, orange, mark=x, mark size=3, mark options={solid}, loosely dotted]
table {%
0.2 0.290636311474063
0.25 0.356658050545139
0.3 0.436383799136827
0.32 0.487656985143098
0.34 0.603779658580738
0.36 0.816432678202417
0.38 0.999705021749975
0.4 0.999840005000062
0.425 0.999785014749463
0.45 0.999915002999938
};
\addlegendentry{col SC-AMP }
\end{axis}
\end{tikzpicture}
    \caption{Matrix SC-AMP vs column-wise SC-AMP for noiseless pooled data with $\pi=[1/3, 1/3, 1/3]$, $p=20000$ and spatial coupling parameters $\omega=6, \Lambda=40$.}
    \label{fig:pool_iterative-sc_amp}
\end{figure}

\subsection{Numerical Simulations}
We present simulation results for noiseless pooled data using matrix SC-AMP, measuring the performance  via the normalized squared correlation. The normalized squared correlation of the AMP estimate after $k \ge 1$ iterations and its state evolution prediction are given by:
\begin{align}
    \frac{\big(\frac{1}{p}\sum_{j=1}^p\langle \hB^k_{j,:},B_{j,:}\rangle\big)^2}{\big(\frac{1}{p}\sum_{j=1}^p\|\hB^k_{j,:}\|^2\big)\cdot\big(\frac{1}{p}\sum_{j=1}^p \|B_{j,:}\|^2\big)} 
\,     \stackrel{p \to \infty}{\longrightarrow} \, 
    \frac{\left(\frac{1}{\sfC}\sum_{\sfc=1}^\sfC\E\Big[\langle f_k(\bar{B}+G_\sfc^k,\sfc),\Bar{B}\rangle\Big]\right)^2}{\left(\frac{1}{\sfC}\sum_{\sfc=1}^\sfC\E\left[\|f_k(\bar{B}+G_\sfc^k,\sfc)\|^2_2\right]\right) \cdot \E\left[\|\Bar{B}\|_2^2\right]},
\end{align}
where the almost sure convergence to the state evolution prediction on  the right follows from \eqref{eq:matrix_SC_AMP_convergence}.

    % \frac{\Tr\left((\hB^k)^\top B\right)^2}{\|\hB^k\|_F^2\cdot\|B\|_F^2}
    % =\frac{\big(\frac{1}{p}\sum_{j=1}^p\langle \hB^k_{j,:},B_{j,:}\rangle\big)^2}{\big(\frac{1}{p}\sum_{j=1}^p\|\hB^k_{j,:}\|^2\big)\cdot\big(\frac{1}{p}\sum_{j=1}^p \|B_{j,:}\|^2\big)} \nonumber
% \pc{Worth including LHS expression? Updated performance measure and convergence below} after $k$ iterations of the algorithm, where $\Tr(\cdot)$ corresponds to the trace matrix operator. , this normalized squared correlation converges as 
% \begin{align*}
%     \frac{\big(\frac{1}{p}\sum_{j=1}^p\langle \hB^k_{j,:},B_{j,:}\rangle\big)^2}{\big(\frac{1}{p}\sum_{j=1}^p\|\hB^k_{j,:}\|^2\big)\cdot\big(\frac{1}{p}\sum_{j=1}^p \|B_{j,:}\|^2\big)} 
% \,     \stackrel{a.s.}{\rightarrow} \, 
%     \frac{\left(\frac{1}{\sfC}\sum_{\sfc=1}^\sfC\E\Big[\langle f_k(\bar{B}+G_\sfc^k,\sfc),\Bar{B}\rangle\Big]\right)^2}{\left(\frac{1}{\sfC}\sum_{\sfc=1}^\sfC\E\left[\|f_k(\bar{B}+G_\sfc^k,\sfc)\|^2_2\right]\right) \cdot \E\left[\|\Bar{B}\|_2^2\right]}.
% \end{align*}

Each point in the AMP performance curves is obtained from 10 independent runs; in each run, the algorithm is executed for 300 iterations.  Our benchmark will be the linear programming (LP) estimator adapted to the pooled data problem \cite[Section 4.1]{Tan23d}. Recall from \eqref{eq:gamma_star} that the information theoretic lower bound on the number of tests for the noiseless pooled data problem is
\begin{align*}
    n^*
    &=\frac{p}{\log p}
    \left(\max_{r\in\{1,\dots,L-1\}}\frac{2[H(\pi)-H(\pi^{(r)})]}{L-r}\right).
\end{align*}

Figure \ref{fig:corr_v_delta_pool} shows how the normalized squared correlation varies with the sampling ratio $\delta$ for pooled data with $L=3$ equally likely categories. Figure \ref{fig:pool_sc_amp} shows that the state evolution prediction  of the matrix SC-AMP performance (SC-SE curves) improves as the spatial coupling parameters $(\omega,\Lambda)$ increase  from $(6,40)$ to $(20,200)$. As in QGT, the gap between the empirical performance of matrix SC-AMP  and the state evolution prediction  for $(6,40)$ is due to finite length effects. We did not implement the matrix SC-AMP algorithm for $(20,200)$ as it requires a large amount of computational memory.  In Figure \ref{fig:pool_sc_amp_lp}, we observe that iid-AMP performs better than  iid-LP and SC-LP, while SC-AMP outperforms all three, justifying the use of a SC design matrix with matrix SC-AMP for recovery. We also implemented the iterative hard thresholding algorithm \cite[Section 4.1]{Tan23d} but found that it performed significantly worse than AMP and LP, and so omitted it from our comparisons.

Figure \ref{fig:pool_iterative-sc_amp}  compares the performance of matrix SC-AMP with the column-wise SC-AMP algorithm. To make the algorithms comparable, the estimates from each algorithm (and the corresponding SE) were quantized in the same way after their final iteration, using the rule in \eqref{eq:quant_col_SCAMP}.  This leads to a slight difference in the AMP performance curves and the theoretical SE estimates compared to Figure \ref{fig:corr_v_delta_pool}, where no quantization was used.  As expected, the matrix SC-AMP algorithm outperforms the column-wise SC-AMP algorithm since the former takes advantage of the correlation within each row of the matrix signal. Nevertheless, the column-wise SC-AMP algorithm performs slightly better than the matrix AMP algorithm with an i.i.d.~matrix.

\section{Proof of Theorem \ref{thm:SC_AMP_QGT}} \label{sec:SC_AMP_QGT_proof}

We start by defining \emph{generalized white noise  matrices}, which will be used in the proof of the theorem.

\begin{definition} \label{def:gen_white_noise_matrix} \textup{\cite[Definition 2.15]{Wan22}}
A generalized white noise matrix $\tX\in\mb{R}^{n\times p}$ with a (deterministic) variance profile $S\in\mb{R}^{n\times p}$ is one satisfying the following conditions, for $i \in [n], j \in [p]$:
\begin{enumerate}
    \item All entries $\tX_{ij}$ are independent.
    \item Each entry $\tX_{ij}$ has mean 0, variance $n^{-1}S_{ij}$, and higher moments satisfying, for each integer $m\geq3$,
    \begin{align}
        \lim_{n,p\rightarrow\infty} p\cdot\max_{i\in[n]}\max_{j\in[p]}\E\Big[|\tX_{ij}|^m\Big]=0.
        \label{eq:mom_condition}
    \end{align}
    \item For a constant $C>0$,
    \begin{align}
        \max_{i\in[n]}\max_{j\in[p]}S_{ij}\leq C,\quad
\lim_{n,p\rightarrow\infty}\max_{i\in[n]}\Big|\frac{1}{p}\sum_{j=1}^pS_{ij}-1\Big|=0,\quad
\lim_{n,p\rightarrow\infty}\max_{j\in[p]}\Big|\frac{1}{n}\sum_{i=1}^nS_{ij}-1\Big|=0.
\label{eq:var_conditions}
    \end{align}
\end{enumerate}
\end{definition}

Definition \ref{def:gen_white_noise_matrix} simplifies for the case of $S_{ij}=1$ for all $(i,j)\in[n]\times[p]$. In this case, the entries are all i.i.d.~with variance $1/n$, the third condition in the definition is trivially satisfied, and the second condition requires moments of order 3 and higher to decay faster than $1/p$. The rescaled i.i.d.~design matrix $\tXiid$ in \eqref{eq:rescaled_QGT_iid_matrix} is a generalized white noise matrix, but the rescaled spatially coupled matrix $\tXsc$ in \eqref{eq:rescaled_QGT_SC_matrix} is not. Indeed, $\tXsc$ satisfies the first two requirements in Definition \ref{def:gen_white_noise_matrix} and from \eqref{eq:omega_Lambda_var}, its variance profile satisfies the first and last conditions in \eqref{eq:var_conditions}, but not the second: for $i \in [n]$, we have  $\frac{1}{p}\sum_{j=1}^pS_{ij} = \frac{\sfR}{\sfC}$, which is close to, but not equal to $1$ for large $\Lambda/\omega$.

We prove Theorem \ref{thm:SC_AMP_QGT} via a more general result, for  a \emph{generalized linear model}  with a spatially coupled design, where the observations $y_i \in \reals$ are generated as:
\begin{align}
    \ty_i
    =q\left(\big(\tX_{i,:}^{\text{sc}}\big)^\top\beta,\,\tPsi_i\right)
    =q\left(\Theta_i,\tPsi_i\right), \quad 
    \text{ for $i\in[n]$}.
    \label{eq:GLM}
\end{align}
Here $\beta \in \reals^p$ is the signal to be estimated, $\tPsi \in \reals^n$ is a noise vector, and $q: \reals^2 \to \reals$ is a known output function. We also allow  $\tXsc$ to be more general than the one in Definition \ref{def:omega_Lambda_base_matrix}. The generalized spatially coupled matrix $\tXsc \in \reals^{n \times p}$ consists of independent zero-mean entries whose variances are specified by a generic base matrix $\tW\in\mb{R}^{\sfR\times\sfC}$, which satisfies the following conditions:
\begin{align}
    \sum_{\sfr=1}^\sfR\tW_{\sfr\sfc}=1
    \quad
    \text{for } \sfc\in[\sfC],
    \quad \text{ and } \quad 
    \kappa_1\leq\sum_{\sfc=1}^\sfC\tW_{\sfr\sfc}\leq\kappa_2,    \quad
    \text{for } \sfr \in [\sfR],
    \label{eq:gen_base_matrix}
\end{align}
for some $\kappa_1,\kappa_2>0$. Given a base matrix $\tW$
satisfying \eqref{eq:gen_base_matrix}, we construct the spatially coupled matrix $\tXsc$ with independent entries drawn from a distribution with zero-mean and variance 
$\E[ |{\tXsc}_{ij}|^2] =
\frac{\tW_{\sfr(i) \sfc(j)}}{n/R}$, 
for $i \in [n], j \in [p]$. 
We also assume that the higher moments $\E[  |\tXsc_{ij}|^m]$ for $m \ge 3$, satisfy \eqref{eq:mom_condition}.

The first condition in \eqref{eq:gen_base_matrix} ensures that the expected squared norm of each column of $\tXsc$ is $1$, and the second condition in \eqref{eq:gen_base_matrix} bounds the variance of each entry of $\tXsc\beta$ from above and below.

\paragraph{High-level sketch of proof of Theorem \ref{thm:SC_AMP_QGT}.}  The proof  consists of three reductions.
\begin{enumerate}
    \item In Section \ref{sec:reduction_to_SC-GAMP}, we introduce the spatially coupled generalized approximate message passing  algorithm (SC-GAMP) for the generalized linear model in \eqref{eq:GLM}, and characterize its performance via state evolution (Theorem \ref{thm:SC_GAMP}). We then reduce the SC-AMP algorithm in \eqref{eq:SC_AMP} to SC-GAMP, and use the state evolution result of the latter to prove Theorem 
    \ref{thm:SC_AMP_QGT}. 
    \item To prove the state evolution result for SC-GAMP (Theorem \ref{thm:SC_GAMP}), we show that the algorithm can be written as an instance of an abstract matrix-AMP iteration defined via a generalized white noise matrix. This reduction, shown  in Appendix \ref{sec:reduction_to_abs_matrix-AMP},   is similar to the one used in \cite[Appendix A]{Cob23} for reducing the SC-GAMP algorithm for a Gaussian design to an abstract matrix-AMP iteration. 
    
    \item To prove the state evolution result for the abstract matrix-AMP  (Theorem \ref{thm:abs_matrix_AMP}), we show that it is a special case of an AMP iteration  for generalized white noise matrices,
    for which a rigorous state evolution result was established in \cite{Wan22}.
    %along with its state evolution result.  
    We refer to the latter iteration as U-AMP, where the `U' stands for universal.  The technique for reducing the abstract matrix-AMP to U-AMP is similar to the one presented in \cite{Tan23d}. This is shown  in Appendix \ref{sec:reduction_to_U-AMP}.
\end{enumerate}

As before, to simplify notation,  for vectors $a\in\mb{R}^p$ and $b\in\mb{R}^n$, we will use $a_\sfc:=a_{\mathcal{J}_\sfc}$ and $b_\sfr:=b_{\mathcal{I}_\sfr}$, where  $\mathcal{J}_\sfc$ and $\mathcal{I}_\sfr$ are defined in \eqref{eq:J_c_I_r_def}. There will be no notation simplifications for matrices. 
%\rv{check this}

\subsection{The SC-GAMP Algorithm and its State Evolution} \label{sec:reduction_to_SC-GAMP}

The SC-GAMP algorithm aims to estimate the signal $\beta\in\mb{R}^p$ from observations $\ty\in\mb{R}^n$ generated according to the generalized linear model \eqref{eq:GLM}. For iteration $k\geq0$, the algorithm computes:
    \begin{align}
    \begin{split}
        &\Theta^k=\tX^{\text{sc}}\hbeta^k-b^k\odot\hR^{k-1},\quad
        \hR^k=g_k(\Theta^k,\ty,\mathcal{R}), \\
        &\beta^{k+1}=(\tX^{\text{sc}})^\top\hR^k-c^k\odot\hbeta^k,\quad
        \hbeta^{k+1}=f_{k+1}(\beta^{k+1},\mathcal{C}),
    \end{split}
    \label{eq:SC_GAMP}
    \end{align}
    where $\odot$ denotes element-wise product. The algorithm is initialized with some $\hbeta^0\in\mb{R}^p$ and 
    $\Theta^0=\tX^{\text{sc}}\hbeta^0$.
    The functions $g_k: \reals^{2} \times [\sfR]$ and $f_{k+1}: \reals \times [\sfC]$ act row-wise on their input, and 
    \begin{align}
    \begin{split}
        \mathcal{C}
        &=(\underbrace{1,\dots,1}_{\text{$p/\sfC$ entries}},\underbrace{2,\dots,2}_{{\text{$p/\sfC$ entries}}},\dots,\underbrace{\sfC,\dots,\sfC}_{{\text{$p/\sfC$ entries}}})^\top
        \in\mb{R}^p, \\
        \mathcal{R}
        &=(\underbrace{1,\dots,1}_{\text{$n/\sfR$ entries}},\underbrace{2,\dots,2}_{{\text{$n/\sfR$ entries}}},\dots,\underbrace{\sfR,\dots,\sfR}_{{\text{$n/\sfR$ entries}}})^\top
        \in\mb{R}^n.
    \end{split}
    \label{eq:S_c_and_S_r_def}
    \end{align}
    The entries of $c^k \in  \reals^p$ and $b^k \in \reals^n$ are defined as follows, for $j \in [p]$, $i \in [n]$:
    \begin{align*}
        c_j^k
        &=\sum_{\sfr=1}^\sfR\frac{\tW_{\sfr\sfc}}{n/\sfR}\sum_{i\in\mathcal{I}_r}\partial_1g_k(\Theta_i^k,\ty_i,\sfr), \quad
        b_i^{k}
        =\sum_{\sfc=1}^\sfC \frac{\tW_{\sfr\sfc}}{n/\sfR}\sum_{j\in\mathcal{J}_c}\partial_1f_{k}(\beta_j^{k},\sfc),
    \end{align*}
   where $\partial_1$ denotes the derivative with respect to the first argument.
% \end{itemize}

\paragraph{State evolution.} The `memory' terms $-b^k\odot\hR^{k-1}$ and $-c^k\odot\hbeta^k$ in \eqref{eq:SC_GAMP} debias the iterates $\Theta^k$ and $\beta^{k+1}$, ensuring that their empirical distributions are accurately captured by state evolution in the high-dimensional limit. Theorem \ref{thm:SC_GAMP} below shows that for each $k\geq1$ and $\sfc \in [\sfC]$, the empirical distribution of $\beta_\sfc^k$ converges to the distribution of $\mu_{\beta,\sfc}^k\bar{\beta}+G_{\beta,\sfc}^k$ where $\bar{\beta}$ is the random variable representing the limiting distribution of the entries of the signal $\beta_\sfc$, and $G_{\beta,\sfc}^k\sim\normal\big(0,(\sigma_{\beta,\sfc})^2\big)$ is independent of $\bar{\beta}$. The deterministic parameters $\mu_{\beta,\sfc}^k\in\mb{R}$ and $\sigma_{\beta,\sfc}^k\in\mb{R}$ are defined below. The result implies that the empirical distribution of the estimate $\hbeta_\sfc^k$ converges to the distribution of $f_k\big(\mu_{\beta.\sfc}^k\bar{\beta}+G_{\beta,\sfc}^k\big)$. Thus, $f_k$ can be viewed as a denoising function that can be tailored to take advantage of the prior in $\bar{\beta}$. Theorem \ref{thm:SC_GAMP} also shows that the joint empirical distribution of  
the rows of $( \Theta_\sfr, \Theta_\sfr^k)$ converges to $\normal(0,\Sigma^{k,\sfr})$, where $\Sigma^{k,\sfr}\in\mb{R}^{2\times2}$ is defined below. 
 
We now describe the state evolution recursion defining  $\mu_{\beta,\sfc}^k,\sigma_{\beta,\sfc}^k\in\mb{R}$ and $\Sigma^{k,\sfr}\in\mb{R}^{2\times2}$. Define $\bar{g}_k:\mb{R}^3\times[\sfR]\rightarrow\mb{R}$ such that
\begin{align}
    g_k(\Theta_i^k,\ty_i,\sfr)
    &=\bar{g}_k(\Theta_i,\Theta_i^k,\tPsi_i,\sfr)
    \text{ for $i\in\mathcal{I}_\sfr$},
    \label{eq:g_k_eq_g_k_bar}
\end{align}
since $\ty_i=q(\Theta_i,\tPsi_i)$. Starting with an initializer $\Sigma^{0,\sfr}\in\mb{R}^{2\times 2}$ for $\sfr\in[\sfR]$ (defined later in \eqref{eq:cov_init}), the state evolution parameters are iteratively computed as follows for $k \ge 0$, and $\sfr \in [\sfR], \sfc \in [\sfC]$:
\begin{align}
    & \mu^{k+1}_{\beta,\sfc}
    =\sum_{\sfr=1}^\sfR\tW_{\sfr\sfc}\E[\partial_1\bar{g}_k(Z_\sfr,Z_\sfr^k,\bar{\Psi},\sfr)],
    \quad
    (\sigma^{k+1}_{\beta,\sfc})^2
    =\sum_{\sfr=1}^\sfR\tW_{\sfr\sfc}\E[\bar{g}_k(Z_\sfr,Z_\sfr^k,\bar{\Psi},\sfr)^2], \nonumber
    \\
    & \Sigma^{k+1,\sfr}
    =\begin{bmatrix}
        \Sigma_{11}^{k+1,\sfr} & \Sigma_{12}^{k+1,\sfr} \\
        \Sigma_{21}^{k+1,\sfr} & \Sigma_{22}^{k+1,\sfr}
    \end{bmatrix},
    \label{eq:muk1_beta_sigk1_beta}
\end{align}
where  $(Z_\sfr,Z_\sfr^k)\sim\normal(0,\Sigma^{k,\sfr})$ are independent of $\bar{\Psi}$, and 
\begin{align}
\begin{split}
    \Sigma_{11}^{k+1,\sfr}
    &=\E[(Z_\sfr)^2]
    =\frac{\E[\bar{\beta}^2]}{\deltain}\sum_{\sfc=1}^\sfC\tW_{\sfr\sfc}, \\
    \Sigma_{12}^{k+1,\sfr}
    &=\Sigma_{21}^{k+1,\sfr}
    =\frac{1}{\deltain}\sum_{\sfc=1}^\sfC\tW_{\sfr\sfc}\E[\bar{\beta}f_{k+1}(\mu^{k+1}_{\beta,\sfc}\bar{\beta}+G^{k+1}_{\beta,\sfc},\sfc)], \\
    \Sigma_{22}^{k+1,\sfr}
    &=\frac{1}{\deltain}\sum_{\sfc=1}^\sfC\tW_{\sfr\sfc}\E[f_{k+1}(\mu^{k+1}_{\beta,\sfc}\bar{\beta}+G^{k+1}_{\beta,\sfc},\sfc)^2].
\end{split}
\label{eq:Sigma_SC_GAMP}
\end{align}
Here $G_{\beta,\sfc}^{k+1}\sim\normal(0,(\sigma_{\beta,\sfc}^{k+1})^2)$ is independent of $\bar{\beta}$.

The SC-GAMP algorithm and its state evolution equations are similar to those introduced in \cite{Cob23}, the only difference being that the SC design matrix $\tXsc\in \mb{R}^{n\times p}$ is now a generalized spatially coupled matrix instead of the spatially coupled \emph{Gaussian} one used in \cite{Cob23}. We note that  $\tXsc$ is not a generalized white noise matrix since it has variance profile $S_{ij}^{\text{sc}}:=\sfR\tW_{\sfr(i),\sfc(j)}$, which is not guaranteed to satisfy the condition $\max_{i\in[n]}\big|\frac{1}{p}\sum_{j=1}^pS_{ij}^{\text{sc}}-1\big|\rightarrow0$ in Definition \ref{def:gen_white_noise_matrix}. Nevertheless, $\tXsc$ can be related to a generalized white noise matrix $\tX\in\mb{R}^{n\times p}$ defined as follows.  For $i \in [n], j \in [p]$, let:
\begin{align}
    \tX_{ij}
    := \begin{cases}
     \frac{\tXsc_{ij}}{\sqrt{\sfR\cdot\tW_{\sfr(i),\sfc(j)}}}  &  \text{ if } \tW_{\sfr(i),\sfc(j)}\neq 0, \\ 
      \stackrel{\text{\iid}}{\sim}
      \normal\big(0,\frac{1}{n}\big) & \text{ otherwise}.
    \end{cases}
    \label{eq:SC_LM_matrix}
\end{align}
(In the second line of the definition, we could use any sub-Gaussian distribution with mean zero and variance $1/n$  instead of $\normal(0,1/n)$.)
From the construction of $\tXsc$ (see below \eqref{eq:gen_base_matrix}), it follows that  $\tX$  is a generalized white noise matrix with variance profile $S_{ij}=1$ for all $(i,j)\in[n]\times[p]$.

The state evolution result for SC-GAMP requires the following assumptions on the model and the algorithm:
\begin{itemize}
    \item[\textbf{(A1)}] As $n, p \to \infty$, we have $\frac{n}{p} \to \delta$. The signal $\beta$, initializer $\hbeta^0$, and the noise vector $\tPsi$ are independent of $\Xsc$, and their empirical distributions have well-defined limits. 
    There exist random variables $\bar{\beta}\sim P_{\bar{\beta}}$  and $\bar{\Psi}\sim P_{\bar{\Psi}}$  with $\beta \stackrel{W}{\rightarrow}\bar{\beta}$ and $\tPsi\stackrel{W}{\rightarrow}\bar{\Psi}$, respectively.
    
    \item[\textbf{(A2)}] As $p\rightarrow\infty$, $(\beta_\sfc,\hbeta_\sfc^0)\stackrel{W}{\rightarrow}(\bar{\beta},\bar{\beta}_\sfc^0)$ almost surely, with joint laws $(\bar{\beta},\bar{\beta}_\sfc^0)$ having finite moments of all orders, for $\sfc\in[\sfC]$. Furthermore,   multivariate polynomials are dense in the real $L^2$-spaces of functions $f:\mb{R}\rightarrow\mb{R}$ and $g:\mb{R}^{2}\rightarrow\mb{R}$ with the inner-products
         $$
         \big\langle f,\tilde{f}\big\rangle
         :=\E\big[f\big(\bar{\Psi}\big)\tilde{f}\big(\bar{\Psi}\big)\big]
         \quad
         \text{and}
         \quad
         \big\langle g,\tilde{g}\big\rangle
         :=\E\big[g\big(\bar{\beta},\bar{\beta}_\sfc^0\big)\tilde{g}\big(\bar{\beta},\bar{\beta}_\sfc^0\big)\big].
         $$   
    \item[\textbf{(A3)}] For $k \ge 0$ and $\sfr \in [\sfR], \sfc \in [\sfC]$, the functions $f_{k}(\cdot,\sfc)$ and  $\bar{g}_{k}(\cdot,\cdot,\cdot,\sfr)$ are each continuous, Lipschitz w.r.t.~the first argument, and satisfy the polynomial growth condition in \eqref{eq:poly_growth_cond} for some order $r\geq 1$.
     \item[\textbf{(A4)}] The matrix $\tX$ defined as in \eqref{eq:SC_LM_matrix} satisfies $\|\tX\|_{\text{op}}<C$ for some constant $C$, and for any fixed polynomial function $f^\dag:\mb{R}^{2\sfR+2}\rightarrow\mb{R}$, as $n,p\rightarrow\infty$,
    \begin{align*}
        \max_{i\in[n]}&\bigg|\left\langle f^\dag\left(\beta_\sfc,\beta_\sfc\sqrt{\sfR\tW_{1\sfc}},\dots,\beta_\sfc\sqrt{\sfR\tW_{\sfR\sfc}}, \, \hbeta_\sfc^0\sqrt{\sfR\tW_{1\sfc}},\dots,\hbeta_\sfc^0\sqrt{\sfR\tW_{\sfR\sfc}},\sfc\right)\odot S_{i,\mathcal{J}_\sfc} \right\rangle \\
        &\quad-\left\langle f^\dag\left(\beta_\sfc,\beta_\sfc\sqrt{\sfR\tW_{1\sfc}},\dots,\beta_\sfc\sqrt{\sfR\tW_{\sfR\sfc}}, \, \hbeta_\sfc^0\sqrt{\sfR\tW_{1\sfc}},\dots,\hbeta_\sfc^0\sqrt{\sfR\tW_{\sfR\sfc}},\sfc\right)\right\rangle\cdot\langle S_{i,\mathcal{J}_\sfc}\rangle\bigg|
        \stackrel{a.s.}{\rightarrow}0,
    \end{align*}
    for all $\sfc\in[\sfC]$, where  $S$ is the variance profile of $\tX$ (see Definition \ref{def:gen_white_noise_matrix}) and $f^\dag$ acts element-wise on $\beta_\sfc$.
    \item[\textbf{(A5)}] For any fixed polynomial function $f^\ddag:\mb{R}^{2}\rightarrow\mb{R}$, as $n,p\rightarrow\infty$,
    \begin{align*}
        \max_{j\in[p]}&\left|\left\langle f^\ddag\left(\tPsi_\sfr,\sfr\right)\odot S_{\mathcal{I}_\sfr,j}\right\rangle-\left\langle f^\ddag\left(\tPsi_\sfr,\sfr\right)\right\rangle\cdot\langle S_{\mathcal{I}_\sfr,j}\rangle\right|
        \stackrel{a.s.}{\rightarrow}0,
    \end{align*}
    for all $\sfr\in[\sfR]$, where $f^\ddag$ acts element-wise on $\tPsi_\sfr$.
\end{itemize}

    The state evolution recursion in \eqref{eq:muk1_beta_sigk1_beta}-\eqref{eq:Sigma_SC_GAMP} is initialized with
    \begin{align}
            &\Sigma^{0,\sfr}=
            \begin{bmatrix}
                \frac{1}{\deltain}\sum_{\sfc=1}^\sfC\tW_{\sfr\sfc}\E[\bar{\beta}^2]
                & \frac{1}{\deltain}\sum_{\sfc=1}^\sfC\tW_{\sfr\sfc}\E[\bar{\beta}\bar{\beta}_\sfc^0] \\
                \frac{1}{\deltain}\sum_{\sfc=1}^\sfC\tW_{\sfr\sfc}\E[\bar{\beta}_\sfc^0\bar{\beta}]
                & \frac{1}{\deltain}\sum_{\sfc=1}^\sfC\tW_{\sfr\sfc}\E[(\bar{\beta}_\sfc^0)^2]
            \end{bmatrix}, \qquad \sfr \in [\sfR].
        \label{eq:cov_init}
    \end{align}

\begin{theorem}[State evolution for SC-GAMP] \label{thm:SC_GAMP}
%\textbf{\textup{()}}
Consider the GLM in \eqref{eq:GLM} with spatially coupled design $\tXsc$ defined via a base matrix satisfying \eqref{eq:gen_base_matrix}, and signal estimation using the SC-GAMP recursion in \eqref{eq:SC_GAMP}. Let Assumptions (A1)--(A5) be  satisfied, and assume $\sigma_{\beta,\sfc}^{1}>0$ for $\sfc\in[\sfC]$. Then for each $k \ge 0$, we have
\begin{align}
\begin{split}
    \big(\beta_\sfc, \beta_\sfc^{k+1}\big)
    &\stackrel{W_2}{\rightarrow}
    \big(\bar{\beta}, \mu_{\beta,\sfc}^{k+1}\bar{\beta}+G_{\beta,\sfc}^{k+1}\big),
    \quad
    \big(\Theta_\sfr,\Theta_\sfr^k\big)
    \stackrel{W_2}{\rightarrow}
    \big(Z_\sfr,Z_{\sfr}^k\big),
\end{split}
\label{eq:SC_GAMP_asym_dist}
\end{align}
almost surely as $n,p\rightarrow\infty$ with $n/p\rightarrow\delta$.
\end{theorem}

The proof  is  given in Appendix \ref{sec:reduction_to_abs_matrix-AMP}. We now use Theorem \ref{thm:SC_GAMP} to prove Theorem \ref{thm:SC_AMP_QGT}.

\subsection{Proof of Theorem \ref{thm:SC_AMP_QGT} using Theorem \ref{thm:SC_GAMP}} 
We first verify that the rescaled QGT model 
\[
\ty_i
    =q\left(\big(\tX_{i,:}^{\text{sc}}\big)^\top\beta,\,\tPsi_i\right)
    =\big(\tX_{i,:}^{\text{sc}}\big)^\top\beta+\tPsi_i
\]  
is a special case of the GLM \eqref{eq:GLM} with a generalized spatially coupled design $\tXsc$ constructed as described below \eqref{eq:gen_base_matrix}.
The rescaled  QGT design $\tXsc$  in \eqref{eq:rescaled_QGT_SC_matrix} has independent~zero mean entries with variances 
$\frac{\tW_{\sfr(i) \sfc(j)}}{n/R}$  for $i \in [n], j \in [p]$, where from  \eqref{eq:W_tilde} we have  $\tW_{\sfr\sfc}=1/\omega$ for $\sfc \le \sfr \le \sfc + \omega -1$, and 0 otherwise. Moreover, this $\tW$ satisfies the conditions in \eqref{eq:gen_base_matrix}.

Next, we show that the SC-AMP in \eqref{eq:SC_AMP} is a special case of the SC-GAMP algorithm by choosing
\begin{align}
    f_k( \beta^k_j \, ,\sfc) & =\E\big[\bar{\beta}\,\big|\,\mu_{\beta,\sfc}^k\bar{\beta}+G_{\beta,\sfc}^k 
    = \beta^k_j\big], \quad \text{ for } j \in \mathcal{J}_{\sfc},
    \label{eq:chosen_fk} \\
       g_k(\Theta_i^k,\ty_i,\sfr)
    &=\frac{\ty_i-\Theta_i^k}{\Sigma_{11}^{k,\sfr}-\Sigma_{12}^{k,\sfr}+\sigma^2},  \quad \text{ for } i \in \mathcal{I}_\sfr.
    \label{eq:chosen_gk}
\end{align}
The choices in \eqref{eq:chosen_fk} and \eqref{eq:chosen_gk} are based on the Bayes-optimal denoisers for an i.i.d.~design and  Gaussian noise $\bar{\Psi}\sim\normal(0,\sigma^2)$ (see \cite[Section 4.2]{Fen21}), and take into account the block-wise dependence of the state evolution parameters.  With this choice of $f_k$, in \eqref{eq:Sigma_SC_GAMP} we have $\E[\bar{\beta} f_{k}]=\E\big[f_k^2\big]$ which implies that $    \Sigma_{12}^{k,\sfr}=\Sigma_{21}^{k,\sfr}=\Sigma_{22}^{k,\sfr}$, for $k \ge 1$.

With our choice of denoisers, the iterate $\hR^k \in \reals^n$ in \eqref{eq:SC_GAMP} can be written as
\begin{align}
    \hR^k=Q^k\odot(\ty-\Theta^k),
    \label{eq:hatR_k_simplified}
\end{align}
where the entries of $Q^k\in\mb{R}^n$ are   $Q_i^k=\big(\Sigma_{11}^{k,\sfr(i)}-\Sigma_{12}^{k,\sfr(i)}+\sigma^2\big)^{-1}$, for $i \in [n]$. For $j \in [p]$, we have:
\begin{align}
    c_j^k
    &=\sum_{\sfr=1}^\sfR\frac{\tW_{\sfr\sfc}}{n/\sfR}\sum_{i\in\mathcal{I}_r}\partial_1g_k(\Theta_i^k,\ty_i,\sfr)
    =\frac{\sfR}{n}\sum_{\sfr=1}^\sfR\tW_{\sfr\sfc}\sum_{i\in\mathcal{I}_\sfr}\frac{-1}{\Sigma_{11}^{k,\sfr(i)}-\Sigma_{12}^{k,\sfr(i)}+\sigma^2} \nonumber \\
    &=\frac{\sfR}{n}\sum_{\sfr=1}^\sfR\tW_{\sfr\sfc}\cdot\frac{n}{\sfR}\cdot\frac{-1}{\Sigma_{11}^{k,\sfr}-\Sigma_{12}^{k,\sfr}+\sigma^2}
    =-\sum_{\sfr=1}^\sfR\tW_{\sfr\sfc}\big(\Sigma_{11}^{k,\sfr}-\Sigma_{12}^{k,\sfr}+\sigma^2\big)^{-1}.
    \label{eq:c_j_k_explicit}
\end{align}
Next, we have
\begin{align}
    (\sigma^{k+1}_{\beta,\sfc})^2
    &=\sum_{\sfr=1}^\sfR\tW_{\sfr\sfc}\E[\bar{g}_k(Z_\sfr,Z_\sfr^k,\bar{\Psi},\sfr)^2]
    \stackrel{(a)}{=}\sum_{\sfr=1}^\sfR\tW_{\sfr\sfc}
    \E\left[\frac{(Z_\sfr+\bar{\Psi}-Z_\sfr^k)^2}{(\Sigma_{11}^{k,\sfr}-\Sigma_{12}^{k,\sfr}+\sigma^2)^2}\right] \nonumber \\
    &\stackrel{(b)}{=}\sum_{\sfr=1}^\sfR\tW_{\sfr\sfc}\frac{\sigma^2+\E[(Z_\sfr-Z_\sfr^k)^2]}{(\Sigma_{11}^{k,\sfr}-\Sigma_{12}^{k,\sfr}+\sigma^2)^2}
    =\sum_{\sfr=1}^\sfR\tW_{\sfr\sfc}\big(\Sigma_{11}^{k,\sfr}-\Sigma_{12}^{k,\sfr}+\sigma^2\big)^{-1}
    \stackrel{(c)}{=}-c_j^k,
    \label{eq:c_j_k_simplified}
\end{align}
where (a) applies \eqref{eq:chosen_gk}, (b) uses the independence between $(Z_\sfr,Z_\sfr^k)$ and $\bar{\Psi}$, and (c) uses \eqref{eq:c_j_k_explicit}. Substituting the definitions of  
$\Sigma_{11}^{k,\sfr}$ and $\Sigma_{12}^{k,\sfr}$ in \eqref{eq:Sigma_SC_GAMP} into \eqref{eq:c_j_k_simplified}, we get
\begin{align}
    (\sigma^{k+1}_{\beta,\sfc})^2
    =\sum_{\sfr=1}^\sfR\tW_{\sfr\sfc}
    \left(
    \sigma^2+\frac{1}{\deltain}\sum_{\sfc'=1}^\sfC\tW_{\sfr\sfc'}\E\Big[(\bar{\beta}-f_k(\mu_{\beta,\sfc}^k\bar{\beta}+G_{\beta,\sfc}^k, \sfc))^2\Big]
    \right)^{-1}.
    \label{eq:sigma_bc_k_simplified}
\end{align}
We also have the identity
\begin{align}
    \mu_{\beta,\sfc}^{k+1}
    &=\sum_{\sfr=1}^\sfR\tW_{\sfr\sfc}\E[\partial_1\bar{g}_k(Z_\sfr,Z_\sfr^k,\bar{\Psi},\sfr)]
    \stackrel{(a)}{=}\sum_{\sfr=1}^\sfR\tW_{\sfr\sfc}\big(\Sigma_{11}^{k,\sfr}-\Sigma_{12}^{k,\sfr}+\sigma^2\big)^{-1}
    \stackrel{(b)}{=}\big(\sigma_{\beta,\sfc}^{k+1}\big)^2,
    \label{eq:bayes_opt_prop2}
\end{align}
where (a) uses \eqref{eq:chosen_gk} and $\bar{y}_\sfr=Z_\sfr+\bar{\Psi}$, and (b) uses the last equality in \eqref{eq:c_j_k_simplified}. 

Letting $\chi_{\sfc}^k=\sigma_{\beta,\sfc}^k$, we observe that the update equations in \eqref{eq:sigma_bc_k_simplified}--\eqref{eq:bayes_opt_prop2} match the state evolution recursion of the SC-AMP algorithm in \eqref{eq:SC_SE_param1}.
Then, substituting \eqref{eq:chosen_fk}-\eqref{eq:sigma_bc_k_simplified} into SC-GAMP in \eqref{eq:SC_GAMP}, followed by a change of variables from $\Theta^k$ to  $\tTheta^k:=\ty-\Theta^k$,  gives us the SC-AMP algorithm in \eqref{eq:SC_AMP}. 
%Furthermore, substituting $\chi_{\sfc}^k=\sigma_{\beta,\sfc}^k$ gives us the state evolution recursion of the SC-AMP algorithm. 
Finally, we check that the assumptions of Theorem \ref{thm:SC_GAMP} are satisfied:
\begin{itemize}
    \item Assumptions \textbf{(A1)} and \textbf{(A2)} hold due to the model assumptions in Section \ref{sec:prelim}, the noise scaling assumption (p.~\pageref{par:noise_scaling}), and the SC-AMP initialization $\hbeta^0=\E[\bar{\beta}]1_p$. Recalling that $\bar{\beta} \sim \text{Bernoulli}(\pi)$ for the QGT model, the state evolution initialization in \eqref{eq:cov_init} becomes
    \begin{align*}
        \Sigma^{0,\sfr}=
        \frac{1}{\deltain}
        \begin{bmatrix}
            \sum_{\sfc=1}^\sfC\tW_{\sfr\sfc}\pi
            & \sum_{\sfc=1}^\sfC\tW_{\sfr\sfc}\pi^2 \\
            \sum_{\sfc=1}^\sfC\tW_{\sfr\sfc}\pi^2
            & \sum_{\sfc=1}^\sfC\tW_{\sfr\sfc}\pi^2
        \end{bmatrix}, \qquad \sfr \in [\sfR].
    \end{align*}
    Using this in \eqref{eq:sigma_bc_k_simplified}, we obtain that $\sigma_{\beta,\sfc}^1=\chi_{\sfc}^1$, where the latter is defined in \eqref{eq:SEchi_init}.
    \item \textbf{(A3)}.  With $\bar{\beta} \sim \text{Bernoulli}(\pi)$, the denoiser $f_k( \cdot, \sfc)$ in \eqref{eq:chosen_fk} can be explicitly computed (see \eqref{eq:f_k_bayes}), and the choice for $\bar{g}_k(\cdot, \cdot, \cdot,  \sfr)$ is given by \eqref{eq:chosen_gk} and \eqref{eq:g_k_eq_g_k_bar}. From these expressions, it can be verified  that both functions are continuous, Lipschitz w.r.t.~the first argument, and satisfy the polynomial growth condition with $r=2$.
    \item \textbf{(A4)} and \textbf{(A5)}. Recalling the definition of $\tX$ in \eqref{eq:SC_LM_matrix} and  of $\tXsc$ in \eqref{eq:rescaled_QGT_SC_matrix}, we note the matrix $\sqrt{n}\tX$ has independent ~sub-Gaussian entries of variance 1. Using a concentration inequality for the operator norm of sub-Gaussian matrices \cite[Theorem 4.4.5]{Ver18} together with the Borel-Cantelli lemma, we obtain that $\|\tX\|_{\text{op}}<C$ almost surely for sufficiently large $p$. Since the variance profile $S_{ij}=1$ for all $(i,j)$, the second condition in (A3) is trivially satisfied. Assumption (A5) is similarly   satisfied.
\end{itemize}
This completes the proof. \qed

\section{Proof of Theorem \ref{thm:asym_MSE}}
\label{sec:proof:thm:asym_MSE}

\subsection{Proof of \eqref{eq:asym_MSE_SC_AMP}} \label{sec:asym_thm_p1_proof}

The idea is to rewrite the SC-AMP state evolution in  \eqref{eq:SC_SE_param1} in terms of a general coupled  recursion analyzed by Yedla et al.~in \cite{Yed14}. We then apply the fixed point characterization of \cite[Theorem 1]{Yed14} to the SC-AMP state evolution to obtain \eqref{eq:asym_MSE_SC_AMP}. 

\paragraph{General coupled  recursion \cite{Yed14}.} 
Let $\mathcal{X}=[0,x_{\max}]$, $\mathcal{Y}=[0,y_{\max}]$ with $x_{\max}, y_{\max} \in(0,\infty)$. Let $f:\mathcal{Y}\rightarrow\mathcal{X}$ be a non-decreasing $C^1$ function, and let $g:\mathcal{X}\rightarrow\mathcal{Y}$ be a strictly increasing $C^2$ function with $y_{\max} = g(x_{\max})$. (We say a function $f:\mathcal{Z}\rightarrow\mb{R}$ is $C^d$ if its $d$th derivative exists and is continuous on $\mathcal{Z}$.)  Consider a matrix  $A\in\mb{R}^{\sfC\times \sfR}$ with   
$\sfR=\sfC+\omega-1$, whose entries are defined as follows, for $\sfr \in [\sfR], \sfc \in [\sfC]$:
\begin{align*}
    A_{\sfc\sfr}
    =\begin{cases}
        \frac{1}{\omega} &\text{if $ \sfc \leq\sfr\leq \sfc + \omega -1$,} \\
        0 &\text{otherwise.}
    \end{cases}
\end{align*} 
Using $A$, we define the following coupled recursion. For $\sfr \in [\sfR]$:
\begin{align}
\begin{split}
    y_\sfr^{k+1}&=g(x_\sfr^k),  \quad
    x_\sfr^{k+1}=\sum_{\sfc=1}^\sfC A_{\sfc\sfr}f
    \left(
    \sum_{\sfr'=1}^\sfR A_{\sfc\sfr'}y_{\sfr'}^{k+1}
    \right).
\end{split}
\label{eq:gen_scalar_recursion}
\end{align}
The recursion is initialized with $x_\sfr^{0}=x_{\max}$ for $\sfr\in[\sfR]$. This initialization, along with the monotonicity of $f$ and $g$, ensures that the coupled recursion converges to a fixed point \cite{Yed14}. The fixed point $\{\lim_{k\rightarrow\infty}x_\sfr^k\}_{\sfr\in[\sfR]}$ is characterized by the lemma below in terms of the following potential function:
\begin{align}
    V(x):=xg(x)-\int_0^xg(z)dz-\int_0^{g(x)}f(z)dz.
    \label{eq:gen_pot_fn}
\end{align}
\begin{lemma} \textup{\cite[Theorem 1]{Yed14}} \label{lem:Yelda_thm}
    For any $\gamma>0$, there exists $\omega_0<\infty$ such that for all $\omega>\omega_0$ and $\sfC\in[1,\infty]$, the fixed point $x_\sfr^\infty:=\lim_{k\rightarrow\infty}x_\sfr^k$, for $\sfr\in[\sfR]$, of the coupled recursion in \eqref{eq:gen_scalar_recursion} satisfies the upper bound
    \begin{align}
        \max_{\sfr\in[\sfR]}x_\sfr^\infty
        \leq\max\left\{\argmin_{x\in\mathcal{X}}V(x)\right\}+\gamma.
    \end{align}
\end{lemma}
\paragraph{Analyzing state evolution using Lemma \ref{lem:Yelda_thm}.}
Let us define the function 
$$
\mmse(s)= \E\left[ \Big( \bar{\beta} - \E[\bar{\beta} \mid \sqrt{s}\, \bar{\beta}+G]  \Big)^2 \right],
$$
where $G\sim\normal(0,1)$ is independent of $\bar{\beta}$. Then, recalling the definition of $f_k$ from \eqref{eq:optimal_fk}, the state evolution 
recursion in \eqref{eq:SC_SE_param1} is:
\begin{align}
    (\chi^{k+1}_{\sfc})^2
    &=\sum_{\sfr=1}^\sfR\tW_{\sfr\sfc}
    \Big(
    \underbrace{
    \sigma^2+\frac{1}{\deltain}\sum_{\sfc=1}^\sfC\tW_{\sfr\sfc}
    \underbrace{\mmse\big((\chi_{\sfc}^{k})^2\big)}_{=:\psi_\sfc^k}
    }_{=:\phi_\sfr^k}
    \Big)^{-1}.
    \label{eq:SE_recursion_sigma_form}
\end{align}
Using the definitions above, the state evolution recursion can be rewritten as:
\begin{align}
    \phi_\sfr^k
    &=\sigma^2+\frac{1}{\deltain}\sum_{\sfc=1}^\sfC\tW_{\sfr\sfc}\psi_\sfc^k, 
    \quad
    \psi_\sfc^{k+1}=\mmse\left(\sum_{\sfr=1}^\sfR\tW_{\sfr\sfc}\big(\phi_\sfr^k\big)^{-1}\right),
    \label{eq:SE_phi_psi}
\end{align}
which  combined into one equation gives:
\begin{align*}
    \phi_\sfr^{k+1}
    &=\sigma^2+\frac{1}{\deltain}
    \underbrace{
    \sum_{\sfc=1}^\sfC\tW_{\sfr\sfc}\mmse
    \left(\sum_{\sfr'=1}^\sfR\tW_{\sfr'\sfc}\big(\phi_{\sfr'}^k\big)^{-1}\right)
    }_{=:x_\sfr^{k+1}}.
\end{align*}
Rewriting the recursion in terms of $x_\sfr^{k+1}$ defined above, we get:
\begin{align}
    x_\sfr^{k+1}
    =\sum_{\sfc=1}^\sfC\tW_{\sfr\sfc}\mmse
    \left(
    \sum_{\sfr'=1}^\sfR\tW_{\sfr'\sfc}\Big(\sigma^2+\frac{x_\sfr^k}{\deltain}\Big)^{-1}
    \right), \quad \sfr \in [\sfR].
    \label{eq:modified_SE_recursion}
\end{align}

The modified recursion in \eqref{eq:modified_SE_recursion} is an instance of the coupled recursion in \eqref{eq:gen_scalar_recursion}, which can be seen by taking $A=\tW^\top$ and
\begin{align*}
    f(y)
    &=\mmse\Big(\frac{1}{\sigma^2}-y\Big),
    \quad
    g(x)=\frac{1}{\sigma^2}-\frac{1}{\sigma^2+x/\deltain}.
\end{align*}
It is shown in \cite[Section VI.E]{Yed14} that with these functions, which satisfy the assumptions stated at the start of this section, the potential function $V(x)$ in \eqref{eq:gen_pot_fn} equals $U(b;\deltain)$ defined in \eqref{eq:U_function} (upto an additive constant).
Invoking Lemma \ref{lem:Yelda_thm}, we have that the fixed points of \eqref{eq:modified_SE_recursion}, denoted by $(x_\sfr^\infty)_{\sfr \in [\sfR]}$ satisfy:
\begin{align}
    \max_{\sfr\in[\sfR]} \, x_\sfr^\infty
    &\leq\max\left\{\argmin_{x\in[0, \mmse(0)]}V(x)\right\}+\gamma
    =\max\left\{\argmin_{b\in[0,\Var(\bar{\beta})]}U(b;\deltain)\right\}+\gamma,
    \label{eq:xr_inf_bound}
\end{align}
where the last equality uses the fact that $\mmse(0)=\Var[\bar{\beta}]$. 

We now use the bound on $x_\sfr^\infty$ to upper bound the asymptotic MSE. Using \eqref{eq:mean_sq_error}, the asymptotic MSE (as $k \to \infty$) can be written as $\frac{1}{\sfC}\sum_{\sfc=1}^\sfC \mmse\left( (\chi^\infty_\sfc)^2\right)$ which can be further written as $\frac{1}{\sfC}\sum_{\sfc=1}^\sfC\psi_\sfc^\infty$ using \eqref{eq:SE_phi_psi}. From \eqref{eq:SE_phi_psi} and \eqref{eq:modified_SE_recursion}, we can write $x_\sfr^\infty=\sum_{\sfc=1}^\sfC\tW_{\sfr\sfc}\psi_\sfc^\infty$, which can be written more explicitly as
\begin{align}
    \begin{bmatrix}
        x_1^\infty \\
        x_2^\infty \\
        \vdots \\
        x_\omega^\infty \\
        x_{\omega+1}^\infty \\
        \vdots \\
        x_\Lambda^\infty \\
        x_{\Lambda+1}^\infty \\
        \vdots \\
        x_{\Lambda+\omega-1}^\infty
    \end{bmatrix}
    =\begin{bmatrix}
        \frac{1}{\omega} & 0 & \dots & 0 & 0 \\
        \frac{1}{\omega} & \frac{1}{\omega} & \dots & 0 & 0 \\
        \vdots & \vdots & & \vdots & \vdots \\
        \frac{1}{\omega} & \frac{1}{\omega} & \dots & 0 & 0 \\
        0 & \frac{1}{\omega} & \dots & 0 & 0 \\
        \vdots & \vdots & & \vdots & \vdots \\
        0 & 0 & \dots & \frac{1}{\omega} & \frac{1}{\omega} \\
        0 & 0 & \dots & \frac{1}{\omega} & \frac{1}{\omega} \\
        \vdots & \vdots & & \vdots & \vdots \\
        0 & 0 & \dots & 0 & \frac{1}{\omega}
    \end{bmatrix}
    \begin{bmatrix}
        \psi_1^\infty \\
        \psi_2^\infty \\
        \vdots \\
        \psi_{\Lambda-1}^\infty \\
        \psi_\Lambda^\infty
    \end{bmatrix}
    =\begin{bmatrix}
        \frac{1}{\omega}\psi_1^\infty \\
        \frac{1}{\omega}(\psi_1^\infty+\psi_2^\infty) \\
        \vdots \\
        \frac{1}{\omega}(\psi_1^\infty+\psi_2^\infty+\dots+\psi_\omega^\infty) \\
        \frac{1}{\omega}(\psi_2^\infty+\psi_3^\infty+\dots+\psi_{\omega+1}^\infty) \\
        \vdots \\
        \frac{1}{\omega}(\psi_{\Lambda-\omega+1}^\infty+\psi_{\Lambda-\omega+2}^\infty+\dots+\psi_\Lambda^\infty) \\
        \frac{1}{\omega}(\psi_{\Lambda-\omega+2}^\infty+\psi_{\Lambda-\omega+3}^\infty+\dots+\psi_\Lambda^\infty) \\
        \vdots \\
        \frac{1}{\omega}\psi_\Lambda^\infty
    \end{bmatrix}.
    \label{eq:explicit_x_formula}
\end{align}
For notational convenience, let us denote
$x^*
=\max\left\{\argmin_{b\in[0,\Var(\bar{\beta})}U(b;\deltain)\right\}+\gamma$.
From \eqref{eq:xr_inf_bound}, we have that $x_\sfr^\infty \leq x^*$ for all $\sfr\in[\sfR]$ where $\sfR=\Lambda+\omega-1$. In the rightmost vector in \eqref{eq:explicit_x_formula}, we observe that each entry  contains the sum of at most $\omega$ consecutive terms. This implies that
\begin{align}
    \sum_{\sfc'=\sfc}^{\sfc+\omega-1}\psi_{\sfc'}^\infty
    \leq x^*\omega, \quad \sfc \in [\sfC].
    \label{eq:bound_of_omega_group}
\end{align}
Recalling that $\sfC =\Lambda$ and dividing the elements of $[\psi_1^\infty,\dots,\psi_\Lambda^\infty]$ into groups of non-intersecting consecutive terms -- with index groups $[1:\omega],[\omega+1:2\omega],\dots,[\Lambda-\omega+1:\Lambda]$ -- gives us at most $\lceil\frac{\Lambda}{\omega}\rceil$ disjoint groups, with the sum of each group having an upper bound of $x^*\omega$ by \eqref{eq:bound_of_omega_group}. Hence, the asymptotic MSE can be bounded as
\begin{align*}
    \frac{1}{\sfC}\sum_{\sfc=1}^\sfC\psi_\sfc^\infty
    \leq\frac{1}{\Lambda}\left\lceil\frac{\Lambda}{\omega}\right\rceil x^*\omega
    <  \frac{1}{\Lambda}\Big(\frac{\Lambda}{\omega}+1\Big)x^*\omega
    =\frac{\Lambda+\omega}{\Lambda}x^*.
\end{align*}
This completes the proof of the first part of Theorem \ref{thm:asym_MSE}. 

\subsection{Proof of \eqref{eq:asym_MSE_AMP}}

For the i.i.d.~design, we have $\sfR= \sfC=1$ and $\tW_{11}=1$, so the state evolution reduces to 
\begin{align*}
    x^{k+1}
    &=\mmse\left(\Big(\sigma^2+\frac{x^k}{\delta}\Big)^{-1}\right),
\end{align*}
with the initialization $x^0= \mmse(0)=\Var(\bar{\beta})$. Since $\mmse(s)$ is strictly decreasing in $s \in [0, \infty)$, the sequence $(x^k)$ is monotonically decreasing in $k$, and since it is bounded below, it converges to a fixed point. Since the recursion is initialized at $x^0=\Var(\bar{\beta})$, the fixed point is given by the largest solution of $x=\mmse\left(\Big(\sigma^2+\frac{x}{\delta}\Big)^{-1}\right)$.  Finally, we observe that the same equation is obtained by setting $\partial_1U(b;\delta)=0$. This completes the proof of \eqref{eq:asym_MSE_AMP}.  \qed

\section{Discussion and Future Directions}

We have shown that for noiseless QGT and pooled data,   a spatially coupled Bernoulli test design with an AMP recovery algorithm achieves almost-exact recovery with $n=o(p)$ tests. A key open question is to determine how  $n$  scales with $p$ for almost-exact recovery with SC-AMP.  Deriving this scaling is beyond the reach of our asymptotic analysis, which requires that $n/p  \to \delta >0$, but recent non-asymptotic analyses of AMP \cite{LiFanWei_Z2, Bao23} might provide tools to address this question, and allow comparisons with the information-theoretic  bound of $n^* =  \gamma^* \frac{p}{\log p}$ (see \eqref{eq:gamma_star}).

Another open question is to determine the number of tests required for \emph{exact} recovery in the linear regime  for an efficient scheme with a random design. We recall that exact recovery requires $\mb{P}\big[\tbeta\neq\beta\big]\rightarrow0$ as $p \to \infty$, in contrast to the almost-exact recovery criterion in \eqref{eq:almost_exact_recovery}.

In this paper, the only assumption on the QGT signal vector  $\beta$  is that its empirical distribution converges to a Bernoulli distribution.  The items are not required to be independent, and in some applications there may be known correlations between the items. Although the current SC-AMP algorithm does not exploit correlations between signal entries, it can be adapted to do so, using non-separable denoising functions \cite{Ber20}. 

An interesting direction for future work is to study variants of QGT and pooled data with additional structure, e.g., constraints on the tests or side-information that captures correlations between the items. For example, in graph-constrained group testing \cite{Che12}, the items are vertices on a graph, and  items included in each test have to conform to constraints imposed by the graph. Recent work in Boolean group testing has also shown that exploiting correlations or community structure among the items can significantly improve testing efficiency \cite{Nik23,Ahn23, jain2024sparsity}. 
Designing efficient schemes for quantitative group testing in such structured settings is an open question.

\vspace{1in}

{\begin{center}  \Large \textbf{Appendix}  \end{center}}
\appendix

\section{Proof of Theorem \ref{thm:SC_GAMP}}

\subsection{Proof of Theorem \ref{thm:SC_GAMP} via Reduction to Abstract Matrix-AMP} \label{sec:reduction_to_abs_matrix-AMP}

We describe an abstract matrix-AMP iteration for which a state evolution result can be established, and then prove Theorem \ref{thm:SC_GAMP} by reducing the SC-GAMP algorithm to the abstract matrix-AMP. For $k\geq 0$, the abstract matrix-AMP produces iterates $H^{k+1}\in\mb{R}^{p\times l_H}$ and $E^{k+1}\in\mb{R}^{n
\times l_E}$ as follows:
\begin{align}
\begin{split}
    H^{k+1}
    &=\tX^\top \hR^k-\hH^k\cdot(\mathsf{D}^k)^\top,
    \quad
    \hR^k=\tg_k(E^k,\gamma,\mathcal{R}), 
    \quad
    \sfD^k=\frac{1}{\sfR}\sum_{\sfr=1}^\sfR\E[\tg_k'(\bar{E}_\sfr^k,\bar{\gamma},\sfr)], \\
    E^{k+1}&=\tX\hH^{k+1}-\hR^k\cdot(\sfB^{k+1})^\top,
    \quad
    \hH^{k+1}=\tf_{k+1}(H^{k+1},\beta,\mathcal{C}),
    \quad
    \sfB^{k+1}=\frac{1}{\delta\sfC}\sum_{\sfc=1}^\sfC\E[\tf_{k+1}'(\bH_{\sfc}^{k+1},\bar{\beta},\sfc)],
\end{split}
\label{eq:abs_matrix_AMP}
\end{align}
where $\beta\in\mb{R}^p$, $\gamma\in\mb{R}^n$, and $\mathcal{C}$ and $\mathcal{R}$ are defined in \eqref{eq:S_c_and_S_r_def}. The functions $\tf_{k+1}:\mb{R}^{l_H}\times\mb{R}\times[\sfC]\rightarrow\mb{R}^{l_E}$ and $\tg_k:\mb{R}^{l_E}\times\mb{R}\times[\sfR]\rightarrow\mb{R}^{l_H}$ act row-wise on their inputs, and $\tf'_{k+1},\tg'_k$ denote the Jacobians with respect to their first arguments. The joint laws of $(\bar{E}_\sfr^k,\bar{\gamma})$ and $(\bH_{\sfc}^{k+1},\bar{\beta})$ are described later (below \eqref{eq:abs_matrix_AMP_SE}).
The algorithm is initialized with $\hH^0\in\mb{R}^{p\times l_H}$ and $E^0=\tX \hH^0\in\mb{R}^{n\times l_E}$. 

We have the following assumptions for the abstract matrix-AMP algorithm.
\begin{itemize}
    \item[\textbf{(B1)}] As dimensions $p,n\rightarrow\infty$, the ratio $n/p\rightarrow\delta>0$. Furthermore, $l_E$, $l_H$, $\sfR$, and $\sfC$ are positive integers that do not scale with $p$ as $n,p\rightarrow\infty$.
    \item[\textbf{(B2)}] Almost surely for all $\sfc\in[\sfC]$, as $n,p\rightarrow\infty$, $(\beta_\sfc,\hH_{\mathcal{J}_\sfc,:}^0)\stackrel{W}{\rightarrow}(\bar{\beta},\bar{H}_\sfc^0)$ and $\gamma_\sfr\stackrel{W}{\rightarrow}\bar{\gamma}$, with the joint law of $(\bar{\beta}, \, \bar{H}_\sfc^0) \in \reals \times \reals^{l_E}$ having finite moments of all orders, where $\mathcal{J}_{\sfc}$ is defined in \eqref{eq:J_c_I_r_def}. 
    Multivariate polynomials are dense in the real $L^2$-spaces of functions $f:\mb{R}^{2}\rightarrow\mb{R}$ and $g:\mb{R}^{l_E+2}\rightarrow\mb{R}$ with the inner products
    \begin{align*}
        \langle f,\tilde{f}\rangle
        :=\E[f(\bar{\gamma},\sfr)\tilde{f}(\bar{\gamma},\sfr)]
        \quad
        \text{and}
        \quad
        \langle g, \tilde{g}\rangle
        :=\E[g(\bar{\beta},\bar{H}_\sfc^0,\sfc)\tilde{g}(\bar{\beta},\bar{H}_\sfc^0,\sfc)],
    \end{align*}
    for all $\sfr\in[\sfR]$ and $\sfc\in[\sfC]$.
    \item[\textbf{(B3)}] For $k\geq0$, the functions $\tf_{k+1}$ and $\tg_k$ are continuous, Lipschitz w.r.t.~their first argument, and satisfy the polynomial growth condition in \eqref{eq:poly_growth_cond} for some order $r\geq 1$.
    \item[\textbf{(B4)}] $\tX$ is a generalized white noise matrix where $\|\tX\|_{\text{op}}<C$ almost surely for sufficiently large $n,p$ for some constant $C$. For any fixed polynomial functions $f^\dag:\mb{R}^{l_E+2}\rightarrow\mb{R}$ and $f^{\ddag}:\mb{R}^{2}\rightarrow\mb{R}$, as $n,p\rightarrow\infty$,
    \begin{align*}
        \max_{i\in[n]}\left|\left\langle f^\dag(\beta_\sfc,\hH_{\mathcal{J}_\sfc,:}^0,\sfc)\odot S_{i,\mathcal{J}_\sfc}\right\rangle-\left\langle f^\dag(\beta_\sfc,\hH_{\mathcal{J}_\sfc,:}^0,\sfc)\right\rangle\cdot \left\langle S_{i,\mathcal{J}_\sfc}\right\rangle\right|
        &\stackrel{a.s.}{\rightarrow}
        0, \\
        \max_{j\in[p]}\left|\left\langle f^\ddag(\gamma_\sfr,\sfr)\odot S_{\mathcal{I}_\sfr,j}\right\rangle-\left\langle f^\ddag(\gamma_\sfr,\sfr)\right\rangle\cdot \left\langle S_{\mathcal{I}_\sfr,j}\right\rangle\right|
        &\stackrel{a.s.}{\rightarrow}
        0,
    \end{align*}
    for all $\sfc\in[\sfC]$ and $\sfr\in[\sfR]$, where $S$ is the variance profile of $\tX$ (see Definition \ref{def:gen_white_noise_matrix}).
\end{itemize}

\paragraph{State evolution.} The state evolution parameters for $k\geq0$ are
\begin{align}
\begin{split}
    \Omega^{k+1}
    &=\frac{1}{\sfR}\sum_{\sfr=1}^\sfR\widehat{\Omega}^{k+1,\sfr},\quad
    \widehat{\Omega}^{k+1,\sfr}
    =\E\big[\tg_k(\bar{E}_\sfr^k,\bar{\gamma},\sfr)\tg_k(\bar{E}_\sfr^k,\bar{\gamma},\sfr)^\top\big]\in\mb{R}^{l_H\times l_H}, \\
    \Pi^{k+1}
    &=\frac{1}{\sfC}\sum_{\sfc=1}^\sfC\widehat{\Pi}^{k+1,\sfc},\quad
    \widehat{\Pi}^{k+1,\sfc}=\frac{1}{\delta}\E\big[\tf_{k+1}(\bH_{\sfc}^{k+1},\bar{\beta},\sfc)\tf_{k+1}(\bH_{\sfc}^{k+1},\bar{\beta},\sfc)^\top\big]\in\mb{R}^{l_E\times l_E},
\end{split}
\label{eq:abs_matrix_AMP_SE}
\end{align}
with $\bar{E}_\sfr^k\sim\normal(0,\Pi^k)$ independent of $\bar{\gamma}$, and $\bH_{\sfc}^{k+1}\sim\normal(0,\Omega^{k+1})$ independent of $\bar{\beta}$. The state evolution is initialized with 
\begin{align}
    \Pi^0=\frac{1}{\sfC}\sum_{\sfc=1}^\sfC\widehat{\Pi}^{0,\sfc}, \quad \text{ where } \quad 
    \widehat{\Pi}^{0,\sfc} = \frac{1}{\delta}\E\left[ \bar{H}_\sfc^0 
    (\bar{H}_\sfc^0)^\top \right].
    \label{eq:Pi0c_def}
\end{align}

\begin{theorem}[State evolution for abstract matrix-AMP] \label{thm:abs_matrix_AMP}
Consider the abstract matrix-AMP in \eqref{eq:abs_matrix_AMP} with the assumptions (B1)--(B4) being satisfied. For  $k\geq1$, and for $\sfr\in[\sfR], \sfc\in[\sfC]$, the iterates of the abstract matrix AMP satisfy
\begin{align*}
    \big(H_{\mathcal{J}_\sfc,:}^k,\,\beta_\sfc\big)
    &\stackrel{W_2}{\rightarrow}
    \big(\bar{H}_\sfc^k,\,\bar{\beta}\big), 
    \quad
    \big(E_{\mathcal{I}_\sfr,:}^k,\,\gamma_\sfr\big)
    \stackrel{W_2}{\rightarrow}
    \big(\bar{E}_\sfr^k,\,\bar{\gamma}\big),
\end{align*}
where $\bar{H}_\sfc^k$ is independent of $\bar{\beta}$, and $\bar{E}_\sfr^k$ is independent of $\bar{\gamma}$.
\end{theorem}

Theorem \ref{thm:abs_matrix_AMP} is proved in Section \ref{sec:reduction_to_U-AMP}. 

\paragraph{Proof of Theorem \ref{thm:SC_GAMP} using  Theorem \ref{thm:abs_matrix_AMP}.} We reduce the SC-GAMP algorithm to the abstract matrix-AMP iteration. As given in \eqref{eq:SC_LM_matrix}, we can obtain $\tX$ from $\tXsc$, which is a generalized white noise matrix (see Definition \ref{def:gen_white_noise_matrix}) with variance profile $S_{ij}=1$ for all $(i,j)\in[n]\times[p]$. Next, we set $\gamma:=\tPsi$,  the same $\beta$ for both algorithms, and the functions $\tf_k:\mb{R}^{\sfC}\times\mb{R}\times[\sfC]\rightarrow\mb{R}^{2\sfR}$ and $\tg_k:\mb{R}^{2\sfR}\times\mb{R}\times[\sfR]\rightarrow\mb{R}^\sfC$ as follows:
\begin{align*}
    \tf_k(H_{j,:}^k,\beta_j,\sfc)
    &=\left[\beta_j\left(\sqrt{\sfR\tW_{1\sfc}},\dots,\sqrt{\sfR\tW_{\sfR\sfc}}\right),
    \,\,
    f_k(H_{j\sfc}^k+\mu_{\beta,\sfc}^{k}\beta_j,\sfc)\left(\sqrt{\sfR\tW_{1\sfc}},\dots,\sqrt{\sfR\tW_{\sfR\sfc}}\right)\right],
\end{align*}
for $j\in\mathcal{J}_\sfc$ and $H^k_{j,:}\in\mb{R}^\sfC$ (i.e., $l_H=\sfC$). We also set
\begin{align*}
    \tg_k(E_{i,:}^k,\gamma_i,\sfr)
    &=g_k(E_{i\sfr}^k, \, q(E_{i,\sfr+\sfR},\tPsi_i),\, \sfr)\left(\sqrt{\sfR\tW_{\sfr 1}},\dots,\sqrt{\sfR\tW_{\sfr\sfC}}\right),
\end{align*}
for $i\in\mathcal{I}_\sfr$ and $E_{i,:}^k\in\mb{R}^{2\sfR}$ (i.e., $l_E=2\sfR$). The abstract matrix-AMP iteration is initialized with
\begin{align*}
    \hH_{j,:}^0
    &=
    \left[
    \beta_j\left(\sqrt{\sfR\tW_{1\sfc}},\dots,\sqrt{\sfR\tW_{\sfR\sfc}}\right), \, 
    \hbeta_j^0\left(\sqrt{\sfR\tW_{1\sfc}},\dots,\sqrt{\sfR\tW_{\sfR\sfc}}\right)
    \right], \quad \text{ for } j \in \mathcal{J}_\sfc.
\end{align*}
The state evolution parameters $\Pi^k\in\mb{R}^{2\sfR\times 2\sfR}$ and $\Omega^k\in\mb{R}^{\sfC\times \sfC}$ are recursively computed as follows. We have $\bE_{\sfr}^k\sim\normal(0,\Pi^k)$ independent of $\bar{\Psi}$, and the entries of $\widehat{\Omega}^{k+1,\sfr}$ are
$$
\widehat{\Omega}^{k+1,\sfr}_{\sfc\sfc'}
=\sfR\cdot\E\left[g_k(\bE_{\sfr,\sfr}^k,q(\bE_{\sfr,\sfr+\sfR},\bar{\Psi}),\sfr)^2\right]\sqrt{\tW_{\sfr\sfc}\tW_{\sfr\sfc'}}, \quad
\text{for $\sfc,\sfc'\in[\sfC]$}.
$$
Next, we have $\bH_{\sfc}^k\sim\normal(0,\Omega^k)$ independent of $\bar{\beta}$ and
\begin{align*}
    \widehat{\Pi}^{k+1,\sfc}_{rs}
    &=
    \begin{cases}
        \frac{1}{\delta}\sfR\E[\bar{\beta}^2]\sqrt{\tW_{r\sfc}\tW_{s\sfc}} &r,s\in[\sfR], \\
        \frac{1}{\delta}\sfR\E[\bar{\beta}f_k(\{\bH_\sfc^k\}_\sfc+\mu_{\beta,\sfc}^{k+1}\bar{\beta},\sfc)]\sqrt{\tW_{r\sfc}\tW_{(r-\sfR)\sfc}} &r\in[\sfR],\sfR+1\leq r\leq2\sfR,\\
        \frac{1}{\delta}\sfR\E[\bar{\beta}f_k(\{\bH_\sfc^k\}_\sfc+\mu_{\beta,\sfc}^{k+1}\bar{\beta},\sfc)]\sqrt{\tW_{(r-\sfR)\sfc}\tW_{s\sfc}} &\sfR+1\leq r\leq2\sfR,s\in[\sfR], \\
        \frac{1}{\delta}\sfR\E[f_k(\{\bH_\sfc^k\}_\sfc+\mu_{\beta,\sfc}^{k+1}\bar{\beta},\sfc)^2]\sqrt{\tW_{(r-\sfR)\sfc}\tW_{(s-\sfR)\sfc}} &\sfR+1\leq r,s\leq 2\sfR.
    \end{cases}
\end{align*}
The state evolution is initialized with 
$$
\widehat{\Pi}^{0,\sfc}
=\frac{1}{\delta}\lim_{p\rightarrow\infty}
\frac{1}{p/\sfC}(\hH_{\mathcal{J}_\sfc,:}^0)^\top\hH_{\mathcal{J}_\sfc,:}^0,
$$
for $\sfc\in[\sfC]$, with $\hH_{\mathcal{J}_\sfc,:}^0\in\mb{R}^{p/\sfC\times 2\sfR}$. By assumption (A2) (see \eqref{eq:cov_init}), the entries of $\widehat{\Pi}^{0,\sfc}$ are given by
\begin{align*}
    \widehat{\Pi}_{rs}^{0,\sfc}
    &=
    \begin{cases}
        \frac{1}{\delta}\sfR\E[\bar{\beta}^2]\sqrt{\tW_{r\sfc}\tW_{s\sfc}} &r,s\in[\sfR], \\
        \frac{1}{\delta}\E[\bar{\beta}\bar{\beta}_\sfc^0]\sqrt{\tW_{r\sfc}\tW_{(s-\sfR)\sfc}} & r\in[\sfR],\sfR+1\leq s\leq2\sfR, \\
        \frac{1}{\delta}\E[\bar{\beta}\bar{\beta}_\sfc^0]\sqrt{\tW_{(r-\sfR)\sfc}\tW_{s\sfc}} & \sfR+1\leq s\leq2\sfR,s\in[\sfR],\\
        \frac{1}{\delta}\E[(\bar{\beta}_\sfc^0)^2]\sqrt{\tW_{(r-\sfR)\sfc}\tW_{(s-\sfR)\sfc}} & \sfR+1\leq s,r\leq 2\sfR.
    \end{cases}
\end{align*}
We then have $\Pi^0=\frac{1}{\sfC}\sum_{\sfc=1}^\sfC\widehat{\Pi}^{0,\sfc}$.

We can then show that
\begin{align}
    \begin{bmatrix}
        \Pi_{\sfr\sfr}^k & \Pi_{\sfr(\sfr+\sfR)}^k \\
        \Pi_{(\sfr+\sfR)\sfr}^k & \Pi_{(\sfr+\sfR)(\sfr+\sfR)}
    \end{bmatrix}
    =
    \Sigma^{k,\sfr}
    \quad
    \text{and}
    \quad
    \Omega_{\sfc\sfc}^{k+1}
    =(\sigma_{\beta,\sfc}^{k+1})^2,
    \label{eq:reduction_assum1}
\end{align}
implying that $\big(\{\bE_\sfr^k\}_\sfr,\{\bE_{\sfr}^k\}_{\sfr+\sfR}\big)\stackrel{d}{=}\big(Z_\sfr,Z_\sfr^k\big)$ and $\{\bH_\sfc^k\}_\sfc\stackrel{d}{=}G_{\beta,\sfc}^k$. We can also show that for $k\geq0$,
\begin{align}
\begin{split}
    E_{i,\sfr}^k=\Theta_i^k, \quad
    E_{i,\sfr+\sfR}^k=\Theta_i, \quad
    &\text{for $i\in\mathcal{I}_\sfr,\,\sfr\in[\sfR]$,} \\
    H_{j,\sfc}^{k+1}+\mu_{\beta,\sfc}^{k+1}\beta_j=\beta_j^{k+1}, \quad
    &\text{for $j\in\mathcal{J}_\sfc,\,\sfc\in[\sfC]$.}
\end{split}
\label{eq:reduction_assum2}
\end{align}
Both \eqref{eq:reduction_assum1} and \eqref{eq:reduction_assum2} are shown using steps identical to those in \cite[Section 5.1.2]{Cob23}, so we omit repeating the proof for brevity. Theorem \ref{thm:SC_GAMP} follows by using \eqref{eq:reduction_assum1} and \eqref{eq:reduction_assum2} in Theorem \ref{thm:abs_matrix_AMP}.

\subsection{Proof of Theorem \ref{thm:abs_matrix_AMP} via Reduction to U-AMP} \label{sec:reduction_to_U-AMP}

Theorem \ref{thm:abs_matrix_AMP} is proved by reducing the abstract matrix-AMP recursion to the U-AMP recursion which is defined as follows. Given a generalized white noise matrix $\tX$, for $t\geq 1$, the iterates of U-AMP, denoted by $h^t\in\mb{R}^{p}$ and $e^t\in\mb{R}^n$, are produced using functions $f_t^v:\mb{R}^{t+L_d+1}\rightarrow\mb{R}$, $f_{t+1}^u:\mb{R}^{t+L_c+1}\rightarrow\mb{R}$.
Given an initializer $u^1\in\mb{R}^n$, side information vectors $c^1,\dots,c^{L_c}\in\mb{R}^n$ and $d^1,\dots,d^{L_d}\in\mb{R}^p$, all independent of $\tX$, the iterates of the U-AMP recursion are computed as:
\begin{align}
\begin{split}
    h^t&=\sqrt{\delta}\tX^\top u^t-\sum_{s=1}^{t-1}b_{s}^tv^s, \qquad 
    v^t=f_t^v(h^1,\dots,h^t,d^1,\dots,d^{L_d},\mathcal{C}), \\
    e^t&=\sqrt{\delta}\tX v^t-\sum_{s=1}^ta_{s}^tu^s, \qquad
    u^{t+1}=f_{t+1}^u(e^1,\dots,e^t,c^1,\dots,c^{L_c},\mathcal{R}),
\end{split} \label{eq:modified_U-AMP}
\end{align}
where $\mathcal{C},\mathcal{R}$ were defined in \eqref{eq:S_c_and_S_r_def}, and the functions $f_t^v$ and $f_{t+1}^u$ act row-wise. The coefficients $a^t_s$ and $b^t_s$ are defined later in \eqref{eq:as_bs_UAMP} in terms of state evolution parameters.

Recalling the notation simplification for sub-blocks of vectors presented in the paragraph below \eqref{eq:Qc_def},  we have the following assumptions:
\begin{itemize}
    \item[\textbf{(C1)}] As $n,p\rightarrow\infty$, we have $n/p=\delta>0$, for fixed $L_c$ and $L_d$. Furthermore, for all $\sfc\in[\sfC]$ and $\sfr\in[\sfR]$, 
    \begin{align*}
        (u_\sfr^1,c_\sfr^1,\dots,c_\sfr^{L_c})
        \stackrel{W}{\rightarrow}
        (\bu_\sfr^1,\bc_\sfr^1,\dots,\bc_\sfr^{L_c})
        \text{ and }
        (d_\sfc^1,\dots,d_\sfc^{L_d})
        \stackrel{W}{\rightarrow}
        (\bd_\sfc^1,\dots,\bd_\sfc^{L_d}),
    \end{align*}
    for joint limit laws $(\bu_\sfr^1,\bc_\sfr^1,\dots,\bc_\sfr^{L_c})$ and $(\bd_\sfc^1,\dots,\bd_\sfc^{L_d})$ having finite moments of all orders, where $\E[(\bu^1)^2]\geq 0$. Multivariate polynomials are dense in the real $L^2$-spaces of functions $f:\mb{R}^{L_c+1}\rightarrow\mb{R}$ and $g:\mb{R}^{L_d}\rightarrow\mb{R}$ with the inner products
    \begin{align*}
        \langle f,\tf\rangle
        :=\E[f(\bu_\sfr^1,\bc_\sfr^1,\dots,\bc_\sfr^{L_c})\tf(\bu_\sfr^1,\bc_\sfr^1,\dots,\bc_\sfr^{L_c})]
        \text{ and }
        \langle g,\tg\rangle
        :=\E[g(\bd_\sfc^1,\dots,\bd_\sfc^{L_d})\tg(\bd_\sfc^1,\dots,\bd_\sfc^{L_d})].
    \end{align*}
    \item[\textbf{(C2)}] Each function $f_t^v:\mb{R}^{t+L_d+1}\rightarrow\mb{R}$ and $f_{t+1}^u:\mb{R}^{t+L_c+1}\rightarrow\mb{R}$ is continuous, is Lipschitz in its first $t$ arguments, and satisfies the polynomial growth condition in \eqref{eq:poly_growth_cond} for some order $r\geq1$.
    \item[\textbf{(C3)}] $\|\tX\|_{\text{op}}<C$, for some constant $C$ almost surely for all sufficiently large $n$ and $p$.
    \item[\textbf{(C4)}] For any fixed polynomial functions $f^\dag:\mb{R}^{L_c+1}\rightarrow\mb{R}$ and $f^\ddag:\mb{R}^{L_d}\rightarrow\mb{R}$, almost surely as $n,p\rightarrow\infty$,
    \begin{align*}
        \max_{j\in[p]}
        \left|
        \langle f^\dag(u^1,c^1,\dots,c^{L_c})\odot S_{:,j}\rangle
        -
        \langle f^\dag(u^1,c^1,\dots,c^{L_c})\rangle
        \cdot
        \langle S_{:,j}\rangle
        \right|
        &\rightarrow0, \\
        \max_{i\in[n] }\left|
        \langle f^\ddag(d^1,\dots,d^{L_d})\odot S_{i,:}\rangle
        -
        \langle f^\ddag(d^1,\dots,d^{L_d})\rangle
        \cdot
        \langle S_{i,:}\rangle
        \right|
        &\rightarrow 0.
    \end{align*}
\end{itemize}

\paragraph{State evolution.} The state evolution result below states that the joint empirical distribution of $(h_\sfc^1,\dots,h_\sfc^t)$ converges to a Gaussian law $\normal(0,\Xi^t)$, for $\sfc \in [\sfC]$. Similarly, the joint empirical distribution of $(e_\sfr^1,\dots,e_\sfr^t)$ converges to $\normal(0,\Gamma^t)$, for $\sfr \in [\sfR]$. The covariance matrices $\Xi^t,\Gamma^t\in\mb{R}^{t\times t}$ are iteratively defined as follows, starting from $\Xi^1=\frac{\delta}{\sfR}\sum_{\sfr=1}^\sfR \E[(\bu_\sfr^1)^2]$. 
%where  $\widehat{\Xi}^{1,\sfr}=\delta\E[(\bu_\sfr^1)^2]\in\mb{R}^{1\times 1}$ for $\sfr\in[\sfR]$. 
Given $\Xi^t$, for $t\geq1$, let $(\bh^1,\dots,\bh^t)\sim\normal(0,\Xi^t)$ be independent of $(\bd_\sfc^1,\dots,\bd_\sfc^{L_d})$ and define
\begin{align}
  \bv_\sfc^s=f_s^v(\bh^1,\dots,\bh^s,\bd_\sfc^1,\dots,\bd_\sfc^{L_d}, \sfc),
\quad s\in[t],\, \sfc\in[\sfC].  
\label{eq:bar_vsc_def}
\end{align}
Then we have
\begin{align}
    \Gamma^t
    &=\frac{1}{\sfC}\sum_{\sfc=1}^\sfC\widehat{\Gamma}^{t,\sfc}
    \quad
    \text{where}
    \quad
    \widehat{\Gamma}^{t,\sfc}
    =\big(\E[\bv_\sfc^r\bv_\sfc^s]\big)_{r,s=1}^t.
    \label{eq:Gamma_def_U_AMP}
\end{align}
Next, let $(\be^1,\dots,\be^t)\sim\normal(0,\Gamma^t)$ be independent of $(\bu_\sfr^1,\bc_\sfr^1,\dots,\bc_\sfr^{L_c})$ and define
$$
\bu_\sfr^{s+1}
=f_{s+1}^u(\be^1,\dots,\be^s,\bc_\sfr^1,\dots,\bc_\sfr^{L_c}, \sfr),
\quad
s\in[t],\,
\sfr\in[\sfR].
$$
Then, we have
\begin{align}
    \Xi^{t+1}
    =\frac{1}{\sfR}\sum_{\sfr=1}^\sfR\widehat{\Xi}^{t+1,\sfr}
    \quad
    \text{where}
    \quad
    \widehat{\Xi}^{t+1,\sfr}
    =\delta\big(\E[\bu_\sfr^r\bu_\sfr^s]\big)_{r,s=1}^{t+1}.
    \label{eq:Xi_def_U_AMP}
\end{align}
We define the memory coefficients $a_s^t$ and $b_s^t$ in \eqref{eq:modified_U-AMP} as:
\begin{align}
    a_s^t
    &=\frac{1}{\sfC}\sum_{\sfc=1}^\sfC\E\left[\partial_sf_t^v(\bh^1,\dots,\bh^t,\bd_\sfc^1.\dots,\bd_\sfc^{L_d},\sfc)\right],
    \quad
    b_s^t
    =\frac{\delta}{\sfR}\sum_{\sfr=1}^\sfR\E\left[\partial_sf_t^u(\be^1,\dots,\be^{t-1},\bc_\sfr^1,\dots,\bc_\sfr^{L_c},\sfr)\right],
    \label{eq:as_bs_UAMP}
\end{align}
where $\partial_s$ denotes the partial derivative with respect to the $s$th argument. The following corollary gives the state evolution result for the U-AMP recursion.

\begin{corollary}[State evolution for U-AMP] \label{cor:U-AMP}
Let $\tX\in\mb{R}^{n\times p}$ be a generalized white noise matrix (as defined in Definition \ref{def:gen_white_noise_matrix}) with variance profile $S\in\mb{R}^{n\times p}$, and let $u_\sfr^1,c_\sfr^1,\dots,c_\sfr^{L_c},d_\sfc^1,\dots,d_\sfc^{L_d}$ be independent of $\tX$ and satisfy Assumptions (C1)--(C4). Further assume that each matrix $\Xi^t$ and $\Gamma^t$ is non-singular. Then for any fixed $t\geq1$, as $n,p\rightarrow\infty$, the iterates of the abstract AMP in \eqref{eq:modified_U-AMP} almost surely satisfy the following, for $\sfr\in[\sfR]$ and $\sfc\in[\sfC]$:
\begin{align*}
    & (u_\sfr^1,c_\sfr^1,\dots,c_\sfr^{L_c},e_\sfr^1,\dots,e_\sfr^t)
    \stackrel{W_2}{\rightarrow}(\bu_\sfr^1,\bc_\sfr^1,\dots,\bc_\sfr^{L_c},\be^1,\dots,\be^t), \\
    & (d_\sfc^1,\dots,d_\sfc^{L_d},h_\sfc^1,\dots,h_\sfc^t)
    \stackrel{W_2}{\rightarrow}
    (\bd_\sfc^1,\dots,\bd_\sfc^{L_d},\bh^1,\dots,\bh^t),
\end{align*}
where $(\bh^1,\dots,\bh^t)\sim\normal(0,\Xi^t)$ and $(\be^1,\dots,\be^t)\sim\normal(0,\Gamma^t)$ are independent of $(\bd_\sfc^1,\dots,\bd_\sfc^{L_d})$ and $(\bu_\sfr^1,\bc_\sfr^1.\dots,\bc_\sfr^{L_c})$ respectively.
\end{corollary}

Corollary \ref{cor:U-AMP} is obtained from the AMP universality result in  \cite[Theorem 2.17]{Wan22}.  A proof  is provided in Section \ref{sec:mod_of_U-AMP}. 

\paragraph{Proof of Theorem \ref{thm:abs_matrix_AMP} using Corollary \ref{cor:U-AMP}.}
We  reduce the abstract matrix-AMP iteration to the U-AMP. We set the initializer to be $u^1=0$, $L_c=1$, $L_d=l_E+1$ and the side information vectors to be
\begin{align}
c^1=\gamma,
\quad
d^1=\beta,
\quad
d^2=\hH_{:,1}^0,
\quad
\dots,
\quad
d^{l_E+1}=\hH_{:,l_E}^0.
\label{eq:initializer_reduction}
\end{align}
We show the reduction through induction.

\paragraph{Base case.} We consider the case $t=0$, and our goal is to reduce $\hH^0$, $E^0$, $H^1$, $\hR^0$, $\hH^1$, and $E^1$ to iterates of U-AMP defined via careful choices of the functions $f_t^v$ and $f_{t+1}^u$. We provide a summary of the reductions before giving their derivations.
\begin{itemize}
    \item For $t=1,\dots,l_E$: we have
    \begin{align}
        h^1,\dots,h^{l_E}=0, \quad
        (v^1,\dots,v^{l_E})=\frac{1}{\sqrt{\delta}}\hH^0, \quad
        (e^1,\dots,e^{l_E})=E^0, \quad
        u^1,\dots,u^{l_E}=0.
        \label{eq:abs_matrix_AMP_reduction1}
    \end{align}
    \item For $t=l_E+1,\dots,l_E+l_H$: we have
    \begin{align}
    \begin{split}
        &(h^{l_E+1},\dots,h^{l_E+l_H})=H^1, \quad
        v^{l_E+1},\dots,v^{l_E+l_H}=0, \\
        &e^{l_E+1},\dots,e^{l_E+l_H}=0, \quad
        (u^{l_E+1},\dots,u^{l_E+l_H})=\frac{1}{\sqrt{\delta}}\hR^0.
    \end{split}
    \label{eq:abs_matrix_AMP_reduction2}
    \end{align}
    \item For $t=l_E+l_H+1,\dots,2l_E+l_H$: we have
    \begin{align}
    \begin{split}
        h^{l_E+l_H+1},\dots,h^{2l_E+l_H}=0, \quad
        (v^{l_E+l_H+1},\dots,v^{2l_E+l_H})=\frac{1}{\sqrt{\delta}}\hH^1, \\
        (e^{l_E+l_H+1},\dots,e^{2l_E+l_H})=E^1, \quad
        u^{l_E+l_H+1},\dots,u^{2l_E+l_H}=0.
    \end{split}
    \label{eq:abs_matrix_AMP_reduction3}
    \end{align}
\end{itemize}

We now provide the derivations of \eqref{eq:abs_matrix_AMP_reduction1}--\eqref{eq:abs_matrix_AMP_reduction3}. For $t=1,\dots,l_E-1$, we set
\begin{align*}
    f_t^v(h^1,\dots,h^t,d^1,\dots,d^{l_E+1}, \mathcal{C})
    &=\frac{1}{\sqrt{\delta}}d^{t+1}, \quad
    f_{t+1}^u(e^1,\dots,e^t,c^1,\mathcal{R})
    =0.
\end{align*}
For $t=1$, we have our initializer $u^1=0$, and using \eqref{eq:initializer_reduction} gives
\begin{align*}
    h^1=0,
    \quad
    v^1=\frac{1}{\sqrt{\delta}}\hH_{:,1}^0,
    \quad
    e^1=E_{:,1}^0,
    \quad
    u^2=0.
\end{align*}
Following similar steps, for $t=2,\dots,l_E-1$, we have
\begin{align*}
    h^t=0,
    \quad
    v^t=\frac{1}{\sqrt{\delta}}\hH_{:,t}^0,
    \quad
    e^t=E_{:,t}^0,
    \quad
    u^{t+1}=0.
\end{align*}
For $t=l_E$, set
\begin{align}
    f_{l_E}^v(h^1,\dots,h^t,d^1,\dots,d^{l_E+1},\mathcal{C})
    &=\frac{1}{\sqrt{\delta}}d^{l_E+1},
    \quad
    f_{l_E+1}^u(\underbrace{e^1,\dots,e^{l_E}}_{=E^0},\underbrace{c^1}_{=\gamma},\mathcal{R})
    =\frac{1}{\sqrt{\delta}}\{\tg_0(E^0,\gamma,\mathcal{R})\}_{:,1},
    \label{eq:fuv_LE}
\end{align}
which gives
\begin{align*}
    h^{l_E}=0,
    \quad
    v^{l_E}=\frac{1}{\sqrt{\delta}}\hH_{:,l_E}^0,
    \quad
    e^{l_E}=E_{:,l_E}^0,
    \quad
    u^{l_E+1}=\frac{1}{\sqrt{\delta}}\hR_{:,1}^0.
\end{align*}
This completes the derivation of \eqref{eq:abs_matrix_AMP_reduction1}. 
For $t=l_E+1,\dots,l_E+l_H-1$, we set
\begin{align}
    f_t^v(h^1,\dots,h^t,d^1,\dots,d^{l_E+1},\mathcal{C})
    &=0,
    \quad
    f_{t+1}^u(e^1,\dots,e^t,c^1,\mathcal{R})
    =\frac{1}{\sqrt{\delta}}\{\tg_0(E^0,\gamma,\mathcal{R})\}_{:,t+1-l_E}.
    \label{eq:fuv_LE_next}
\end{align}
For $t=l_E+1$, we have the following identity:
\begin{align}
    \{\hH^0(\sfD^0)^\top\}_{:,1}
    &=\frac{1}{\sfR}\sum_{\sfr=1}^\sfR\left\{\hH^0\Big(\E[\tg_0'(\bar{E}_\sfr^0,\bar{\gamma},\sfr)]\Big)^\top\right\}_{:,1}
    =\frac{1}{\sfR}\sum_{\sfr=1}^\sfR\hH^0\cdot
    \begin{bmatrix}
        \E[\partial_{1}\tg_{0,1}(\bar{E}_\sfr^0,\bar{\gamma},\sfr)] \\
        \vdots \\
        \E[\partial_{l_E}\tg_{0,1}(\bar{E}_\sfr^0,\bar{\gamma},\sfr)]
    \end{bmatrix} \nonumber \\
    &=\frac{1}{\sfR}\sum_{\sfr=1}^\sfR\sum_{s=1}^{l_E}\hH_{:,s}^0\E[\partial_s\tg_{0,1}(\bar{E}_\sfr^0,\bar{\gamma},\sfr)].
    \label{eq:abs_matrix_AMP_identity1}
\end{align}
Then, we have
\begin{align*}
    h^{l_E+1}
    &\stackrel{(a)}{=}\sqrt{\delta}\tX^\top u^{l_E+1}-\sum_{s=1}^{l_E}b_s^{l_E+1}v^s
    \stackrel{(b)}{=}\tX^\top\hR_{:,1}^0-\frac{1}{\sfR}\sum_{\sfr=1}^\sfR\sum_{s=1}^{l_E}\E[\partial_s\tg_{0,1}(\bar{E}_\sfr^0,\bar{\gamma},\sfr)]\hH_{:,s}^0 \\
    &\stackrel{(c)}{=}\tX^\top\hR_{:,1}^0-\{\hH^0(\sfD^0)^\top\}_{:,1}
    =H_{:,1}^1,
\end{align*}
where we use \eqref{eq:modified_U-AMP} in (a), substitute the definitions of $u^{l_E+1}$, $b_s^{l_E+1}$, and $v^s$ in (b), and apply \eqref{eq:abs_matrix_AMP_identity1} in (c). Next, we have
\begin{align*}
    v^{l_E+1}=0,
    \quad
    e^{l_E+1}=0,
    \quad
    u^{l_E+2}=\frac{1}{\sqrt{\delta}}\hR_{:,2}^0.
\end{align*}
Similarly, for $t=l_E+2,\dots,l_E+l_H-1$, we have
\begin{align*}
    h^t=H_{:,t-l_E}^1, 
    \quad
    v^t=0,
    \quad
    e^t=0,
    \quad
    u^{t+1}=\frac{1}{\sqrt{\delta}}\hR_{:,t+1-l_E}^0.
\end{align*}
For $t=l_E+l_H$, we set $f_{t}^v=0$ and $f_{t+1}^u=0$, so that
\begin{align*}
    h^{l_E+l_H}=H_{:,l_H}^1, \quad
    v^{l_E+l_H}=0, \quad
    e^{l_E+l_H}=0, \quad
    u^{l_E+l_H+1}=0.
\end{align*}
This completes the derivation for \eqref{eq:abs_matrix_AMP_reduction2}.

For $t=l_E+l_H+1,\dots,2l_E+l_H-1$, we set
\begin{align}
    & f_t^v(h^1,\dots,h^{l_E},\underbrace{h^{l_E+1},\dots,h^{l_E+l_H}}_{=H^1},h^{l_E+l_H+1},\dots,h^t,\underbrace{d^1}_{=\beta},d^2,\dots,d^{l_E+1}, \mathcal{C}) \label{eq:fv_lElH1}
    =\frac{1}{\sqrt{\delta}}\{\tf_1(H^1,\beta,\mathcal{C})\}_{:,1}, \\
    & f_{t+1}^u(e^1,\dots,e^t,c^1,\mathcal{R})
    =0. \nonumber
\end{align}
For $t=l_E+l_H+1$, we have
\begin{align*}
    h^{l_E+l_H+1}=0, \quad
    v^{l_E+l_H+1}=\frac{1}{\sqrt{\delta}}\hH_{:,1}^1.
\end{align*}
We pause to show an identity:
\begin{align}
    \big\{\hR^0(\sfB^1)^\top\big\}_{:,1}
    &=\frac{1}{\delta\sfC}\sum_{\sfc=1}^\sfC\left\{\hR^0\E[\tf_{1}'(\bH_{\sfc}^{1},\bar{\beta},\sfc)]^\top\right\}_{:,1}
    =\frac{1}{\delta\sfC}\sum_{\sfc=1}^\sfC\hR^0
    \begin{bmatrix}
        \E[\partial_1\tf_{1,1}(\bH_{\sfc}^{1},\bar{\beta},\sfc)] \\
        \vdots \\
        \E[\partial_L\tf_{1,1}(\bH_{\sfc}^{1},\bar{\beta},\sfc)]
    \end{bmatrix} \nonumber \\
    &=\frac{1}{\delta\sfC}\sum_{\sfc=1}^\sfC\sum_{s=1}^{l_H}\hR_{:,s}^0\E[\partial_s\tf_{1,1}(\bH_{\sfc}^{1},\bar{\beta},\sfc)].
    \label{eq:abs_matrix_AMP_identity2}
\end{align}
Then, from \eqref{eq:modified_U-AMP}, we have
\begin{align*}
    e^{l_E+l_H+1}
    &=\sqrt{\delta}\tX v^{l_E+l_H+1}-\sum_{s=1}^{l_E+l_H+1}a_s^{l_E+l_H+1}u^s
    \stackrel{(a)}{=}\tX\hH_{:,1}^1-\frac{1}{\delta \sfC}\sum_{\sfc=1}^\sfC\sum_{s=1}^{l_H}\hR_{:,s}^0\E\big[\partial_s\tf_{1,1}(\hH_{\sfc,:}^1,\bar{\beta},\sfc)\big] \\
    &\stackrel{(b)}{=}\tX\hH_{:,1}^1-\left\{\hR^0(B^1)^\top\right\}_{:,1}
    =E_{:,1}^1,
\end{align*}
where (a) uses the definitions of $v^{l_E+l_H+1}$, $a_s^{l_E+l_H+1}$ and $u^s$, and (b) uses \eqref{eq:abs_matrix_AMP_identity2}. Next, we have $u^{l_E+l_H+1}=0$. Similarly, for $t=l_E+l_H+2,\dots,2l_E+l_H-1$, we have
\begin{align*}
    h^t=0, \quad
    v^t=\frac{1}{\sqrt{\delta}}\hH_{:,t-l_E-l_H}^1, \quad
    e^t=E_{:,t-l_E-l_H}^1, \quad
    u^{t+1}=0.
\end{align*}
For $t=2l_E+l_H$, we set
\begin{align}
    & f_{2l_E+l_H}^v(h^1,\dots,h^t,d^1,\dots,d^{l_E+1}, \mathcal{C})
    =\frac{1}{\sqrt{\delta}}\{\tf_1(H^1,\beta,\mathcal{C})\}_{:,l_E} ,
     \label{eq:fv_2lElH} \\
   &  f_{2l_E+l_H+1}^u(e^1,\dots,e^{l_E+l_H},\underbrace{e^{l_E+l_H+1},\dots,e^{2l_E+l_H}}_{=E^1},\underbrace{c^1}_{=\gamma},\mathcal{R})
    =\frac{1}{\sqrt{\delta}}\big\{\tg_1(E^1,\gamma,\mathcal{R})\big\}_{:,1}.
\nonumber 
\end{align}
This gives
\begin{align*}
    h^{2l_E+l_H}=0, \quad
    v^{2l_E+l_H}=\frac{1}{\sqrt{\delta}}\hH_{:,l_E}^1, \quad
    e^{2l_E+l_H}=E_{:,l_E}^1, \quad
    u^{2l_E+l_H+1}=\frac{1}{\sqrt{\delta}}\hR_{:,1}^1,
\end{align*}
completing the derivation of \eqref{eq:abs_matrix_AMP_reduction3}. This concludes the reduction of the abstract matrix-AMP iterates to the U-AMP iterates for $t=0$.

We now show the convergence statements in Theorem \ref{thm:abs_matrix_AMP} for $E^0, H^1$, and $E^1$ by reducing the abstract matrix-AMP SE parameters to the corresponding U-AMP SE parameters.

\paragraph{Convergence of $(E^0, \gamma)$.}  
From Assumption (C1) and Corollary \ref{cor:U-AMP}, we have  $(d_\sfc^1,\dots,d_\sfc^{L_d}, \,  h_\sfc^1,\dots,$ $h_\sfc^{l_E})\stackrel{W}{\rightarrow}(\bar{d}_\sfc^1,\dots,\bar{d}_\sfc^{L_d}, \bh^1,\dots,\bh^{l_E})$, for $\sfc \in [\sfC]$. Recalling that 
$$v_\sfc^s=f_s^v(h^1,\dots,h^s,d_\sfc^1,\dots,d_\sfc^{L_d}, \sfc), \quad \bv_\sfc^s=f_s^v(\bh^1,\dots,\bh^s,\bar{d}_\sfc^1,\dots,\bar{d}_\sfc^{L_d}, \sfc), \quad  \text{ for } s \ge 1,$$  the convergence above implies that $(v^1_\sfc,\dots,v^{l_E}_\sfc) \stackrel{W}{\rightarrow} (\bv^1_\sfc,\dots,\bv^{l_E}_\sfc)$. Since we have shown in \eqref{eq:abs_matrix_AMP_reduction1} that $(v^1,\dots,v^{l_E})=\frac{1}{\sqrt{\delta}}\hH^0$,  we must have $(\bv_\sfc^1,\dots,\bv_\sfc^{l_E})=\frac{1}{\sqrt{\delta}}\bH_{\sfc}^0$, where the latter is given by Assumption (B2). 

For $\sfr \in [\sfR]$, Corollary \ref{cor:U-AMP} implies that $(e_{\sfr}^1,\dots,e_{\sfr}^{l_E})\stackrel{W_2}{\rightarrow}(\be^1,\dots,\be^{l_E})\sim\normal(0,\Gamma^{l_E})$, where
\begin{equation}
    \Gamma^{l_E}
=\frac{1}{\sfC}\sum_{\sfc=1}^\sfC\widehat{\Gamma}^{l_E,\sfc},
\quad
\widehat{\Gamma}^{l_E,\sfc}
=\big(\E[\bv_\sfc^r\bv_\sfc^s]\big)_{r,s=1}^{l_E}.
\end{equation}
Using $(\bv_\sfc^1,\dots,\bv_\sfc^{l_E})=\frac{1}{\sqrt{\delta}}\bH_{\sfc}^0$ in  \eqref{eq:Pi0c_def}, we have that $\widehat{\Pi}^{0,\sfc} = \widehat{\Gamma}^{l_E,\sfc}$ and $\Pi^{0}=\Gamma^{l_E}$. Since $(e_\sfr^1,\dots,e_\sfr^{l_E})=E_{\mathcal{I}_\sfr,:}^0$ from \eqref{eq:abs_matrix_AMP_reduction1}, and Corollary \ref{cor:U-AMP} guarantees that $(e_\sfr^1,\dots,e_\sfr^{l_E})\stackrel{W_2}{\rightarrow}(\be^1,\dots,\be^{l_E}) \sim\normal(0,\Gamma^{l_E})$, we have $E_{\mathcal{I}_\sfr,:}^0
\stackrel{W_2}{\rightarrow}\bE_{\sfr}^0\sim\normal(0,\Pi^0)$. Since $c^1=\gamma$, and $\gamma$ is independent of $E^0$, we can  use Assumption (B2) and apply Corollary \ref{cor:U-AMP} to $ (E_{\mathcal{I}_\sfr,:}^0,\gamma_\sfr)$ to obtain
\begin{align}
    (E_{\mathcal{I}_\sfr,:}^0,\gamma_\sfr) \stackrel{W_2}{\rightarrow}(\bE_{\sfr}^0,\bar{\gamma}).
    \label{eq:E^0_convergence}
\end{align}

\paragraph{Convergence of $(H^1, \beta)$.}
For $\sfr \in [\sfR]$, recall that  $u^{s+1}_\sfr=f_{s+1}^u(e_\sfr^1,\dots,e_\sfr^s,c_\sfr^1,\dots,c_\sfr^{L_c},\sfr)$ and 
$\bu^{s+1}_\sfr=f_{s+1}^u(\be^1,\dots,\be^s,\bc_\sfr^1,\dots,\bc_\sfr^{L_c},\sfr)$.
Corollary \ref{cor:U-AMP} implies that $(u_\sfr^{1},\dots,u_\sfr ^{l_E+l_H})
\stackrel{W_2}{\rightarrow} (\bu_\sfr^{1},\dots,\bu_\sfr^{l_E+l_H})$.
Moreover, we have shown in \eqref{eq:abs_matrix_AMP_reduction2} that $(u^{l_E+1},\dots,u^{l_E+l_H})=\frac{1}{\sqrt{\delta}}\hR^0=\frac{1}{\sqrt{\delta}}\tg_0(E^0,\gamma,\mathcal{R})$. Hence, using \eqref{eq:E^0_convergence} and noting that $\tg_0$ satisfies the polynomial growth condition in \eqref{eq:poly_growth_cond}, we  have $(\bu_\sfr^{l_E+1},\dots,\bu_\sfr^{l_E+l_H})=\frac{1}{\sqrt{\delta}}\tg_0(\bE_{\sfr}^0,\bar{\gamma},\sfr)$, for $\sfr \in [\sfR]$. 

Corollary \ref{cor:U-AMP} states that $(h_{\sfc}^1,\dots,h_{\sfc}^{2l_E+l_H})\stackrel{W_2}{\rightarrow}(\bh^1,\dots,\bh^{2l_E+l_H})\sim\normal(0,\Xi^{2l_E+l_H})$, where
\begin{align}
\Xi^{2l_E+l_H}
=\frac{1}{\sfR}\sum_{\sfr=1}^\sfR\widehat{\Xi}^{2l_E+l_H,\sfr},
\quad
\widehat{\Xi}^{2l_E+l_H}
=\delta\big(\E[\bu_\sfr^r\bu_\sfr^s]\big)_{r,s=1}^{2l_E+l_H}.
\label{eq:Xi_2l_E+l_H}
\end{align}
Then recalling the definition of  $\Omega^1$ from \eqref{eq:abs_matrix_AMP_SE}, and the functions $f^u_{l_E+1},\ldots, f^u_{l_E+l_H}$ from \eqref{eq:fuv_LE}-\eqref{eq:fuv_LE_next}, we have
\begin{align}
    \widehat{\Xi}_{[l_E+1:l_E+l_H],[l_E+1:l_E+l_H]}^{2l_E+l_H,\sfr}
    =\widehat{\Omega}^{1,\sfr}
    \implies
    \Xi_{[l_E+1:l_E+l_H],[l_E+1:l_E+l_H]}^{2l_E+l_H}
    =\Omega^1.
    \label{eq:Xi_eq_Omega}
\end{align}
We have shown that $(h^{l_E+1},\dots,h^{l_E+l_H})=H^1$ (see \eqref{eq:abs_matrix_AMP_reduction2}), and Corollary \ref{cor:U-AMP} states that
$$
(d_\sfc^1, h_\sfc^{l_E+1},\dots,h_\sfc^{l_E+l_H})\stackrel{W_2}{\rightarrow}(\bd^1_\sfc, \bh^{l_E+1},\dots,\bh^{l_E+l_H}),
$$
where $\bd_\sfc^1$ 
and $(\bh^{l_E+1},\dots,\bh^{l_E+l_H}) \sim \normal\big(0,\Xi^{l_E+l_H}_{[l_E+1:l_E+l_H],[l_E+1:l_E+l_H]}\big)$ are independent.
Thus, recalling that $d_\sfc^1= \beta_\sfc$, by the equivalence of  the covariance matrices in \eqref{eq:Xi_eq_Omega}, we have
\begin{align}
(H_{\mathcal{J}_\sfc,:}^1,\beta_\sfc)\stackrel{W_2}{\rightarrow}(\bH_{\sfc}^1,\bar{\beta}), \quad \text{where }   \bar{\beta} \text{ is independent of } \bH_{\sfc}^1 \sim  \normal(0,\Omega^1).
\label{eq:H^1_convergence}
\end{align}

\paragraph{Convergence of $(E^1, \gamma)$.} Recall that $v_\sfc^s=f^v_s(h_\sfc^1,\dots,h_\sfc^t,d_\sfc^1,\dots,d_\sfc^{L_d},\sfc)$ and $\bv_\sfc^s = f^v_s(\bh^1,\dots,\bh^t,\bd_\sfc^1,$ $\dots,\bd_\sfc^{L_d},\sfc)$, for $\sfc \in [\sfC]$. Corollary \ref{cor:U-AMP} implies that $(v_\sfc^1, \ldots, v_\sfc^{2l_E+l_H}) \stackrel{W_2}{\rightarrow} (\bv_\sfc^1, \ldots, \bv_\sfc^{2l_E+l_H})$. Moreover, we have shown in \eqref{eq:abs_matrix_AMP_reduction3} that $(v^{l_E+l_H+1},\dots,v^{2l_E+l_H})=\frac{1}{\sqrt{\delta}}\hH^1=\frac{1}{\sqrt{\delta}}\tf_{1}(H^1,\beta,\mathcal{C})$. Thus, using \eqref{eq:H^1_convergence} and noting that $\tf_1$ satisfying the polynomial growth condition in \eqref{eq:poly_growth_cond}, we have $(\bv_\sfc^{l_E+l_H+1},\dots,\bv_\sfc^{2l_E+l_H})$ $=\frac{1}{\sqrt{\delta}}\tf_1(\bar{H}_\sfc^1,\bar{\beta},\sfc)$, for $\sfc \in [\sfC]$.

Corollary \ref{cor:U-AMP} states that $(e_{\sfr}^1,\dots,e_{\sfr}^{2l_E+l_H})\stackrel{W_2}{\rightarrow}(\be^1,\dots,\be^{2l_E+l_H})\sim\normal(0,\Gamma^{2l_E+l_H})$, where
\begin{align}
    \Gamma^{2l_E+l_H}
    =\frac{1}{\sfC}\sum_{\sfc=1}^\sfC\widehat{\Gamma}^{2l_E+l_H,\sfc},
    \quad
    \widehat{\Gamma}^{2l_E+l_H,\sfc}
    =\big(\E[\bv_\sfc^r\bv_\sfc^s]\big)_{r,s=1}^{2l_E+l_H}.
    \label{eq:Gamma_2l_E+l_H}
\end{align}
 Then, recalling the definition of $\Pi^1$ from \eqref{eq:abs_matrix_AMP_SE}, and the functions $f^v_{l_E+l_H+1}, \ldots, f^v_{2l_E+l_H}$ from \eqref{eq:fv_lElH1} and \eqref{eq:fv_2lElH}, we have
\begin{align}
    \widehat{\Gamma}_{[l_E+l_H+1:2l_E+l_H],[l_E+l_H+1:2l_E+l_H]}^{2l_E+l_H,\sfc}
    =\widehat{\Pi}^{1,\sfc}
    \implies
    \Gamma_{[l_E+l_H+1:2l_E+l_H],[l_E+l_H+1:2l_E+l_H]}^{2l_E+l_H}
    =\Pi^1.
    \label{eq:Gamma_eqv_Pi}
\end{align}
From \eqref{eq:abs_matrix_AMP_reduction3} $(e^{l_E+l_H+1},\dots,e^{2l_E+l_H})=E^1$, while Corollary \ref{cor:U-AMP} states that
$$
(c^1_\sfr, e_\sfr^{l_E+l_H+1},\dots,e_\sfr^{2l_E+l_H})
\stackrel{W_2}{\rightarrow}
(\bc^1_\sfr, \be^{l_E+l_H+1},\dots,\be^{2l_E+l_H}),
$$
where $\bc^1_\sfr$ and $(\be^{l_E+l_H+1},\dots,\be^{2l_E+l_H}) \sim
\normal\big(0,\Gamma_{[l_E+l_H+1:2l_E+l_H],[l_E+l_H+1:2l_E+l_H]}^{2l_E+l_H}\big)$ are independent. 
Therefore, using $c^1_\sfr = \gamma_\sfr$ and \eqref{eq:Gamma_eqv_Pi}, we have $(E_{\mathcal{I}_\sfr,:}^1,\gamma_\sfr)\stackrel{W_2}{\rightarrow}(\bE_{\sfr}^1,\bar{\gamma})$, where $\bE_{\sfr}^1 \sim \normal(0, \Pi^1)$ is independent of $\bar{\gamma}$.

\paragraph{Inductive hypothesis.} For $k \ge 1$, assume that we can reduce $H^k$, $\hR^{k-1}$, $\hH^k$, and $E^k$ to U-AMP in iterations $t=l_Ek+l_H(k-1)+1,\dots,l_E(k+1)+l_Hk$. Mathematically, this means the following:
\begin{itemize}
    \item For $t=l_Ek+l_H(k-1)+1,\dots,l_Ek+kl_H$: we have
    \begin{align}
    \begin{split}
        (h^{l_Ek+l_H(k-1)+1},\dots,h^{l_Ek+kl_H})=H^k, \quad
        v^{l_Ek+l_H(k-1)+1},\dots,v^{l_Ek+kl_H}=0, \\
        e^{l_Ek+l_H(k-1)+1},\dots,e^{l_Ek+kl_H}=0, \quad
        (u^{l_Ek+l_H(k-1)+1},\dots,u^{l_Ek+kl_H})=\frac{1}{\sqrt{\delta}}\hR^{k-1}.
    \end{split}
    \label{eq:abs_matrix_AMP_reduction4}
    \end{align}
    \item For $t=l_Ek+l_Hk+1,\dots,l_E(k+1)+l_Hk$: we have
    \begin{align}
    \begin{split}
        h^{l_Ek+l_Hk+1},\dots,h^{l_E(k+1)+l_Hk}=0, \quad
        (v^{l_Ek+l_Hk+1},\dots,v^{l_E(k+1)+l_Hk})=\frac{1}{\sqrt{\delta}}\hH^k, \\
        (e^{l_Ek+l_Hk+1},\dots,e^{l_E(k+1)+l_Hk})=E^k, \quad
        u^{l_Ek+l_Hk+1},\dots,u^{l_E(k+1)+l_Hk}=0.
    \end{split}
    \label{eq:abs_matrix_AMP_reduction5}
    \end{align}
\end{itemize}
Defining the index sets
$$
\mathcal{I}_k=[l_Ek+l_Hk+1:l_E(k+1)+l_Hk],
\quad
\mathcal{J}_k=[l_Ek+l_H(k-1)+1:l_Ek+l_Hk],
$$
we also assume that the following convergence statements hold, for $1 \le s \le k$:
\begin{align}
\begin{split}
      & (H_{\mathcal{J}_\sfc,:}^s,\beta_\sfc)\stackrel{W_2}{\rightarrow}(\bH_{\sfc}^s,\bar{\beta}), \quad  \bar{\beta} \text{ independent of } \bH_{\sfc}^s\sim\normal(0,\Omega^s), \\
  & (E_{\mathcal{I}_\sfr,:}^s,\gamma_\sfr)\stackrel{W_2}{\rightarrow}(\bE_{\sfr}^s,\bar{\gamma}),  \quad \bar{\gamma} \text{ independent of } \bE_{\sfr}^k\sim\normal(0,\Pi^s).
\end{split}
\end{align}

\paragraph{Inductive step.} We  need to show that we can reduce $H^{k+1}$, $\hR^{k}$, $\hH^{k+1}$, and $E^{k+1}$ to U-AMP in iterations $t=l_E(k+1)+l_Hk+1,\dots,l_E(k+2)+l_H(k+1)$, and that the corresponding convergence statements  hold. The choices for the functions $(f^v_t, f^u_{t+1})$ are analogous to those in  \eqref{eq:fuv_LE_next}, \eqref{eq:fv_lElH1},  and the steps for the reduction are very similar to the base case for $t\in\{l_E+1,\dots,2L_E+l_H\}$, and are omitted for brevity. We provide the summary of the reductions below:
\begin{itemize}
    \item For $t=l_E(k+1)+l_Hk+1,\dots,l_E(k+1)+l_H(k+1)$: we have
    \begin{align}
    \begin{split}
        \big(h^{l_E(k+1)+l_Hk+1},\dots,h^{l_E(k+1)+l_H(k+1)}\big)=H^{k+1}, \quad
        v^{l_E(k+1)+l_Hk+1},\dots,v^{l_E(k+1)+l_H(k+1)}=0, \\
        e^{l_E(k+1)+l_Hk+1},\dots,e^{l_E(k+1)+l_H(k+1)}=0, \quad
        \big(u^{l_E(k+1)+l_Hk+1},\dots,u^{l_E(k+1)+l_H(k+1)}\big)=\frac{1}{\sqrt{\delta}}\hR^{k}.
    \end{split}
    \label{eq:abs_matrix_AMP_reduction6}
    \end{align}
    \item For $t=l_E(k+1)+l_H(k+1)+1,\dots,l_E(k+2)+l_H(k+1)$: we have
    \begin{align}
    \begin{split}
        h^{l_E(k+1)+l_H(k+1)+1},\dots,h^{l_E(k+2)+l_H(k+1)}=0,\\
        \big(v^{l_E(k+1)+l_H(k+1)+1},\dots,v^{l_E(k+2)+l_H(k+1)}\big)=\frac{1}{\sqrt{\delta}}\hH^{k+1}, \\
        \big(e^{l_E(k+1)+l_H(k+1)+1},\dots,e^{l_E(k+2)+l_H(k+1)}\big)=E^{k+1},\\
        u^{l_E(k+1)+l_H(k+1)+1},\dots,u^{l_E(k+2)+l_H(k+1)}=0.
    \end{split}
    \label{eq:abs_matrix_AMP_reduction7}
    \end{align}
\end{itemize}

We now show the convergence statements in Theorem \ref{thm:abs_matrix_AMP} for $H^{k+1}$ and $E^{k+1}$ by reducing the abstract matrix-AMP SE parameters to the corresponding U-AMP SE parameters. Define the index sets
\begin{align*}
\mathcal{I}_{k+1}=[l_E(k+1)+l_H(k+1)+1:l_E(k+2)+l_H(k+1)],\\
\mathcal{J}_{k+1}=[l_E(k+1)+l_Hk+1:l_E(k+1)+l_H(k+1)].
\end{align*}

\paragraph{Convergence of $(H^{k+1}, \beta)$. }
From \eqref{eq:abs_matrix_AMP_reduction6} we have $(u_\sfr^l)_{l\in\mathcal{J}_{k+1}}=\frac{1}{\sqrt{\delta}}\hR_{\mathcal{I}_\sfr,:}^k=\frac{1}{\sqrt{\delta}}\tg_k(E_{\mathcal{I}_\sfr,:}^k,\gamma_\sfr,\sfr)$, and 
by the inductive hypothesis,  $(E_{\mathcal{I}_\sfr,:}^k,\gamma_\sfr)\stackrel{W_2}{\rightarrow}(\bE_{\sfr}^k,\bar{\gamma})$. Since $\tg_{k}$ satisfies the polynomial growth condition in \eqref{eq:poly_growth_cond},  Corollary \ref{cor:U-AMP} implies that 
$(u_\sfr^l)_{l\in\mathcal{J}_{k+1}} \stackrel{W_2}{\rightarrow} (\bu_\sfr^l)_{l\in\mathcal{J}_{k+1}} = \frac{1}{\sqrt{\delta}}\tg_{k}(\bE_{\sfr}^{k},\bar{\gamma},\sfr)$.

Corollary \ref{cor:U-AMP} states that $(h_{\sfc}^l)_{l\in\mathcal{J}_{k+1}}\stackrel{W_2}{\rightarrow}(\bh^l)_{l\in\mathcal{J}_{k+1}}\sim\normal(0,\Xi_{\mathcal{J}_{k+1},\mathcal{J}_{k+1}}^{l_E(k+2)+l_H(k+1)})$, where
\begin{align*}
\Xi_{\mathcal{J}_{k+1},\mathcal{J}_{k+1}}^{l_E(k+2)+l_H(k+1)}
=\frac{1}{\sfR}\sum_{\sfr=1}^\sfR\widehat{\Xi}_{\mathcal{J}_{k+1},\mathcal{J}_{k+1}}^{l_E(k+2)+l_H(k+1),\sfr},
\quad
\widehat{\Xi}^{l_E(k+2)+l_H(k+1),\sfr}
=\delta\big(\E[\bu_\sfr^r\bu_\sfr^s]\big)_{r,s\in\mathcal{J}_{k+1}}.
\end{align*}
Then, comparing the above with the definition of $\Omega^{k+1}$in \eqref{eq:abs_matrix_AMP_SE}, we get
\begin{align}
    \widehat{\Xi}_{\mathcal{J}_{k+1},\mathcal{J}_{k+1}}^{l_E(k+2)+l_H(k+1),\sfr}
    =\widehat{\Omega}^{k+1,\sfr}
    \implies
    \Xi_{\mathcal{J}_{k+1},\mathcal{J}_{k+1}}^{l_E(k+2)+l_H(k+1)}
    =\Omega^{k+1}.
    \label{eq:Xi_k+1_eq_Omega_k+1}
\end{align}
From \eqref{eq:abs_matrix_AMP_reduction6}, we have  $(h^{l})_{l\in\mathcal{J}_{k+1}}=H^{k+1}$, and Corollary \ref{cor:U-AMP} states that
$
\left(d^1_\sfc,  (h_\sfc^{l})_{l\in\mathcal{J}_{k+1}} \right) 
\stackrel{W_2}{\rightarrow}
\left( \bd^1_\sfc, \right.$ $\left.(\bh^{l})_{l\in\mathcal{J}_{k+1}} \right)
%\sim \normal(0,\Omega^{k+1}),
$, where $\bd^1_\sfc$ and $(\bh^{l})_{l\in\mathcal{J}_{k+1}}$ are independent.
Thus, by the equivalence of covariance matrices in \eqref{eq:Xi_k+1_eq_Omega_k+1} and recalling $d^1_\sfc= \beta_\sfc$, we have $(H_{\mathcal{J}_\sfc,:}^{k+1},\beta_\sfc)\stackrel{W_2}{\rightarrow}(\bH_{\sfc}^{k+1},\bar{\beta})$, where $\bar{\beta}$ is independent of $\bH_{\sfc}^{k+1} \sim \normal(0, \Omega^{k+1})$.
% $$
% H_{\mathcal{J}_\sfc,:}^{k+1}
% \stackrel{W_2}{\rightarrow}\bH_{\sfc}^{k+1}\sim\normal(0,\Omega^{k+1}).
% $$
% Since we have $d^1=\beta$, by applying Corollary \ref{cor:U-AMP} to this collection of vectors we obtain $(H_{\mathcal{J}_\sfc,:}^{k+1},\beta_\sfc)\stackrel{W_2}{\rightarrow}(\bH_{\sfc}^{k+1},\bar{\beta})$.

\paragraph{Convergence of $(E^{k+1}, \gamma)$. } From \eqref{eq:abs_matrix_AMP_reduction7}, we have  $(v_\sfc^l)_{l\in\mathcal{I}_{k+1}}=\frac{1}{\sqrt{\delta}}\hH_{\mathcal{J}_\sfc,:}^{k+1}=\frac{1}{\sqrt{\delta}}\tf_{k+1}(H_{\mathcal{J}_\sfc,:}^{k+1},\beta_\sfc,\sfc)$, and we have shown that $(H_{\mathcal{J}_\sfc,:}^{k+1},\beta_\sfc)\stackrel{W_2}{\rightarrow}(\bar{H}_\sfc^{k+1},\bar{\beta})$.  Since $\tf_{k+1}$ satisfies the polynomial growth condition, Corollary \ref{cor:U-AMP} implies that $(v_\sfc^l)_{l\in\mathcal{I}_{k+1}} \stackrel{W_2}{\rightarrow} (\bv_\sfc^l)_{l\in\mathcal{I}_{k+1}} = 
\frac{1}{\sqrt{\delta}}\tf_{k+1}(\bar{H}^{k+1}, \bar{\beta},\sfc)$.

Corollary \ref{cor:U-AMP} states that $(e_{\sfr}^l)_{l\in\mathcal{I}_{k+1}}\stackrel{W_2}{\rightarrow}(\be^l)_{l\in\mathcal{I}_{k+1}}\sim\normal(0,\Gamma_{\mathcal{I}_{k+1},\mathcal{I}_{k+1}}^{l_E(k+2)+l_H(k+1)})$, where
$$
\Gamma_{\mathcal{I}_{k+1},\mathcal{I}_{k+1}}^{l_E(k+2)+l_H(k+1)}
=\frac{1}{\sfC}\sum_{\sfc=1}^\sfC\widehat{\Gamma}_{\mathcal{I}_{k+1},\mathcal{I}_{k+1}}^{l_E(k+2)+l_H(k+1),\sfc},
\quad
\widehat{\Gamma}_{\mathcal{I}_{k+1},\mathcal{I}_{k+1}}^{l_E(k+2)+l_H(k+1),\sfc}
=\big(\E[\bv_\sfc^r\bv_\sfc^s]\big)_{r,s\in\mathcal{I}_{k+1}}.
$$
 Then, comparing the above  with the definition of $\Pi^{k+1}$ in \eqref{eq:abs_matrix_AMP_SE}, we get
\begin{align}
    \widehat{\Gamma}_{\mathcal{I}_{k+1},\mathcal{I}_{k+1}}^{l_E(k+2)+l_H(k+1),\sfc}
    =\widehat{\Pi}^{k+1,\sfc}
    \implies
    \Gamma_{\mathcal{I}_{k+1},\mathcal{I}_{k+1}}^{l_E(k+2)+l_H(k+1)}
    =\Pi^{k+1}
    \label{eq:Gamma_k+1_eqv_Pi_k+1}
\end{align}
From \eqref{eq:abs_matrix_AMP_reduction7}, we have $(e^l)_{l\in\mathcal{I}_{k+1}}=E^{k+1}$  and Corollary \ref{cor:U-AMP} states that $\left(c^1, (e_\sfr^l)_{l\in\mathcal{I}_{k+1}} \right) \stackrel{W_2}{\rightarrow}
\left(\bc^1, \right.$ $\left.(\be^l)_{l\in\mathcal{I}_{k+1}} \right)$, where $\bc^1$ and  $(\be^l)_{l\in\mathcal{I}_{k+1}}$ are independent.
Thus, by the equivalence of covariance matrices in \eqref{eq:Gamma_k+1_eqv_Pi_k+1} and recalling $c^1_\sfr= \gamma_\sfr$, we have $(E_{\mathcal{I}_\sfr,:}^{k+1},\gamma_\sfr)\stackrel{W_2}{\rightarrow}(\bE_{\sfr}^{k+1},\bar{\gamma})$, where $\bar{\gamma}$ is independent of $\bE_{\sfr}^{k+1} \sim \normal(0, \Pi^{k+1})$.

This completes the proof of the inductive step, and hence, of Theorem \ref{thm:abs_matrix_AMP}.

\subsection{Proof of Corollary \ref{cor:U-AMP}} \label{sec:mod_of_U-AMP}

The abstract AMP recursion in \eqref{eq:modified_U-AMP}, without the block-wise dependence of the functions $f^v_t$ and $f^u_{t+1}$, was analyzed in  \cite{Wan22}. We show how the block-wise dependence can be included without loss of generality, and thereby prove Corollary \ref{cor:U-AMP} by referring to the state evolution result of \cite{Wan22}.

The abstract AMP iteration  for a generalized white noise matrix $\tX\in\mb{R}^{n\times p}$ analyzed in \cite{Wan22} is as follows.  Given an initializer $u^1\in\mb{R}^n$, side information $c^1,\dots,c^{L_c^*}\in\mb{R}^n$ and $d^1,\dots,d^{L_d^*}\in\mb{R}^p$, all independent of $\tX$, the iterates of the abstract AMP recursion are computed as:
\begin{align}
\begin{split}
    h^t&=\sqrt{\delta}\tX^\top u^t-\sum_{s=1}^{t-1}b_{s}^tv^s, \qquad 
    v^t=f_t^v(h^1,\dots,h^t,d^1,\dots,d^{L_d^*}), \\
    e^t&=\sqrt{\delta}\tX v^t-\sum_{s=1}^ta_{s}^tu^s, \qquad
    u^{t+1}=f_{t+1}^u(e^1,\dots,e^t,c^1,\dots,c^{L_c^*}),
\end{split} \label{eq:U-AMP-new}
\end{align}
where the functions $f_t^v:\mb{R}^{t+L_d^*}\rightarrow\mb{R}$, $f_{t+1}^u:\mb{R}^{t+L_c^*}\rightarrow\mb{R}$ act row-wise. The memory coefficients  $\{b_{s}^t\}_{s<t}$ and $\{a_{s}^t\}_{s\leq t}$ are defined below in \eqref{eq:abstract_AMP_onsager}. We have the following assumptions:

    \textbf{(D1)} When $n,p\rightarrow\infty$, we have $n/p=\delta>0$, for fixed $L_c^*$ and $L_d^*$. Furthermore, we have
    \begin{align*}
        (u^1,c^1,\dots,c^{L_c^*})
        \stackrel{W}{\rightarrow}
        (\bu^1,\bc^1,\dots,\bc^{L_c^*})
        \text{ and }
        (d^1,\dots,d^{L_d^*})
        \stackrel{W}{\rightarrow}
        (\bd^1,\dots,\bd^{L_d^*}),
    \end{align*}
    for joint limit laws $(\bu^1,\bc^1,\dots,\bc^{L_c^*})$ and $(\bd^1,\dots,\bd^{L_d^*})$ having finite moments of all orders, where $\E[(\bu^1)^2]\geq 0$. Multivariate polynomials are dense in the real $L^2$-spaces of functions $f:\mb{R}^{L_c^*+1}\rightarrow\mb{R}$ and $g:\mb{R}^{L_d^*}\rightarrow\mb{R}$ with the inner products
    \begin{align*}
        \langle f,\tf\rangle
        :=\E[f(\bu^1,\bc^1,\dots,\bc^{L_c^*})\tf(\bu^1,\bc^1,\dots,\bc^{L_c^*})]
        \text{ and }
        \langle g,\tg\rangle
        :=\E[g(\bd^1,\dots,\bd^{L_d^*})\tg(\bd^1,\dots,\bd^{L_d^*})].
    \end{align*}
    
    \textbf{(D2)}, \textbf{(D3)}, \textbf{(D4)} These are identical to \textbf{(C2)}, \textbf{(C3)}, \textbf{(C4)}, with $L_d^*$ replacing $L_d$.

The state evolution covariance matrices $\Xi^t, \Gamma^t \in \reals^{t \times t}$ are iteratively defined as follows, starting from $\Xi^1=\delta\E[(\bu^1)^2]\in\mb{R}^{1\times 1}$. Given $\Xi^t$, for $t \ge 1$, let $(\bh^1,\dots,\bh^t)\sim\normal(0,\Xi^t)$ independent of $(\bd^1,\dots,\bd^{L_d^*})$ and define
\begin{align*}
    \bv^s=f_s^v(\bh^1,\dots,\bh^s,\bd^1,\dots,\bd^{L_d^*}), \quad  s\in[t].
\end{align*}
Then, $\Gamma^t=(\E[\bv^r\bv^s])_{r,s=1}^t\in\mb{R}^{t\times t}$. Next, let 
$(\be^1,\dots,\be^t)\sim\normal(0,\Gamma^t)$ independent of $(\bu^1,\bc^1,\dots,$ $\bc^{L_c^*})$ and define
\begin{align*}
    \bu^{s+1}=f_{s+1}^u(\be^1,\dots,\be^s,\bc^1,\dots,\bc^{L_c^*}), \quad s \in [t].
\end{align*}
Then, $\Xi^{t+1}=(\delta\cdot\E[\bu^r\bu^s])_{r,s=1}^{t+1}\in\mb{R}^{(t+1)\times(t+1)}$.
The memory coefficients in \eqref{eq:U-AMP-new} are then defined as
\begin{align}
    a_{s}^t=\E\big[\partial_sf_t^v(\bh^1,\dots,\bh^t,\bd^1,\dots,\bd^{L_d^*})\big]
    \text{ and }
    b_{s}^t=\delta\cdot\E\big[\partial_sf_t^u(\be^1.\dots,\be^{t-1},\bc^1,\dots,\bc^{L_c^*})\big],
    \label{eq:abstract_AMP_onsager}
\end{align}
where $\partial_s$ denotes partial derivative in the $s$th argument. The following theorem gives the state evolution result for the abstract AMP recursion.
\begin{theorem}\textup{\cite[Theorem 2.21]{Wan22}}\label{thm:abs_AMP} 
Let $\tX\in\mb{R}^{n\times p}$ be a generalized white noise matrix (as defined in Definition \ref{def:gen_white_noise_matrix}) with variance profile $S\in\mb{R}^{n\times p}$, and let $u^1,c^1,\dots,c^{L_c^*},d^1,\dots,d^{L_d^*}$ be independent of $\tX$ and satisfy Assumptions (D1)--(D4). Further assume that each matrix $\Xi^t$ and $\Gamma^t$ is non-singular. Then for any fixed $t\geq1$, almost surely as $n,p\rightarrow\infty$ with $n/p=\delta\in(0,\infty)$, the iterates of the abstract AMP in \eqref{eq:U-AMP-new} satisfy
\begin{align*}
    & (u^1,c^1,\dots,c^{L_c^*},e^1,\dots,e^t)
    \stackrel{W_2}{\rightarrow}(\bu^1,\bc^1,\dots,\bc^{L_c^*},\be^1,\dots,\be^t), \\
    & (d^1,\dots,d^{L_d^*},h^1,\dots,h^t)
    \stackrel{W_2}{\rightarrow}
    (\bd^1,\dots,\bd^{L_d^*},\bh^1,\dots,\bh^t),
\end{align*}
where $(\bh^1,\dots,\bh^t)\sim\normal(0,\Xi^t)$ and $(\be^1,\dots,\be^t)\sim\normal(0,\Gamma^t)$ are independent of $(\bu^1,\bc^1.\dots,\bc^{L_c^*})$ and $(\bd^1,\dots,\bd^{L_d^*})$.
\end{theorem}

To obtain the U-AMP recursion \eqref{eq:modified_U-AMP} from the abstract AMP recursion in \eqref{eq:U-AMP-new}, we choose $L_c^*=L_c+1$ and $L_d^*=L_d+1$, and the side information vectors $d^{L_d+1}\in\mb{R}^p$ and $c^{L_c+1}\in\mb{R}^n$ are set as
\begin{align}
    d^{L_d+1}=\mathcal{C},
    \quad
    c^{L_c+1}=\mathcal{R}.
    \label{eq:dL_cL_choice}
\end{align}
The functions $f^v_t$ and $f^u_{t+1}$, as well as the initializer $u^1$, are the same as those in \eqref{eq:U-AMP-new}.

With this choice, the empirical distribution of $d^{L_d+1}$ converges to $\bd^{L_d+1}\sim\text{Uniform}([\sfC])$, and the empirical distribution of $c^{L_c+1}$ converges to $\bar{c}^{L_c+1}\sim\text{Uniform}([\sfR])$.
Moreover,  Assumption (D1) is equivalent to Assumption (C1) of Corollary \ref{cor:U-AMP}. To see this, for $\sfr \in [\sfR]$, let $(\bu_\sfr^1,\bc_\sfr^1,\dots,\bc_\sfr^{L_c})$ be random variables  whose joint law equals the \emph{conditional} law of $(\bu^1,\bc^1,\dots,\bc^{L_c})$ given $\bc^{L_c+1}=\sfr$. Similarly, for $\sfc \in [\sfC]$, let $(\bd_\sfc^1,\dots, \bd_\sfc^{L_d})$ be jointly distributed according to the conditional law of $(\bd^1,\dots, \bd^{L_d})$ given
$\bd^{L_d+1}=\sfc$.

The memory coefficients in \eqref{eq:abstract_AMP_onsager} can then be expressed as:
\begin{align*}
    a_{s}^t =  \E\big[\partial_sf_t^v(\bh^1,\dots,\bh^t,\bd^1,\dots,\bd^{L_d+1})\big]
    &=\E\left[\E\big[\partial_sf_t^v(\bh^1,\dots,\bh^t,\bd^1,\dots,\bd^{L_d+1})\mid \bd^{L_d+1}\big]\right] \\
    &=\frac{1}{\sfC}\sum_{\sfc=1}^\sfC\E\big[\partial_sf_t^v(\bh^1,\dots,\bh^t,\bd^1,\dots, \bd^{L_d}, \sfc) \mid \bd^{L_d+1}=\sfc\big] \\
    &= \frac{1}{\sfC}\sum_{\sfc=1}^\sfC\E\big[\partial_sf_t^v(\bh^1,\dots,\bh^t,\bd_\sfc^1,\dots, \bd_\sfc^{L_d}, \sfc) \big],
\end{align*}
where for the last equality we used the fact that $(\bh^1,\dots,\bh^t)$ is independent of $(\bd^1,\dots, \bd^{L_d+1})$.
Similarly, we have
\begin{align*}
  b_{s}^t = \delta    \E\big[\partial_sf_t^u(\be^1,\dots,\be^{t-1},\bc^1,\dots,\bc^{L_c+1})\big]
    &= \delta \E\left[\E\big[\partial_sf_t^u(\be^1,\dots,\be^{t-1},\bc^1,\dots,\bc^{L_c+1})\,|\,\bc^{L_c+1}\big]\right] \\
    &=\frac{\delta}{\sfR}\sum_{\sfr=1}^\sfR\E\big[\partial_sf_t^u(\be^1,\dots,\be^{t-1},\bc_\sfr^1,\dots, \bc_\sfr^{L_c}, \sfr)\big].
\end{align*}
Next,  for $r,s \in [t]$, the $(r,s)$th element of the state evolution matrix $\Gamma^t \in \reals^{t \times t}$ is 
\begin{align*}
    \left(\Gamma^t \right)_{r,s}
    & =\E\big[f_r^v(\bh^1,\dots,\bh^t,\bd^1,\dots,\bd^{L_d+1})f_s^v(\bh^1,\dots,\bh^t,\bd^1,\dots,\bd^{L_d+1})\big] \\
    &=\E\Big[\E\big[f_r^v(\bh^1,\dots,\bh^t,\bd^1,\dots,\bd^{L_d+1})f_s^v(\bh^1,\dots,\bh^t,\bd^1,\dots,\bd^{L_d+1})\,\big|\,\bd^{L_d+1}\big]\Big] \\
    &=\frac{1}{\sfC}\sum_{\sfc=1}^\sfC \E\big[f_r^v(\bh^1,\dots,\bh^t,\bd_\sfc^1,\dots,\bd_\sfc^{L_d}, \sfc)f_s^v(\bh^1,\dots,\bh^t,\bd_\sfc^1,\dots, \bd_\sfc^{L_d}, \sfc)\big] =\frac{1}{\sfC}\sum_{\sfc=1}^\sfC \E[\bv_\sfc^r\bv_\sfc^s],
\end{align*}
where $\bv_\sfc^s=f_s^v(\bh^1,\dots,\bh^t,\bd_\sfc^1,\dots, \bd_\sfc^{L_d}, \sfc)$. Similarly, we have
\begin{align*}
    \left( \Xi^t \right)_{r,s}
    &=\frac{\delta}{\sfR}\sum_{\sfr=1}^\sfR \E\big[f_r^u(\be^1,\dots,\be^{t-1},\bc^1,\dots,\bc^{L_c+1}=\sfr)f_s^u(\be^1,\dots,\be^{t-1},\bc^1,\dots,\bc^{L_c+1}=\sfr)\big] \\
    &=\frac{\delta}{\sfR}\sum_{\sfr=1}^\sfR \E[\bu_\sfr^r\bu_\sfr^s], 
\end{align*}
where for $s \ge 1$, we have $\bu_\sfr^{s+1}=f_{s+1}^u(\be^1,\dots,\be^{s},\bc_\sfr^1,\dots,\bc_\sfr^{L_c}, \sfr)$. 
We have shown that with the choice of side information in \eqref{eq:dL_cL_choice}, the AMP recursion in \eqref{eq:U-AMP-new} matches that in \eqref{eq:modified_U-AMP}, and the corresponding state evolution recursions also match. 
Applying Theorem \ref{thm:abs_AMP} and recalling the definitions of $\mathcal{R}, \mathcal{C}$ from \eqref{eq:S_c_and_S_r_def} gives us Corollary \ref{cor:U-AMP}.

\section{Proof of Lemma \ref{lem:QGT_pot_fun_vanish}}
\label{proof:lem:QGT_pot_fun_vanish}
For $\delta >0$, evaluating $U(b_0;\delta)$ in \eqref{eq:U_function} with  $b_0 :=\delta\sigma^2$ gives
\begin{align}
    U(b_0;\delta)
    &=-\frac{\delta}{2}+\delta\log2+2I\left(\bar{\beta},\sqrt{\frac{1}{2\sigma^2}}\bar{\beta}+G\right).
\end{align}
We pause to state the following auxiliary result.

\begin{lemma} \label{lem:mutual_info_bound}
\textup{\cite[Proposition 7.15]{Don13}}
For a discrete distribution $P_{\bar{\beta}}$ with finite alphabet, we have
\begin{align}
    \limsup_{s\rightarrow\infty} \, 
    \frac{I(\bar{\beta};\sqrt{s}\bar{\beta}+G)}{\frac{1}{2}\log s}=0.
\end{align}
\end{lemma}
Lemma \ref{lem:mutual_info_bound} implies that for any $\Delta>0$,  we have
$    I(\bar{\beta};\sqrt{s}\bar{\beta}+G)
    \leq\frac{\Delta}{2}\log s$ for all sufficiently large $s$. Taking $s=\frac{1}{2\sigma^2}$ further implies that for sufficiently small $\sigma$, we have
\begin{align}
    &I\left(\bar{\beta};\sqrt{\frac{1}{2\sigma^2}}\bar{\beta}+G\right)
    \leq\frac{\Delta}{2}\log\left(\frac{1}{2\sigma^2}\right)
    \nonumber \\
    \implies
    &U(b_0;\delta)
    \leq-\frac{\delta}{2}+\delta\log2+\Delta\log\left(\frac{1}{2\sigma^2}\right).
    \label{eq:pot_fn_upper_bound}
\end{align}
Hence, for any $\Delta>0$, there exists $\sigma_0(\Delta) >0$ such that for all $\sigma < \sigma_0(\Delta)$ we have the following for all
$b\in(0,\Var(\bar{\beta})]$ and  $\delta>0$: 
\begin{align}
    U(b;\delta)-U(b_0;\delta)
    &\stackrel{(a)}{\geq}
    -\delta\left(1-\frac{\delta\sigma^2}{b+\delta\sigma^2}\right)+\delta\log\left(1+\frac{b}{\delta\sigma^2}\right)+2I\left(\bar{\beta};\sqrt{\frac{\delta}{b+\delta\sigma^2}}\bar{\beta}+G\right) \nonumber \\
    &\qquad\qquad+\frac{\delta}{2}-\delta\log2-\Delta\log\left(\frac{1}{2\sigma^2}\right) \nonumber \\
    &\stackrel{(b)}{\geq}
    \delta\log\left(\frac{b}{\delta\sigma^2}\right)-\Delta\log\left(\frac{1}{2\sigma^2}\right)-\frac{\delta}{2}-\delta\log2,
    \label{eq:pot_fn_diff}
\end{align}
where (a) uses \eqref{eq:U_function} and \eqref{eq:pot_fn_upper_bound}, and (b) uses the non-negativity of mutual information. 
Now, for $\sigma <  \sigma_0 (\Delta)$ and 
\begin{equation}
    b > 2\delta e^{\frac{1}{2}}\left(\frac{1}{2}\right)^{\frac{\Delta}{\delta}}\sigma^{2-\frac{2\Delta}{\delta}},
    \label{eq:bLB}
\end{equation}
the lower bound in \eqref{eq:pot_fn_diff} is strictly positive. Therefore,  
$
U(b;\delta)-U(b_0;\delta)>0
$ for  $b$ satisfying \eqref{eq:bLB},
implying that these values of $b$ cannot be minimizers of $U(b;\delta)$. Therefore, for any  $\delta, \Delta>0$ and $\sigma < \sigma_0(\Delta)$, we have:
\begin{align*}
     \max\bigg\{\argmin_{b\in(0,\Var(\bar{\beta})]}U(b;\delta)\bigg\}
    \leq 2\delta e^{\frac{1}{2}}\left(\frac{1}{2}\right)^{\frac{\Delta}{\delta}}\sigma^{2-\frac{2\Delta}{\delta}}
    < \frac{7}{2} \delta (\sigma^{2-\frac{2\Delta}{\delta}}). 
\end{align*}
\qed

\section{Implementation Details} \label{sec:imp_details} 
\paragraph{SC-AMP denoiser and state evolution parameters for QGT.}
The  Bayes-optimal denoiser $f_{k+1}$ in  \eqref{eq:optimal_fk} can be computed using the prior $\bar{\beta} \sim \text{Bernoulli}(\pi)$. For $j \in \mathcal{J}_\sfc$ and $k \ge 1$, we have
\begin{align}
    f_{k}(s,\sfc)
    &=\E\big[\bar{\beta}\,\big|\,(\chi_{\sfc}^{k})^2\bar{\beta}+\chi_{\sfc}^{k}G=s\big] 
     = \frac{\mb{P}\left[\bar{\beta}=1\right]\cdot\mb{P}\left[(\chi_{\sfc}^{k})^2\bar{\beta}+\chi_{\sfc}^{k}G=s|\bar{\beta}=1\right]}{\sum_{\bar{\beta}\in\{0,1\}}\mb{P}[\bar{\beta}]\cdot\mb{P}[(\chi_{\sfc}^{k})^2\bar{\beta}+\chi_{\sfc}^{k}G=s|\bar{\beta}]} \nonumber \\
    & = \frac{\pi\phi\big( (s-(\chi_\sfc^{k})^2)/\chi_{\sfc}^{k}\big)}{\pi\phi\big( (s-(\chi_\sfc^{k})^2)/\chi_{\sfc}^{k}\big)+(1-\pi)\phi(s/\chi^{k}_\sfc)}, \label{eq:f_k_bayes}
\end{align}
where  $\phi(x)$ is the standard normal density. Instead of precomputing the state evolution parameters $\big( \chi^{k}_\sfc \big)$, they can  be estimated from the SC-AMP iterates as:
\begin{equation}
    \left(\widehat{\chi}_\sfc^{k}\right)^2= \sum_{\sfr=1}^{\sfR}{ \widetilde{W}_{\sfr\sfc} \left(\norm{\tTheta_\sfr^{k} }_2^2\right)^{-1}}\qquad \text{ for }\sfc \in [\sfC],
\end{equation}
where $\tTheta_\sfr^{k}=(\tTheta_i^{k})_{i \in \mathcal{I}_\sfr}$ is the restriction of $\tTheta^{k}$ to indices $\mathcal{I}_\sfr$. The derivative $\partial_1 f_{k}(\beta_j^{k},\sfc)$, required for $b^{k}$ in \eqref{eq:SC_AMP}, can be obtained by applying the Quotient rule to the last expression in \eqref{eq:f_k_bayes}.
\paragraph{SC-AMP denoiser and state evolution parameters for pooled data.}
The Bayes-optimal denoiser $f_{k}: \reals^L \times [\sfC] \to \reals^L$ in the SC-AMP algorithm in \eqref{eq:matrix_SC_AMP} is computed as follows:
\begin{align}
    f_k(s, \sfc)
    &=\E[\bar{B} \mid \bar{B}+G_\sfc^k = s]
    =\sum_{l=1}^Le_l\frac{\mb{P}[\bar{B}=e_l]\mb{P}[ \bar{B}+G_\sfc^k = s|\bar{B}=e_l]}{\mb{P}[ \bar{B}+G_\sfc^k = s]} \nonumber \\
    &
    \overset{(a)}{=}\frac{\sum_{l=1}^L\pi_l e_l \, \exp \left( -\frac{1}{2}(e_l - s)^\sT \left(\Tau_\sfc^{k}\right)^{-1} (e_l - s)\right)}{\sum_{l=1}^L\pi_l \, \exp \left( -\frac{1}{2}(e_l - s)^\sT \left(\Tau_\sfc^{k}\right)^{-1} (e_l - s)\right)},
    \label{eq:Bayes_AMP_fk}
\end{align}
where (a) uses $G_\sfc^k \sim \normal\left(0, \Tau_\sfc^{k}\right)$. The state evolution parameters $\{ \phi_\sfr^k \}_{\sfr \in [\sfR]}$  are estimated from the matrix SC-AMP iterates as follows:
\begin{equation}
    \hat{\phi}_\sfr^k = \frac{1}{n/\sfR}\sum_{i\in\mathcal{I}_{\sfr}} (\tTheta_i^k)^\top {\tTheta_i^k} \, ,\qquad \text{ for }\sfr \in [\sfR].
\end{equation}
The Jacobian $f^\prime_k\left(B^k_{j,:}, \sfc\right)$ in \eqref{eq:U_k} can be computed for all $j\in [p]$ by applying the Quotient rule to \eqref{eq:Bayes_AMP_fk}, following the method in \cite[App. D.1]{Tan23d}.

\paragraph{Potential function.}
To generate the curves in Figure \ref{fig:potential_fn}, for each $\delta$,  the potential function $U(b;\delta)$ in \eqref{eq:U_function} is evaluated at 500 data points between 0 and $\text{Var}(\bar{\beta})$ (for $\pi=0.1$, $\text{Var}(\bar{\beta})=\pi - \pi^2=0.09$). To analyze the noiseless QGT model, we set $\sigma=1\times 10^{-30}$ to avoid computational instability. The mutual information term in \eqref{eq:U_function} is computed via numerical integration (instead of Monte Carlo methods) to ensure that the curves are smooth.

{\small{
\bibliographystyle{IEEEtran}
\bibliography{NT_References}
}}

\end{document}